\documentclass[12pt,a4paper]{article}

\usepackage{mathtools}

\DeclarePairedDelimiter\floor{\lfloor}{\rfloor}

\usepackage[utf8]{inputenc}		
\usepackage[T1]{fontenc}
\usepackage{multirow}
\usepackage{graphicx}
\usepackage{natbib}
\usepackage[dvipsnames]{xcolor}
\usepackage{lscape}
 
\usepackage{booktabs}
\usepackage{multirow}
\usepackage{graphicx}
\usepackage{enumerate}

\newcommand{\biblist}{\begin{list}{}
{\listparindent 0.0cm \leftmargin 0.50cm \itemindent -0.50 cm
\labelwidth 0 cm \labelsep 0.50 cm
\usecounter{list}}\clubpenalty4000\widowpenalty4000}
\newcommand{\ebiblist}{\end{list}}

\usepackage{amsmath,amsthm,amsfonts,amssymb,amscd,amsthm,xspace, dsfont}
\theoremstyle{plain}
\allowdisplaybreaks

\usepackage{latexsym}
\usepackage{amsmath}
\usepackage{amsfonts,amssymb}
\usepackage{multirow}
\usepackage{enumerate}

\newtheorem*{theorem*}{Theorem}

\newtheorem{remark}{Remark}[section]
\newtheorem{proposition}{Proposition}[section]
\newtheorem*{proposition*}{Proposition}
\newtheorem{result}{Result}[section]
\newtheorem*{result*}{Result}

\newtheorem*{lemma*}{Lemma}

\newcommand{\bx}{\mathbf{x}}

\newcommand{\E}{{\mathbb E}}
\newcommand{\V}{{\mathbb V}}

\definecolor{bazaar}{rgb}{0.6, 0.47, 0.48}

\usepackage{hyperref}
\hypersetup{
	colorlinks=true,
	linkcolor=blue,
	filecolor=magenta,
	urlcolor=blue,
	pdftitle={Sharelatex Example},
	bookmarks=true,
	citecolor=blue,
	}

\newcommand{\chapternote}[1]{{%
		\let\thempfn\relax
		\footnotetext[0]{\emph{#1}}
}}

\usepackage[margin=2.8cm]{geometry}

\newcounter{hypH}
\newenvironment{hypH}[1]
		{
		{\vskip 2mm\noindent\refstepcounter{hypH}\textbf{(H\thehypH)}\quad #1}
		\vskip 2mm 
		}

\newcounter{hypC}
\newenvironment{hypC}[1]
		{
		{\vskip 2mm\noindent\refstepcounter{hypC}\textbf{(C\thehypC)}\quad #1}
		\vskip 2mm 
		}

\begin{document}

\baselineskip .3in

\title{\bf {\Large Model-assisted estimation through random forests in finite population sampling}}

\author{Mehdi {\sc Dagdoug}$^{(a)}$,  Camelia {\sc Goga}$^{(a)}$ and  David  {\sc Haziza}$^{(b)}$ \\
	(a) Universit\'e de Bourgogne Franche-Comt\'e,\\ Laboratoire de Math\'ematiques de Besan\c con,  Besan\c con, FRANCE \\
	(b) University of Ottawa, Department of mathematics and statistics,\\ Ottawa, CANADA
}

\date{\today}
\maketitle
\chapternote{Mehdi Dagdoug's research was supported by grants of the region of Franche-Comt\'e and\\ M\'ediam\'etrie.}
\begin{abstract}
In surveys, the interest lies in estimating finite population parameters such as population totals and means. In most surveys, some auxiliary information is available at the estimation stage. This information may be incorporated in the estimation procedures to increase their precision.  In this article, we use random forests to estimate 
the functional relationship between the survey variable and the auxiliary variables. In recent years, random forests have become attractive as National Statistical Offices have now access to a variety of data sources, potentially exhibiting a large number of observations on a large number of variables.  We establish the theoretical properties of model-assisted procedures based on random forests and derive corresponding variance estimators. A model-calibration procedure for handling multiple survey variables is also discussed.  The results of a simulation study suggest that the proposed point and estimation procedures perform well in term of bias, efficiency and coverage of normal-based confidence intervals, in a wide variety of settings. Finally, we apply the proposed methods using data on radio audiences collected by  M\'ediam\'etrie, a French audience company.
\end{abstract}


\section{Introduction}

Since the pioneering work of \cite{sarndal_1980}, \cite{robinson_sarndal_1983} and \cite{sarndal_wright_1984}, model-assisted estimation procedures have attracted a lot of attention in the literature; see also \cite{SarndalLivre} for a comprehensive discussion of the model-assisted approach. At the estimation stage, auxiliary information is often available and can be incorporated in the estimation procedures to increase the precision of the resulting point estimators. The model-assisted approach starts with postulating a working model, describing the relationship between a survey variable $Y$ and a set of $p$ auxiliary variables $X_1, X_2,  \ldots, X_p$. The model is fitted to the sample observations to obtain predicted values, which then serve to build point estimators of population means/totals. Model-assisted estimators are asymptotically design-unbiased and design consistent, irrespective of whether or not the working model is correctly specified, which is an attractive feature; see \cite{SarndalLivre} and \cite{breidt_opsomer_2017}, among others. When the working model holds, model-assisted estimators are expected to be highly efficient. However, when the sample size is small, the use of model-assisted estimators requires some caution as they may suffer from small sample bias. In this article, we use random forests to estimate 
the functional relationship between $Y$ and   $X_1, X_2,  \ldots, X_p.$  In recent years, random forests have become attractive as National Statistical Offices have now access to a variety of data sources, potentially exhibiting a large number of observations on a large number of variables.

Consider a finite population $U = \left\{ 1, ..., k, ..., N \right\}$ of size $N$. We are interested in estimating the population total of a survey variable $Y$, $t_y = \sum_{k \in U} y_k$. We select a sample $S,$ of size $n,$ according to a sampling design $\mathcal{P}(S\mid \mathbf{Z}_U),$ where $\mathbf{Z}_U$ denotes the matrix of design information, available prior to sampling for all the population units. Let $\mathbf{I}_U=(I_1, ..., I_k, ..., I_N)^{\top}$ be the $N$-vector of sample selection indicators such that $I_k=1$ if $k \in S$ and $I_k=0,$ otherwise. The first-order and second-order inclusion probabilities are given by $\pi_k=\mathbb{E}\left[I_k \mid \mathbf{Z}_U\right]$
and $\pi_{kl}=\mathbb{E}\left[I_kI_\ell \mid \mathbf{Z}_U \right],$ respectively. 

A basic estimator of $t_y$ is the well-known Horvitz-Thompson estimator given by 
\begin{equation}\label{ht}
\widehat{t}_{\pi} = \sum_{k \in S} \dfrac{y_k}{\pi_k}.
\end{equation}
Provided that $\pi_k>0$ for all $k \in U$, the estimator (\ref{ht}) is design-unbiased for $t_y$ in the sense $\mathbb{E} \left[\widehat{t}_{y, \pi} \mid \mathbf{y}_U, \mathbf{Z}_U\right]=t_y$, where $\mathbf{y}_U = \left(y_1, y_2, ..., y_N\right)^{\top}$. The Horvitz-Thompson estimator makes no use of auxiliary information beyond what is already contained in the matrix $\mathbf{Z}_U$.

We assume that a  vector $\mathbf{x}_k=(x_{k1}, x_{k2}, \ldots, x_{kp})^{\top}$ of auxiliary variables is available for all $k \in U$. We also assume  that $y_k$, $k \in U,$ are independent realizations from a working model $\xi$, often referred to as a superpopulation model:
\begin{align}\label{mo}
\mathbb{E}\left[y_k\mid \mathbf{X}_k = \mathbf{x}_k\right]  &= m(\mathbf{x}_k), \\
\mathbb{V} (y_k \mid  \mathbf{X}_k = \mathbf{x}_k) &= \sigma^2\nu(\mathbf{x}_k), \nonumber
\end{align}
where $m (\cdot)$ and $\nu(\cdot)$ are two unknown functions and $\sigma^2$ is an unknown parameter. 

Suppose that Model (\ref{mo}) is fitted at the population level and let $\widetilde{m}(\mathbf{x}_k)$ be the population-level fit associated with unit $k$ obtained by fitting a parametric or nonparametric procedure. This leads to the pseudo generalized difference estimator
\begin{equation} \label{dif}
\widehat{t}_{pgd} = \sum_{ k \in U} \widetilde{m} (\mathbf{x}_k) + \sum_{k \in S} \dfrac{y_k - \widetilde{m} (\mathbf{x}_k)}{\pi_k}.
\end{equation}
Because the values $\widetilde{m}(\mathbf{x}_k)$ do not involve the sample selection indicators $I_1, \ldots, I_N$, if follows that $\mathbb{E}\left[\widehat{t}_{pgd} \mid \mathbf{y}_U, \mathbf{Z}_U, \mathbf{X}_U\right]=t_y,$ where $\mathbf{X}_U$ is the $N \times p$ matrix whose $N$ rows are the vectors $\mathbf{x}_1, \ldots, \mathbf{x}_N$. That is, the pseudo generalized difference estimator (\ref{dif}) is design-unbiased for $t_y$. In the sequel, we use the simpler notation $\E_p \left[ \cdot \right]$ instead of $\E \left[ \cdot \rvert \mathbf{Z}_U, \mathbf{X}_U, \mathbf{y}_U\right]$ to denote the expectation operator with respect to the sampling design $\mathcal{P} (S \rvert \mathbf{Z}_U)$. Similarly, the notation $\V_p \left[ \cdot \right]$ is used to denote the design variance of an estimator.  

Most often, the estimator (\ref{dif}) is unfeasible as the population-level fits  $\widetilde{m}(\mathbf{x}_k)$ are unknown. Using the sample observations,  we fit the working model and obtain the sample-level fits $\widehat{m} (\mathbf{x}_k).$ Replacing $\widetilde{m}(\mathbf{x}_k)$ with $\widehat{m} (\mathbf{x}_k)$ in (\ref{dif}), we obtain the so-called model-assisted estimator of $t_y$:
\begin{equation}\label{ma}
\widehat{t}_{ma} = \sum_{ k \in U} \widehat{m} (\mathbf{x}_k) + \sum_{k \in S} \dfrac{y_k - \widehat{m} (\mathbf{x}_k)}{\pi_k}.
\end{equation}
Unlike (\ref{dif}), the estimator (\ref{ma}) is no longer design-unbiased, but can be shown to be design-consistent for $t_y$ for a relatively wide class of procedures $\widehat{m} (\cdot)$. The model-assisted estimator (\ref{ma}) is expressed as the sum of the population total of the predictions $\widehat{m} (\mathbf{x}_k) $  and an adjustment term that can be viewed as a protection against model-misspecification.

If $\widehat{m}(\mathbf{x}_k) =\mathbf{x}_k^{\top} \widehat{\boldsymbol{\beta}}$ with coefficients estimated by weighted least squares, the estimator (\ref{ma}) reduces to the well-known generalized regression (GREG) estimator; e.g., see \citet[Chap. 6]{SarndalLivre}. Model-assisted estimators based on generalized linear models were considered by \cite{lehtonen1998logistic} and \cite{firth1998robust}, among others. There are some practical issues associated with the use of a parametric model such as linear and generalized linear models: they may lead to inefficient  estimators if the function $m(\cdot)$ is misspecified or if the model fails to include interactions or predictors that account for curvature (e.g., quadratic and cubic terms). In contrast, nonparametric procedures are robust to model misspeficiation, which is a desirable property. A number of nonparametric model-assisted estimation procedures have been studied in the last two decades: local polynomial regression \citep{breidt_opsomer_2000}, B-splines \citep{goga_2005} and penalized B-splines \citep{goga_ruiz-gazen}, penalized splines \citep{breidt_claeskens_opsomer_2005, mcconville_breidt_2013}, neural nets \citep{montanari_ranalli_2005}, generalized additive models \citep{opsomer2007model}, nonparametric additive models \citep{wang_wang_2011} and regression trees \citep{toth2011building, mcconville2019automated}. 


	In this paper, we propose a new class of model-assisted estimators of $t_y$ based on random forests (RF).
	Generally speaking, RF is an ensemble method that trains a (large) number of trees and combines them to produce more accurate predictions than a single regression tree would.
%
	Trees define a class of algorithms that recursively split the $p$-dimensional predictor space into distinct and non-overlapping regions. In other words, a tree algorithm  generates a partition of regions or hyperrectangles of $\mathbb R^p$. For an observation belonging to a given region, the prediction is simply obtained by averaging the $y$-values associated with the units belonging to the same region. While regression trees are easy to interpret and allow the user to visualize the partition \citep [pp. 306]{hastie_tibshirani_friedman_2011}, they may suffer from a high model variance, hence their qualification of "weak learners".  A number of tree-based procedures have been proposed with the aim of improving the predictive performances of regression trees, including pruning \citep{Breiman_Friedman_Olshen_Stone1984}, Bayesian regression trees \citep{Chipman1998}, gradient boosting \citep{friedman_2001} and RF \citep{breiman2001random}. 	


Several empirical studies suggest that RF can outperform state-of-the-art prediction models; see e.g. \cite{han2018comparison}, \cite{hamza2005empirical}, \cite{Diaz-Uriarte2006}. RF are widely used due to their predictive performances and their ability to handle small sample sizes with a large number of predictors \citep{scornet_2016_IEEE}. Also, RF algorithms can be parallelized, leading to a decrease in the training time.    RF have been applied in a wide variety of fields, including medicine \citep{Fraiwan2012}, time series analysis \citep{Kane2014}, agriculture \citep{Grimm2008}, missing data \citep{Stekhoven2011}, genomics \citep{qi_2012} and pattern recognition \citep{rogez_al_2008}.
In recent years, neural networks and deep learning algorithms have attracted a lot of attention and have been shown to be effective in a wide range of applications involving mostly  unstructured data,  such as speech recognition, image reconstruction and text translation; see \cite{najafabadi2015deep} and the references therein for a review on the topic. However, to exhibit high levels of performance, deep learning algorithms typically require huge amounts of data \citep{najafabadi2015deep, arnould2020analyzing}.  This is seldom the case in surveys as most data sets consist of structured data consisting of (at most) a few
hundred thousand observations and a few hundred survey variables.  For an empirical comparison of RF and neural networks, see   \cite{han2018comparison}. Finally, unlike RF algorithms that require the specification of a small number of hyper-parameters (see Section \ref{hyper_param}), gradient boosting, Bayesian regression trees or deep learning approaches depend upon the complex choice of a large number of hyper-parameters \citep{bergstra2011algorithms}.

To the best of our knowledge, only little is known about the theoretical properties of RF  based on the original algorithm of \cite{breiman2001random}. Often, the theoretical investigations are made at the expense of simplifying assumptions; see for instance \cite{biau2008consistency} and \cite{biau2012analysis}. Two notable exceptions are \cite{wager2014asymptotic} and \cite{scornet2015consistency} who established the theoretical properties of an algorithm closely related to that of \cite{breiman2001random}. In a finite population setting, the theoretical properties of RF algorithms have yet to be established, even in the ideal situation of 100\% response. This paper aims to fill this important gap.  While we are mostly concerned with RF for regression, we can easily extend our methods to the case of RF for classification.   Some recent empirical studies on the performance of RF for complex survey data can be found in \cite{tipton2013properties}, \cite{buskirk2015finding}, \cite{de2018sample} and \cite{kern2019tree}.

The rest of the paper is organized as follows.  Regression trees and RF are presented in Section \ref{tree_rf_pop}. In Section \ref{section3}, we suggest two classes of model-assisted estimators based on random forests: the first is based on partitions built at the population level, while the second class is based on partitions built at the sample level. In Section \ref{section4}, we establish the theoretical properties of model-assisted estimators based on RF and derive corresponding variance estimators. In Section \ref{model_calib}, we describe a model-calibration procedure for handling multiple survey variables. In Sections \ref{simu1}-\ref{hyper_param}, the finite sample properties of the proposed point and variance estimation procedures  
are evaluated through a simulation study, and in Section \ref{real_data}, we apply the proposed methods using data on radio audiences collected by  M\'ediam\'etrie, a French audience company.  The paper ends with some final remarks in Section \ref{conclusion}. Proofs of major results and further technical details are relegated to the Appendix and the Supplementary Material.

\section{Regression trees and random forests}\label{tree_rf_pop}

\subsection{Regression trees}\label{regres_tree}
The original RF uses regression trees based on the classification and regression tree algorithm (CART) of \citet{Breiman_Friedman_Olshen_Stone1984}, whereby the partition of the predictor space is generated by a greedy recursive algorithm. In this paper, we focus on the CART algorithm for regression, designed for handling quantitative survey variables $Y$, but our methods 
also applies to the case of binary survey variables.
In this case, the interest lies in estimating the probability of the outcome being equal to one. With regression trees, these estimated probabilities always lie between $0$ and $1,$ which is a desirable feature. Alternative criteria may be used with binary variables, such as the Gini impurity or the entropy instead of the CART regression criterion \citep[Chapter 9]{hastie_tibshirani_friedman_2011}.  

The CART algorithm for regression searches for the splitting variable and the splitting position (i.e., the coordinates on the predictor space where to split) for which the difference in empirical variance in the node before and after splitting is maximized. As a starting point, we consider the hypothetical situation, where $y_k$ and $\mathbf{x}_k$ are observed  for all $k \in U$ and assume that the regression tree is fitted at the population level. We use the generic notation $A$  to denote a node with cardinality $\#(A)$ considered for the next split, and $\mathcal{C}_A$ to denote the set of possible splits in the node $A$, which corresponds to the set of all possible pairs $(j,z)=(\text{variable}, \text {position})$. This splitting process is performed by searching for the best split $(j^*,z^*)$ for which the following empirical CART population criterion is maximized: 
\begin{equation} \label{opt1}
	L_N(j,z) = \dfrac{1}{\#(A)} \sum_{k \in U} \mathds{1}_{ \mathbf{x}_k \in A} \left\{\left(y_k - \bar{y}_A\right)^2 -  \left(y_k -\bar{y}_{A_{L}} \mathds{1}_{ x_{kj} < z} -\bar{y}_{A_{R}}\mathds{1}_{ x_{kj} \geq z} \right)^2 \right\} ,
\end{equation}
where 
$A_{L} = \left\{ k \in A ; x_{jk} < z \right\}$, $A_{R} = \left\{ k\in A ; x_{jk} \geqslant z \right\}$
and $\bar{y}_A$ is the average of the $y$-values of units belonging to $A.$ The best cut is always performed in the middle of two consecutive data points. In practice, it is common to impose a minimal number of observations $N_0$ (say) in each terminal node. In this case, the splitting process is performed until an additional split generates a terminal node with fewer observations than $N_0$.


The splitting process leads to the set
\begin{equation}\label{partition}
	\mathcal{P}_U = \left\{A^{(U)}_1, \ldots, A^{(U)}_j, \ldots, A^{(U)}_{J_U}\right\}
\end{equation}  of $J_U$ hyperrectangles of $\mathbb{R}^p$ such that  $A_j^{(U)}\bigcap A_{j'}^{(U)} = \emptyset$, for all $j \neq j' \in \left\{1, 2, \ldots, J_U\right\}$ and $\displaystyle{\bigcup_{j=1}^{J_U} A^{(U)}_j} = \mathbb{R}^p$. Thus, the set $\mathcal{P}_U$ defines a partition of $\mathbb{R}^p$, whose elements are called the terminal nodes. We use the generic notation $A^{(U)}(\mathbf{x}_k)$ to denote a terminal node belonging to the partition $\mathcal P_U$ given in (\ref{partition}) and that contains $\bx_k.$

Figure 1 below illustrates how the recursive splitting procedure creates a partition in the simple case of two auxiliary variable,s $X_1$ and $X_2,$ based on 5 splits. Each grey rotated square represents a split (variable, position) performed at some position along one of the two auxiliary variables, $X_1$ or $X_2$. The white ellipses represent the $6$ terminal nodes, also represented by the scatter plot on the right; see also \cite{biau2014cellular} for a similar illustration.
\begin{figure}[h!]
	\centering
	\includegraphics[width=1.02\textwidth]{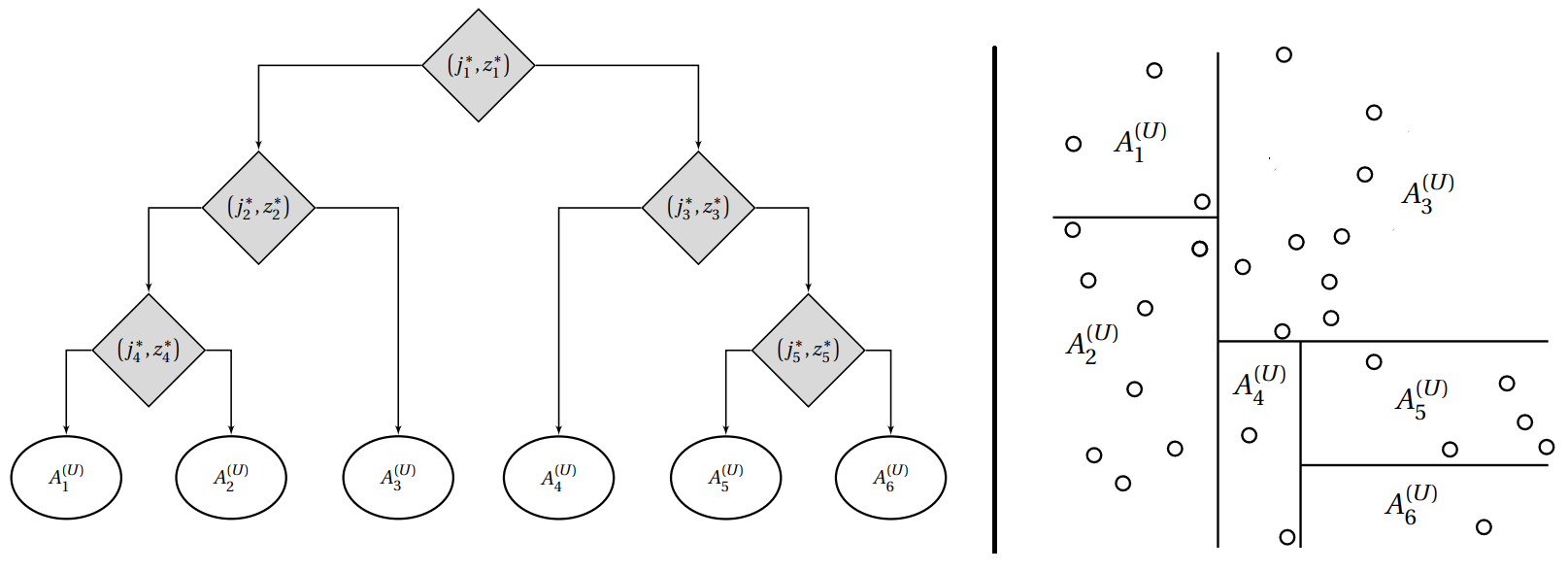}
	\caption{A regression tree (left) and the corresponding partition of $\mathbb{R}^2$ (right).}
\end{figure}


The prediction $\widetilde{m}_{tree}(\mathbf{x}_k)$ at the point $\bx_k$ is simply defined as the average of the  $y$-values of population individuals $\ell$ such that $\bx_\ell$ belongs to $A^{(U)}(\mathbf{x}_k)$: 
\begin{equation} \label{tree1}
\widetilde{m}_{tree}(\mathbf{x}_k) =\sum_{\ell \in U} \dfrac{\mathds{1}_{\bx_\ell \in  A^{(U)} \left(\mathbf{x}_k\right)}y_\ell}{\widetilde{N}(\bx_k)}, 
\end{equation}
where $\widetilde{N}(\bx_k) = \sum_{\ell \in U} \mathds{1}_{\bx_\ell\in A^{(U)}(\mathbf{x}_k)}$ denotes
 the number of units belonging to the terminal node $A^{(U)}(\mathbf{x}_k).$ 
 Given the partition $\mathcal{P}_U$, the population-level fit  $\widetilde{m}_{tree}(\mathbf{x}_k)$ may be viewed as the least squares type prediction obtained by fitting a one-way ANOVA model with $Y$ as the response variable and the node membership indicators $\{\mathds{1}_{\bx_k\in A_j^{(U)}}\}_{j=1}^{J_U}$ as the set of explanatory variables; see \citep[Chapter 9]{hastie_tibshirani_friedman_2011} and the Supplementary Material for more details. 

\subsection{Random forests}\label{rf}
To introduce RF in a finite population setting, we again assume that $y_k$ and $\mathbf{x}_k$ are observed  for all $k \in U$. 
RF are based on a (large) number $B$ (say) of regression trees. 
The prediction attached to unit $k$ is defined as the average of the predictions produced by each of the $B$ regression trees. That is, 
\begin{equation*}
\widetilde{m}_{rf}(\bx_k) = \dfrac{1}{B} \sum_{b = 1}^B \widetilde{m}^{(b)}_{tree}(\mathbf{x}_k),
\end{equation*}
where $\widetilde{m}^{(b)}_{tree}(\mathbf{x}_k)$ is the predicted value attached to unit $k$ obtained from the $b$th regression tree, $b=1, \ldots, B.$

Obviously, if $ \widetilde{m}^{(1)}_{tree}(\mathbf{x}_k) = \ldots =  \widetilde{m}^{(B)}_{tree}(\mathbf{x}_k)$, then $\widetilde{m}_{rf}(\bx_k) = \widetilde{m}^{(1)}_{tree}(\mathbf{x}_k) $. Such a situation would occur if each regression tree uses a deterministic
splitting criterion in (\ref{opt1}), which  would lead to $B$ identical partitions of $\mathbb{R}^p$. To cope with this issue, some amount of randomization is introduced in the tree building process,  leading to $B$ different predictions of $m(\cdot). $ 
The original algorithm of \cite{breiman2001random}  is implemented as follows:
\begin{enumerate}
	\item Select $B$ bootstrap data sets 
	with replacement from the population data set, $D_U = \left\{(\bx_k, y_k)\right\}_{k \in U}$, each data set containing $N$ pairs of the form $(\bx_k, y_k);$
	\item Fit a regression tree on each bootstrap data set. 
	Before each split is performed, $m_{try}$ predictors are selected randomly and without replacement from the full set of $p$ predictors. The $m_{try}$ selected predictors are the split candidates to be considered for searching the best split in (\ref{opt1}).
\end{enumerate} 
The algorithm stops when each terminal node contains less than a predetermined number of observations. This procedure leads to a set $\widetilde{\mathcal{P}}_U = \left\{\mathcal{P}_U^{(1)}, \mathcal{P}_U^{(2)}, \ldots , \mathcal{P}_U^{(B)}\right\}$ of $B$ different partitions of $\mathbb{R}^p$, each of the form  (\ref{partition}). The randomization used in the tree building process is denoted by the random variable $\theta^{(U)},$ assumed to  belong to some measurable space $(\Theta, \mathcal{F})$ and independent of the data \citep{biau2016random}. Let $\theta^{(U)}_b$ be the random variable associated with the $b$th tree. The random variables $\theta_b^{(U)}, b=1, \ldots, B,$ are assumed to be independent and their distribution is identical to that of the generic random variable $\theta^{(U)}$. In the RF algorithm of Breiman, the randomization is induced by the
selection (with replacement) of observations in Step 1 of the above algorithm and the random selection of split variables in Step 2 of the above algorithm.
A number of RF algorithms algorithms have been considered in the literature. 
%
 For example, \citep{biau2008consistency, scornet2016asymptotics}  considered a simple RF algorithm called the uniform random forest (URF) algorithm. In the URF algorithm,  a variable is selected with equal probability among the initial $p$ predictors at each node and a split position is chosen uniformly in the node along the direction of the selected variable. The algorithm stops when each terminal node has a predetermined number of cuts. In this case, the randomization $\theta_b^{(U)}$ is characterized by the random selections of the 
node,  the split variable and the location. For more details on RF algorithms, the reader is referred to \cite{geurts2006extremely}, \cite{biau2008consistency}, \cite{biau2012analysis}, \cite{genuer2012variance}, \cite{scornet2016asymptotics},  among others. 
In the sequel, unless stated otherwise, we assume that the observations in Step 1 of the above algorithm are selected without replacement \citep{scornet2017tuning}, which we will refer to as subsampling.
Also, for more generality, the splitting criterion is left unspecified.

 Let  $\widetilde{m}^{(1)}_{tree}(\cdot, \theta_1^{(U)}), \ldots, \widetilde{m}^{(B)}_{tree}(\cdot, \theta_B^{(U)}),$ denote the predictions obtained with the $B$ stochastic or randomized regression trees. The RF prediction attached to unit $k$ is defined as a bagged estimator of $B$ trees:
\begin{equation} \label{rf1}
\widetilde{m}_{rf}(\bx_k) = \dfrac{1}{B} \sum_{b = 1}^B \widetilde{m}^{(b)}_{tree}(\mathbf{x}_k, \theta_b^{(U)}).
\end{equation}
It is worth pointing out that considering a new set of predictors at each split leads to  $B$  trees which are less correlated with each other; that is, trees that are quite different from one another. As a result, the RF may lead to substantial gains in precision compared to a single tree  \citep [Chapter 8]{hastie_tibshirani_2015}. The number of predictors selected at each split, denoted by $m_{try},$ is thus an important tuning parameter in the RF algorithm. In practice, the choice $m_{try}=\sqrt{p}$ seems to give good results, in general. In Section \ref{hyper_param}, we assess the impact of $m_{try}$ through  a simulation study.

For any RF algorithm, the prediction at the point $\textbf{x}_k$  in (\ref{rf1}) can also be expressed as
\begin{equation} \label{localAve0}
\widetilde{m}_{rf} (\bx_k) = \sum_{ \ell \in U} \widetilde{W}_{\ell} (\bx_k) y_\ell,
\end{equation}
where
\begin{equation} \label{popWeights0}
\widetilde{W}_{\ell} (\bx_k ) = \dfrac{1}{B} \sum_{b=1}^B \dfrac{\psi_\ell^{(b, U)} \mathds{1}_{\bx_\ell\in A^{(U)}\left(\boldsymbol{\mathbf{x}_k} , \theta_b^{(U)}\right)} }{\widetilde{N}(\bx_k, \theta_b^{(U)})  }
\end{equation}
is a prediction weight attached to unit $k$ with $\widetilde{N}(\bx_k, \theta_b^{(U)}) = \sum_{ \ell \in U}\psi_\ell^{(b, U)} \mathds{1}_{ \bx_\ell\in A^{(U)}(\boldsymbol{\mathbf{x}_k} , \theta_b^{(U)})}$ denoting the number of observations belonging to the terminal node $A^{(U)}$ containing $\bx_k$ in the $b$th regression tree. Note that the prediction $\widetilde{m}_{rf}$ in (\ref{localAve0}) can be computed for either a continuous or a categorical  $y$-variable. In the latter case, the prediction $\widetilde{m}_{rf}$ in (\ref{localAve0}) corresponds to the population proportion of units who belong to a given category computed over the $B$ trees. The random variables $\psi_\ell^{(b, U)}$ in (\ref{popWeights0}) depend on the resampling mechanism used in the RF algorithm and depend on $\theta_b^{(U)}$, but are independent of the sampling design $\mathcal{P}(S\mid \mathbf{Z}_U)$. In the case of subsampling, the random variables $\psi_\ell^{(b, U)}$ follow a Bernoulli distribution, $\psi_\ell^{(b, U)} \sim \mathcal{B} \left({N'} /N\right)$, where ${N'}$ denotes the number of units in each subsample. 

\begin{proposition}\label{propr_weight_popRF} Consider the predictor weights $\widetilde{W}_{\ell} (\bx_k )$ given in (\ref{popWeights0}). 
\begin{enumerate}[i)]	
\item The weights $\widetilde{W}_{\ell} (\bx_k )$  are uniformly bounded. That is, 
\begin{equation*}
0 < \widetilde{W}_{\ell} (\bx_k ) \leqslant cN_0^{-1} 
\end{equation*}
for all  $\ell \in U$ and all $\bx_k \in \mathbb{R}^p$, where $c$ is a positive constant that does not depend either on $k$,$\ell,$ or $N_0,$ the minimal number of observations in the terminal nodes.
\item The weight functions sum up to one; that is, $\sum_{\ell \in U} \widetilde{W}_{\ell} (\bx_k ) = 1$ for all $\bx_k \in \mathbb{R}^p$.
\end{enumerate}
\end{proposition}
The proof of Proposition \ref{propr_weight_popRF} is given in the Appendix.

\section{Model-assisted estimation: Random forests}\label{section3}

In Section \ref{tree_rf_pop}, we assumed that $y_k$ and $\mathbf{x}_k$ were observed for all $k\in U$, which led to the population-level fits $\widetilde{m}_{tree}(\mathbf{x}_k)$ and $\widetilde{m}_{rf}(\mathbf{x}_k)$ given by (\ref{tree1}) and (\ref{localAve0}), respectively. However, both (\ref{tree1}) and (\ref{localAve0}) cannot be computed in practice as the $y$-values are observed only for $k\in S.$ Moreover, the regression trees in Sections \ref{regres_tree} and \ref{rf} were based on partitions built recursively at the population level  so as to optimize the population criterion (\ref{opt1}). As a result, these  partitions depend on the vector of predictors $\{\mathbf{x}_k\}_{ k\in U}$ but also on the unknown population values $\{y_k\}_{k \in U}.$
While the former type of dependency is inherent to most parametric and nonparametric procedures, the latter is absent in many commonly used parametric and nonparametric procedures such as spline procedures \citep{breidt_opsomer_2000, goga_2005, goga_ruiz-gazen, breidt_claeskens_opsomer_2005, mcconville_breidt_2013}. Due to the dependency on the unknown population values $\{y_k\}_{k \in U},$  establishing the theoretical properties of model-assisted estimators based on RF is more challenging. 

For these reasons, in Section \ref{Sec_ma}, we start by considering the simpler case of population partitions obtained using a variable $Y^*,$ assumed to be closely related to $Y$ and available for all $k \in U$. While this assumption is somehow strong and not tenable in many practical situations, it provides some insights on how to tackle the problem in the presence of $Y$-dependency.  Algorithms allowing to get rid of the $Y$-dependency have been suggested in the random-forest literature; see e.g. \cite{biau2008consistency}, \cite{biau2012analysis} or \citet [Chap. 20]{Devroye_Gyorfi_Lugosi}. Sample-based partitions are considered in Section \ref{sample-based_partitions}.

\subsection{Model-assisted estimation: Population-based partitions}\label{Sec_ma}
In this section, we consider the case of a splitting criterion that does not depend on the data $\{y_k\}_{k \in s}$. We consider a variable $Y^*$ assumed to be closely related to $Y$ and such that the values $y_k^*$ are available for all $k \in U$. We seek population partitions $\widetilde{\mathcal{P}}^*_U,$ independent of the survey variable $Y,$ that maximize  the following criterion:
\begin{eqnarray} \label{opt2}
L^*_N(j,z)=\small{\dfrac{1}{\#(A)} \sum_{k \in U} \mathds{1}_{ \mathbf{x}_k \in A} \left\{\left(y^*_k - \overline{y}^*_A\right)^2 -  \left(y^*_k -\overline{y}^*_{A_{L}} \mathds{1}_{ x_{kj} < z} -\overline{y}^*_{A_{R}}\mathds{1}_{ x_{kj} \geq z} \right)^2 \right\}},
\end{eqnarray}
where  $A_{R},$ $A_{L}$ are as in (\ref{opt1})  and $\overline{y^*}_A$ is the average of the $y^*$-values for the units belonging to a node $A.$  
Based on  (\ref{opt2}), the population-level fit at the point $\textbf{x}_k$ is given by
\begin{equation} \label{localAve_star}
\widetilde{m}^*_{rf} (\bx_k) = \sum_{ \ell \in U} \widetilde{W}_{\ell}^* (\bx_k) y_\ell,
\end{equation}
where the weights $\widetilde{W}_{\ell}^* (\bx_k)$ in (\ref{localAve_star}) are obtained from   (\ref{popWeights0})  by replacing $A^{(U)}$ with $A^{*(U)},$ a generic member of the partition $\widetilde{\mathcal{P}}^*_U.$ 

The weights $\{\widetilde{W}_{\ell}^* (\cdot)\}_{\ell\in U}$ in (\ref{localAve_star}) are known for all $\ell\in U$ and are independent of $Y$. Since $\widetilde{m}^*_{rf} (\bx_k)$ in (\ref{localAve_star}) requires the $y$-values for all the population units, it cannot be computed. A simple solution consists of replacing the population total on the right hand-side of (\ref{localAve_star}) by its corresponding Horvitz--Thompson estimator, which leads to
\begin{eqnarray} 
\widehat{m}_{rf}^*(\mathbf{x}_k)=\sum_{\ell\in S}\frac{\widetilde{W}^*_{\ell} (\bx_k ) y_\ell}{\pi_\ell}. \label{est_m_popRF}
\end{eqnarray}
A model-assisted estimator of $t_y$ based on population RF is obtained  by plugging $\widehat{m}_{rf}^*(\mathbf{x}_k)$ in  (\ref{dif}): 
\begin{eqnarray}\label{trf1}
\widehat{t}_{rf}^*=\sum_{k\in U}\widehat{m}_{rf}^*(\mathbf{x}_k)+\sum_{k\in S}\frac{y_k-\widehat{m}_{rf}^*(\mathbf{x}_k)}{\pi_k}.
\end{eqnarray}

\begin{proposition}\label{trf_poids_ech}
	The RF estimator given in (\ref{trf1}) can be expressed as 
	\begin{eqnarray*}  \label{p5_pop}
		\widehat{t}_{rf}^* = \sum_{k \in S} w_{ks} y_k,
	\end{eqnarray*}
	where the weights $w_{ks}$ are given by 
	\begin{equation}\label{weight_rf_pop}
	w_{ks} =\frac{1}{\pi_k}\left\{ 1+ \sum_{\ell \in U} \widetilde{W}_{k}^* \left(\mathbf{x}_{\ell}\right) \left(1 - \frac{I_\ell}{\pi_\ell}\right)\right\}, \quad k\in S
	\end{equation}
and	
$$
\sum_{k\in S}w_{ks}=1.
$$	
\end{proposition}

\begin{proof}
	\small{
		By rearranging the sums, we get:
		\begin{align*}
		\widehat{t}_{rf}^* &= \sum_{k\in S}\frac{y_k}{\pi_k}+\sum_{\ell\in U}\left(1-\frac{I_\ell}{\pi_\ell}\right)\widehat{m}_{rf}^*(\mathbf{x}_\ell)= \sum_{k\in S}\frac{y_k}{\pi_k}+\sum_{\ell\in U}\left(1-\frac{I_\ell}{\pi_\ell}\right)\left(\sum_{k\in S}\widetilde{W}_{k}^* \left(\mathbf{x}_\ell\right)\frac{y_k}{\pi_k}\right)\\
		&= \sum_{k\in S}\left\{1+\sum_{\ell\in U}\left(1-\frac{I_\ell}{\pi_\ell} \right)\widetilde{W}_{k}^* \left(\mathbf{x}_\ell\right)\right\}\frac{y_k}{\pi_k}.
		\end{align*}
}
\end{proof}
Since the partitions $\widetilde{\mathcal{P}}^*_U,$ are independent of both the survey variable $Y$ and the sample $S$, the weights $w_{ks}$ given by (\ref{weight_rf_pop}) depend on the sample only through the sample selection indicators $I_\ell, \ell \in U,$ but are independent of $Y$. As a result, these weights may be used to estimate the population total of any survey variable, which is an attractive feature in multipurpose surveys.  However,  for RF algorithms based on the splitting criterion in (\ref{opt2}), we expect the weights $w_{ks}$  to be efficient whenever the survey variable $Y$ is highly correlated to the variable $Y^*$. In multipurpose surveys where the survey variables are not necessarily correlated with one another, it may be preferable to use a splitting criterion that depends on the data $\{\bx_k\}_{k \in U}$ as done in quantile random forests \citep{Devroye_Gyorfi_Lugosi, scornet2016asymptotics}.


\subsection{Model-assisted estimation: Sample-based partitions}\label{sample-based_partitions}
In this section, we seek sample partitions $\widehat{\mathcal{P}}_S = \left\{\widehat{\mathcal{P}}_S^{(1)}, \ldots, \widehat{\mathcal{P}}_S^{(b)}, \ldots,  \widehat{\mathcal{P}}_S^{(B)}\right\}$  using 
the following sample-based criterion:
\begin{eqnarray}
L_n(j,z) &= &\dfrac{1}{\#(A)} \sum_{k \in S} \mathds{1}_{ \mathbf{x}_k \in A} \left\{\left(y_k - \bar{y}_A\right)^2 -  \left(y_k -\bar{y}_{A_{L}} \mathds{1}_{ x_{kj} < z} -\bar{y}_{A_{R}}\mathds{1}_{ x_{kj} \geq z} \right)^2 \right\}.
\label{opt3}
\end{eqnarray}
Based on the partition  $\widehat{\mathcal{P}}_S$, we obtain the sample-level fits 
\begin{equation} \label{samplePred}
\widehat{m}_{rf} (\bx_k) = \sum_{ \ell\in S} \frac{\widehat{W}_{\ell} (\bx_k) y_\ell}{\pi_\ell},
\end{equation}
where 
\begin{equation}
\widehat{W}_{\ell} (\bx_k ) = \dfrac{1}{B} \sum_{b=1}^B \dfrac{ \psi_{\ell}^{(b, S)} \mathds{1}_{ \bx_{\ell}\in A^{(S)} \left(\boldsymbol{\mathbf{x}_k} , \theta_b^{(S)} \right)}}{ \widehat{N}(\bx_k, \theta_b^{(S)}) },\quad \ell\in S,\label{echWeights}
\end{equation}
and $\widehat{N}(\bx_k, \theta_b^{(S)})  = \sum_{\ell\in U} I_\ell \pi_\ell^{-1} \psi_\ell^{(b, S)} \mathds{1}_{\bx_{\ell}\in  A^{(S)} \left(\mathbf{x}_{k} , \theta_b^{(S)}\right)}$ denotes the estimated number of observations in the terminal node $ A^{(S)}$ containing $\bx_k$ in the $b$th regression tree.
The variable $ \psi_\ell^{(b, S)}$ indicates whether or not unit $\ell$ has been selected in the $b$th sub-sample  and is such that 
$\psi_\ell^{(b, S)} \sim \mathcal{B} \left({n'}/n\right)$ for RF based on subsampling, where $n'$ denotes the number of units in each sub-sample. 
 
Plugging $\widehat{m}_{rf}(\cdot)$ in (\ref{ma}) leads to the RF model-assisted estimator \begin{equation}\label{trf_sample}
	\widehat{t}_{rf} = \sum_{ k \in U} \widehat{m}_{rf} (\mathbf{x}_k) + \sum_{k \in S} \dfrac{y_k - \widehat{m}_{rf} (\mathbf{x}_k)}{\pi_k}.
\end{equation}
Using similar arguments to those used in the proof of Proposition \ref{trf_poids_ech}, we can show that $\widehat{t}_{rf}$  can be expressed as 
\begin{equation*} \label{est_w_ver}
\widehat{t}_{rf} = \sum_{k\in S} w'_{ks} y_k,
\end{equation*}
where the weights $w'_{ks}$ are given by 
\begin{equation}\label{weight_rf_ech}
w'_{ks} =\frac{1}{\pi_k}\left\{ 1+ \sum_{\ell \in U} \widehat{W}_{k} \left(\mathbf{x}_\ell\right) \left(1 - \frac{I_\ell}{\pi_\ell}\right)\right\}, \quad k\in S.
\end{equation}

Noting that $\sum_{k \in S}\widehat{W}_{k} \left(\mathbf{x}_\ell\right)\pi_k^{-1}=1$ for all $\ell\in U$, it follows from (\ref{weight_rf_ech}) that $\sum_{k\in S} w'_{ks}=N$ for every sample $S$. That is, the sum of  the weights $w'_{ks}$ match the population size $N$  perfectly, a desirable property shared by other nonparametric model-assisted estimators \citep{goga_2005, goga_ruiz-gazen, breidt_claeskens_opsomer_2005}. 
Unlike the weights $w_{ks}$ in (\ref{weight_rf_pop}), the weights $w'_{ks}$  depend on both the sample selection indicators $I_\ell, \ell \in U,$ and the partition $\widehat{\mathcal{P}}_S$ that varies from one sample to another. This is due to the fact that the nodes $ A^{(S)}$ are  constructed so as to optimize the sample criterion (\ref{opt3}). For this reason, the weights $w'_{ks}, k\in S,$  are variable specific in the sense that depend on the survey variable $Y$. To cope with this issue, we describe a model calibration procedure in Section \ref{model_calib} for handling multiple survey variables while producing a single set of weights.


\begin{remark}
In practice,  the variables $\psi_k^{(b, S)}$ in (\ref{echWeights}) are not generated for the units outside the sample. However, at least conceptually, nothing  precludes defining these variables for $k \in U\setminus S$. For $k \in U\setminus S,$ we set $\psi_k^{(b, S)}\sim \mathcal{B} \left( (N' -n')/(N - n)\right)$ so that
 $\sum_{k \in U} \psi_k^{(b, S)}= N'.$  Defining the variables $\psi_k^{(b, S)}$ for units outside the sample will have no effect on the predictions  $\widehat{m}_{rf}(\cdot)$ associated with the sample units since $I_k = 0$ for  $k \in U \setminus S$. This construction will prove useful in establishing the asymptotic properties of the proposed procedures; see Section \ref{section4}. 
\end{remark}
%
%
%
%

 
As for the RF prediction built at the population level described in Section \ref{rf}, the prediction $\widehat{m}_{rf}(\mathbf{x}_k)$ in (\ref{samplePred}) can be expressed as a bagged predictor \citep{hastie_tibshirani_friedman_2011}. That is,
$$
\widehat{m}_{rf}(\mathbf{x}_k)=\dfrac{1}{B} \sum_{b = 1}^B \widehat{m}^{(b)}_{tree}(\mathbf{x}_k, \theta_b^{(S)}),
$$
where $\widehat{m}^{(b)}_{tree}(\mathbf{x}_k, \theta_b^{(S)})=\sum_{\ell\in S} \psi_{\ell}^{(b, S)} \mathds{1}_{ \bx_{\ell}\in A^{(S)} \left(\boldsymbol{\mathbf{x}_k} , \theta_b^{(S)} \right)}y_{\ell}/\widehat{N}(\bx_k, \theta_b^{(S)}) $ is the prediction  associated with unit $k$ based on the $b$th  stochastic regression tree.
The model-assisted estimator $\widehat{t}_{rf}$ given by (\ref{trf_sample}) can thus be viewed as a bagged estimator:
\begin{equation*}
\widehat{t}_{rf}= \dfrac{1}{B} \sum_{b = 1}^B \widehat{t}_{tree}^{(b)}( \theta_b^{(S)}),
\end{equation*}
where 
$$\widehat{t}_{tree}^{(b)}( \theta_b^{(S)})=\displaystyle \sum_{k\in U}\widehat{m}^{(b)}_{tree}(\mathbf{x}_k, \theta_b^{(S)})+\sum_{k\in S}\frac{y_k-\widehat{m}^{(b)}_{tree}(\mathbf{x}_k, \theta_b^{(S)})}{\pi_k}$$
is the model-assisted estimator of $t_y$ based on the $b$th  stochastic regression tree. As in the case of regression trees built at the population level (see Section \ref{regres_tree}), given the partition $\widehat{\mathcal P}^{(b)}_S=\{A_j^{(bS)}\}_{j=1}^{J_{bS}},$ 
the predictions $\widehat{m}^{(b)}_{tree}(\mathbf{x}_k, \theta_b^{(S)})$ are least squares type predictions obtained by fitting the one-way ANOVA model with $Y$ as the response and the node membership indicators $\{\mathds{1}_{\bx_k\in A_j^{(bS)}}\}_{j=1}^{J_{bS}}$ as the explanatory variables; see the proof of Proposition \ref{prop_projection} and the Supplementary Material for more details. As a result, the estimator $\widehat{t}_{tree}^{(b)}( \theta_b^{(S)})$ is related to the customary post-stratified estimator \citep{SarndalLivre}. 
%

Under mild assumptions, Proposition \ref{bagging} below shows that bagging improves the efficiency of model-assisted estimators. This is similar to what is encountered in the classical RF literature \citep{hastie_tibshirani_friedman_2011}.
\begin{proposition}\label{bagging}
	Let $\widehat{t}^{(1)}, \ldots, \widehat{t}^{(b)}, \ldots,  \widehat{t}^{(B)}$ be a sequence of model-assisted estimators of $t_y$ and let $\widehat{t} = B^{-1} \sum_{b=1}^B \widehat{t}^{(b)}$ be a bagged estimator. Assuming that the $\widehat{t}^{(b)}$'s have approximately the same design bias and design variance, then,  for $B$ large enough:
	\begin{equation*} \label{bound}
	\text{MSE}_p(\widehat{t}) - \text{MSE}_p(\widehat{t}^{(1)}) \leqslant \V_p(\widehat{t}^{(1)}) \left(   \max_{b \neq b'}\left|\mathbb{C}or_p \left(\widehat{t}^{(b)}, \widehat{t}^{(b')} \right) \right|- 1\right) \leqslant 0,
	\end{equation*}
	where $\text{MSE}_p(\cdot)$ and $\mathbb{C}or_p(\cdot)$  denote the mean squared error and correlation operators with respect to the sampling design. 
\end{proposition}
The proof of Proposition \ref{bagging} is given in the Appendix. We end this section by giving 
an alternative expression for $\widehat{t}_{rf}.$
\begin{proposition} \label{prop_projection}
The RF estimator $\widehat{t}_{rf}$  given by (\ref{trf_sample}) can be written as
\begin{align} \label{oob}
\widehat{t}_{rf} &= \sum_{ k \in U} \widehat{m}_{rf}(\bx_k) + \dfrac{1}{B} \sum_{b= 1}^B \sum_{k \in S} \dfrac{\left(1 - \psi_k^{(b, S)}\right)  \left(y_k - \widehat{m}^{(b)}_{tree}(\bx_k,\theta_b^{(S)})\right)}{\pi_k}, 
\end{align}
where $\widehat{m}^{(b)}_{tree}(\mathbf{x}_k, \theta_b^{(S)})$ is the predictor associated with unit $k$ based on the $b$th  stochastic regression tree. 
\end{proposition}
The proof of Proposition \ref{prop_projection}  is given in the Appendix. It follows from Proposition \ref{prop_projection}, that the second term on the right hand-side of (\ref{oob}) vanishes if (i) $\psi_k^{(b, S)}=1$ for all $k\in S.$ That is, the estimator $\widehat{t}_{rf}$  reduces to the so-called projection form \citep{SarndalLivre, breidt_claeskens_opsomer_2005, goga_2005}
\begin{equation*}
\widehat{t}_{rf} = \sum_{ k \in U} \widehat{m}_{rf}(\bx_k)
\end{equation*}
if the RF algorithm does not involve a resampling mechanism. In addition, the second term on the right hand-side of (\ref{oob}) vanishes if $y_k=c$ for all $k$, for some $c \in \mathbb{R}$ or  if the trees in the forest are fully grown (i.e., each terminal node contains a single observation), which implies that the observations $y_k$ and the corresponding prediction $\widehat{m}^{(b)}_{tree}(\bx_k,\theta_b^{(S)})$ coincide. When the estimator $\widehat{t}_{rf}$ can expressed in the projection form, the weights $w_{ks}'$ given by (\ref{weight_rf_ech}) are always positive and cannot exceed the number of terminal nodes from the largest tree of the forest.

In practice, a resampling mechanism is typically used with RF algorithms. In this case, the second term on the right hand-side of (\ref{oob}) does not vanish and is equal to the  weighted sum of residuals computed for the non-resampled units, also called the \textit{out-of-bag} individuals \citep [Chapter 8]{hastie_tibshirani_2015}, from each of the $B$ trees. The second term on the right hand-side of (\ref{oob}) can then be viewed as a correction term which brings additional information from the units not used in computing the predictions $\widehat{m}^{(b)}_{tree}(\cdot, \theta_b^{(S)}), b=1, \ldots, B.$

\section{Asymptotic properties}\label{section4}

To establish the asymptotic properties of the proposed estimators and to derive the associated variance estimators, we consider the asymptotic framework of \cite{isaky_fuller_1982}. We start with an increasing sequence of embedded finite populations $\{U_{v}\}_{v \in \mathbb{N}}$ of size $\{N_v\}_{v \in \mathbb{N}}$. In each finite population $U_v$, a sample  of size $n_v$ is selected according to a sampling design $\mathcal{P}_v(S_v=s_v\mid \mathbf{Z}_U)$. While the finite populations are assumed to be embedded, we do not require this property to hold for the samples $\{S_v\}_{v \in\mathbb N}$. This asymptotic framework assumes that $v$ goes to infinity, so that both the finite population sizes and the samples sizes go to infinity. 
To improve readability, we shall use the subscript $v$ only in the quantities $U_v, N_v$ and $n_v$; quantities such as $\pi_{k,v}$ shall be denoted simply as $\pi_k$. 

\subsubsection*{Assumptions: RF model-assisted estimator $\widehat{t}_{rf}^*$}
We make the following assumptions:

\hypH{{There exists a positive constant $C$ such that  $\sup_{k\in U_v}|y_k|\leqslant C<\infty.$} \label{H1}}

\hypH{We assume  that $ \lim\limits_{v \to \infty}\dfrac{n_v}{N_v} = \pi \in (0, 1)$.\label{H2}}

\hypH{There exist positive constants $\lambda$ and $\lambda^*$ such that $\min\limits_{k \in U_v} \pi_k \geqslant \lambda > 0$ and $\min\limits_{k,\ell \in U_v} \pi_{k\ell} \geqslant \lambda^*> 0.$ Also, we assume that $\limsup\limits_{v \to \infty} n_v \max\limits_{k \neq \ell \in U_v} | \pi_{k\ell} - \pi_k \pi_\ell| < \infty$.
\label{H3}}

Assumptions (H\ref{H1})-(H\ref{H3}) have been extensively used in parametric, nonparametric and functional model-assisted estimation    \citep{robinson_sarndal_1983, breidt_opsomer_2000, breidt_claeskens_opsomer_2005, goga_2005, goga_ruiz-gazen, CGL2012}. Assumption (H\ref{H1}) implies that the survey variable $Y$ is uniformly bounded \citep{breidt_opsomer_2000,cardot_chaouch_goga_labruere_2010}. Assumptions (H\ref{H2}) and (H\ref{H3}) deal with the first and second order inclusion probabilities and they are satisfied for the classical fixed-size sampling designs; see for example, \cite{robinson_sarndal_1983} and \cite{breidt_opsomer_2000}.  Furthermore, we assume that the minimum number of observations $N_{0v}$ in a terminal nodes is growing to infinity and we make the following additional assumption 
\hypC{The number of subsampled elements  $N_v'$ is such that $\lim_{v \to \infty} N'_v/N_v \in (0; 1]$.\label{C1}}
This assumption requires that the number $N'_v$ of elements in each subsample increases at the same speed as the population size $N_v$, allowing  each terminal node to have at least $N_{0v}$ observations.


\subsubsection*{Assumptions: RF model-assisted estimator $\widehat{t}_{rf}$}
In addition to the above assumptions, we make the following assumptions to establish the asymptotic properties of  $\hat t_{rf}$ given by (\ref{trf_sample}).

\hypH{ There exists a positive constant $C_1$ such that $n_v\max_{k\neq \ell\in U_v}\left|\E_p\left\{(I_k-\pi_k)(I_{\ell}-\pi_{\ell})|\widehat{\mathcal P}_S\right\}\right|\leqslant C_1.$\label{H4}}


\hypH{The random forests based on population partitions and those based on sample partitions are such that, for all $\bx \in \mathbb R^p:$ 
$$
\E_p\left(\widehat{\widetilde m}_{rf}(\bx)-\widetilde m_{rf}(\bx)\right)^2=o(1),
$$
where $\widehat{\widetilde m}_{rf}(\bx)= \sum_{ \ell\in U_v}\dfrac{1}{B} \sum_{b=1}^B \dfrac{ \psi_{\ell}^{(b, S)} \mathds{1}_{\bx_{\ell}\in A^{(S)} \left(\boldsymbol{\mathbf{x}} , \theta_b^{(S)} \right)}y_{\ell}}{ \widehat{\widetilde{N}}(\bx, \theta_b^{(S)}) }$
with
$\widehat{\widetilde{N}}(\bx, \theta_b^{(S)})  = \displaystyle\sum_{\ell\in U_v} \psi_\ell^{(b, S)} \mathds{1}_{\bx_{\ell}\in A^{(S)} \left(\boldsymbol{\mathbf{x}} , \theta_b^{(S)}\right)}$.

\label{H5}}

Assumption (H\ref{H4}) is similar to that used by \cite{toth2011building} and \cite{mcconville2019automated}; it requires that, as the sample and population size grow, the influence of extreme observations on the sample partitions decreases. Assumption (H\ref{H5}) requires that the average number of elements at the population level in the sample partitions converges to the average number of population elements in the population partitions. It implicitly assumes that the sample partitions converge to the population partitions. A similar result was established in \cite{toth2011building} in the case of regression trees.  \cite{toth2011building} evaluated the properties of point estimators with respect to the joint distribution induced by the superpopulation model and the sampling design. In a \textit{iid} setting, \cite{scornet2015consistency} showed that the population partitions converge to the theoretical partitions. Assumption (H\ref{H5}) can thus be viewed as a design-based version of the result from \cite{scornet2015consistency}. In the Supplementary Material, we conduct a simulation study, whose results suggest that Assumption  (H\ref{H5}) seems to be verified, at least in our experiments. More research is needed to provide a rigorous proof of Assumption (H\ref{H5}) in the design-based approach and is beyond the scope of this article.   

%

As in  the case of model-assisted estimators based on RF with population-based partitions, we assume that the minimum number of observations, $n_{0v},$ in the terminal nodes is also growing to infinity and we assume the following additional assumption about  the RF resampling algorithm :
\hypC{The number of subsampled elements  $n'_v$ is such that $\lim_{v \to \infty} n'_v/n_v \in (0; 1]$.\label{C1*}}
This assumption requires that the number $n'_v$ of elements in each subsample increases at the same speed as the sample size $n_v,$ allowing  each terminal node to have at least $n_{0v}$ observations.



\subsection{Asymptotic results}

In this section, we state some results pertaining to sequences of RF model-assisted estimators $\{\widehat{t}_{rf}\}$. The corresponding results for the model-assisted estimators $\{\widehat{t}^*_{rf}\}$ can be found in the Supplementary Material. 

\begin{result} \label{res1}
Consider a sequence of RF model-assisted estimators $\{\widehat{t}_{rf}\}$. Then, there exist positive constants $\tilde C_1, \tilde C_2$  such that 
$$
\mathbb{E}_p  \bigg\rvert \dfrac{1}{N_v} \left(\widehat{t}_{rf} - t_y\right)\bigg\rvert \leqslant \frac{\tilde C_1}{\sqrt{n_v}}+\frac{\tilde C_2}{n_{0v}}, \quad\mbox{with $\xi$-probability one.}
$$
If $\displaystyle\frac{n^{u}_v}{n_{0v}}=O(1)$ with $1/2\leqslant u \leqslant 1,$ then there exists a positive constant $\tilde C$ such that
	\begin{equation*}
	\mathbb{E}_p  \bigg\rvert \dfrac{1}{N_v} \left(\widehat{t}_{rf} - t_y\right)\bigg\rvert  \leqslant \dfrac{\tilde C}{\sqrt{n_{v}}}, \quad\mbox{with $\xi$-probability one.}
	\end{equation*}

\end{result}



Result \ref{res1} implies that the RF model-assisted estimator $\{\widehat{t}_{rf}\}$ is asymptotically design-unbiased, i.e.,
\begin{equation*}
\lim_{v \to \infty} \mathds{E}_p \left[\dfrac{1}{N_v} \left(\widehat{t}_{rf} - t_y\right)\right] = 0, \quad \mbox{with $\xi$-probability one},
\end{equation*}
and design-consistent in the sense that
\begin{equation*}
\lim_{v \to \infty} \mathds{E}_p \left[\mathbf{1}_{ \{N_v^{-1}\rvert  \widehat{t}_{rf} - t_y \rvert  >\eta\}} \right]= 0,\quad \mbox{with $\xi$-probability one}
\end{equation*}
for all $\eta >0. $ Moreover, 
 if $n_{0v}$ is large enough with respect to the sample size $n_v$, the RF estimator $\widehat{t}_{rf}$ is $\sqrt{n_v}$-consistent. For a given partition, note that the number of terminal nodes is of order $O(n_v/n_{0v})$, and if $n_{0v}$ satisfies the condition from the Result \ref{res1}, the number of terminal nodes is of order $O(n^{1-u})$ for $1/2\leqslant u \leqslant 1$. 
The next result shows that the RF model-assisted estimator $\widehat{t}_{rf} $ is asymptotically equivalent to the pseudo-generalized difference estimator: 
\begin{eqnarray}\label{pgd_rf}
\widehat{t}_{pgd} = \sum_{k\in U}\widetilde{m}_{rf}(\mathbf{x}_k)+\sum_{k\in S}\frac{y_k-\widetilde{m}_{rf}(\mathbf{x}_k)}{\pi_k},
\end{eqnarray}
where $\widetilde{m}_{rf}(\mathbf{x}_k)$ is given by  (\ref{localAve0}).

\begin{result}\label{res2}
Consider a sequence of RF estimators $\{\widehat{t}_{rf}\}.$ Assume also that $\displaystyle\frac{n^{u}_v}{n_{0v}}=O(1)$ with $1/2< u\leqslant 1.$  
Then, $\{\widehat{t}_{rf}\}$ is asymptotically equivalent to the pseudo-generalized difference estimator $\widehat{t}_{pgd}$ in the sense that
	\begin{equation*}
	\dfrac{\sqrt{n_v}}{N_v} \left(\widehat{t}_{rf}- t_y\right)=\dfrac{\sqrt{n_v}}{N_v}\left(\widehat{t}_{pgd} -t_y\right) + o_\mathbb{P}(1).
	\end{equation*}
\end{result}
\noindent From  Proposition  \ref{res2}, it follows that the asymptotic variance of $\widehat{t}_{rf}$  can be approximated by the variance of (\ref{pgd_rf}). That is, 
\begin{eqnarray}
\mathbb{AV}_p\left(\frac{1}{N_v}\widehat{t}_{rf}\right)&=&\mathbb{V}_p \left(\frac{1}{N_v}\widehat{t}_{pgd}\right)\nonumber \\
&=& \frac{1}{N^2_v}\sum_{k\in U_v}\sum_{\ell\in U_v}(\pi_{kl}-\pi_k\pi_\ell)\frac{y_k-\widetilde{m}_{rf}(\mathbf{x}_k)}{\pi_k}\frac{y_\ell-\widetilde{m}_{rf}(\mathbf{x}_\ell)}{\pi_\ell}.\nonumber\\\label{var_diff_est}
\end{eqnarray}
While the RF model-assisted estimator $\widehat{t}_{rf}$ is design-consistent as long as $n_{0v}$ and $n_v$ grow to infinity (Result \ref{res1}), the asymptotic equivalence of $\widehat{t}_{rf}$ with the pseudo-generalized difference estimator $\widehat{t}_{pgd}$ is obtained only for $n_{0v}$ satisfying a certain rate. Stronger assumptions on higher-order inclusion probabilities \citep{breidt_opsomer_2000, mcconville2019automated} are required in order to show that the asymptotic mean squared error of  $\widehat{t}_{rf}$ is equivalent to the variance of the pseudo-generalized difference estimator. We do not pursue this further. 

Expression (\ref{var_diff_est}) suggests that $\widehat{t}_{rf}$ is efficient if  the residuals $y_k - \widetilde{m}_{rf}(\bx_k)$
are small for all $k\in U_v. $ 
The asymptotic variance given in (\ref{var_diff_est}) cannot be computed in practice because the residuals, $y_k - \widetilde{m}_{rf}(\bx_k)$, $k \in U,$ are unknown.   
Assuming that $\pi_{k\ell} >0$ for all pairs $(k, \ell) \in U_v\times U_v,$ a design-consistent estimator of 
$\mathbb{AV}_p\left(\frac{1}{N_v}\widehat{t}_{rf}\right)$
 is given by
\begin{equation} \label{varestimator}
\widehat{\V}_{rf}\left(\frac{1}{N_v}\hat t_{rf}\right) =  \dfrac{1}{N_v^2} \sum_{k\in U_v} \sum_{\ell \in U_v}  I_k I_\ell\dfrac{\pi_{k\ell} - \pi_k \pi_\ell }{\pi_{k\ell}} \dfrac{y_k - \widehat{m}_{rf}(\bx_k)}{\pi_k} \dfrac{y_\ell - \widehat{m}_{rf}(\bx_\ell)}{\pi_\ell},
\end{equation}
where $\widehat{m}_{rf}(\bx_k)$ is given by (\ref{samplePred}).
To  establish the design consistency of (\ref{varestimator}), we require the following additional assumption:

\hypH{We assume that 
	$ \lim\limits_{v \to \infty} \max\limits_{i,j,k,\ell \in D_{4, N_v}} \rvert \mathbb{E}_p \left\{ \left(I_iI_j - \pi_i\pi_j\right) \left(I_kI_\ell - \pi_k \pi_\ell\right) \right\}\rvert =0,$
	where $D_{4,N_v}$ denotes the set of distinct $4$-tuples from $U_v.$ \label{H6}}

Assumption (H\ref{H6}) was suggested by \cite{breidt_opsomer_2000} and, together with (H\ref{H2})-(H\ref{H3}), is used to establish the design consistency of the unbiased estimator of the variance of the Horvitz-Thompson estimator $\sum_{k\in S_v}y_k/\pi_k,$  assuming that the survey variable $Y$ has finite fourth moment. Assumption (H\ref{H6}) is satisfied for simple random sampling without replacement and stratified simple random sampling without replacement. It is also satisfied for high entropy sampling designs \citep{boistard_lopuhaa_ruiz-gazen_2012, CGL2012d}. 

\begin{result}\label{var_est_F1}
	Consider a sequence of RF model-assisted estimators $\{\widehat{t}_{rf}\}.$  Assume also that $\displaystyle\frac{n^{u}_v}{n_{0v}}=O(1)$ with $1/2< u\leqslant 1.$  
	Then, the variance estimator $\widehat{\V}_{rf}(\widehat t_{rf})$ is asymptotically design-consistent for the asymptotic variance $\mathbb{AV}_p \left( \widehat{t}_{rf}\right).$ That is,
	\begin{equation*}
	\lim_{v \to \infty}\mathbb{E}_p \left(\dfrac{n_v}{N_v^2} \biggr \rvert\widehat{\V}_{rf}(\widehat t_{rf})- \mathbb{AV}_p (\widehat{t}_{rf})  \biggr \rvert \right) = 0. 
	\end{equation*}
\end{result}
Finally, we establish the central limit theorem that can be used to obtain
asymptotically normal confidence intervals of $t_y$. To that end, we assume that $\widehat{t}_{pgd}$ is normally distributed, an assumption that is satisfied in many classical sampling designs; e.g., see \cite{Fuller2009}.    
\hypH{ The sequence of pseudo-generalized  difference estimators $\{\widehat{t}_{pgd}\}$ satisfies 
	\begin{equation*}
	\dfrac{N_v^{-1} \left(\widehat{t}_{pgd} - t_y\right)}{\sqrt{\mathbb{V}_p \left(N^{-1}_v\widehat{t}_{pgd}\right)}}  \xrightarrow[v \to \infty]{\mathcal{L}}\mathcal{N} (0, 1),
	\end{equation*}
	where $\mathbb{V}_p \left(N^{-1}_v\widehat{t}_{pgd}\right)$ is given by (\ref{var_diff_est}).\label{H10}}


\begin{result}\label{TCL_est_F1}
	Consider the sequence of RF estimators  $\{\widehat{t}_{rf}\}.$ 
	Then,
	\begin{equation*}
	\dfrac{N_v^{-1} \left(\widehat{t}_{rf} - t_y\right)}{\sqrt{\widehat{\V}_{rf}(N^{-1}_v\widehat{t}_{rf})}}  \xrightarrow[v \to \infty]{\mathcal{L}}\mathcal{N} \left(0, 1\right).
	\end{equation*}
\end{result}
The proof of Result \ref{TCL_est_F1} is a direct application of Results \ref{res2} and \ref{var_est_F1}, and is thus omitted.

%
%

\section{A model calibration procedure for handling multiple survey variables}\label{model_calib}

In practice, most surveys conducted by national statistical offices (NSO) collect information on multiple survey variables. The collected data are stored in rectangular data files. A column of weights, referred to as a weighting system, is made available on the data file. This weighting system can then be applied to obtain an estimate for any survey variable. However, applying a RF algorithm yield the variable-specific weights (\ref{weight_rf_ech}). In other words, the weights were derived to obtain an estimate of the total for a specific survey variable $Y$. Hence, applying the weights (\ref{weight_rf_ech}) to other survey variables may produce inefficient estimators. A solution to this issue consists of developing multiple sets of weights, one for each survey variable. This is usually deemed undesirable by data users who are used to work with a single set of weights. In this section, we describe a model calibration procedure \citep{wu_sitter_2001}, originally proposed by \cite{montanari_ranalli_2009}, that yields a single weighting system while accounting for multiple survey variables that are deemed important.   

Suppose that we can identify a subset of survey variables $Y_1, \ldots, Y_q,$ that are deemed important. 
We postulate the following working model for each variable:
\begin{equation}\label{mod3}
\mathbb{E}\left[Y_{jk}\mid \boldsymbol{X}_k = \mathbf{x}_k\right]  = m^{(j)}(\mathbf{x}_k^{(j)}), \quad j=1, \cdots, q,
\end{equation}
where $m^{(j)}(\cdot)$ is an unknown function and $\mathbf{x}_k^{(j)}$ is a vector of auxiliary variable associated with unit $k$ for the variable $Y_j$. We allow a different link functions $m(\cdot)$ and different sets of explanatory variables for each of the survey variables $Y_1, \ldots, Y_q.$ The interest lies in estimating the population totals $t_{y_1}, \ldots, t_{y_q}$. We assume that each of these totals is estimated using a model-assisted estimator of the form (\ref{ma}) but with possibly different methods. For instance, some of the estimates may be based on a parametric working model, while others may be based on a nonparametric working model (e.g., RF). We can construct the set of $q$ predicted values  $\widehat{m}^{(1)}(\mathbf{x}_k^{(1)}), \ldots, \widehat{m}^{(q)}(\mathbf{x}_k^{(q)}),$ for  $k \in U.$

In addition, we assume that, at the estimation stage, a vector $\mathbf{v}_k$ of size $q'$ of calibration variables is available for $k \in S$ and that the corresponding vector of population totals $\mathbf{t}_{\mathbf{v}}= \sum_{k \in U}\mathbf{v}_k$ is known. In practice, survey managers often want to ensure consistency between survey estimates and known population totals for important variables such as  gender and age group. 


Given these predictions  $\widehat{m}^{(1)}(\mathbf{x}_k^{(1)}), \ldots, \widehat{m}^{(q)}(\mathbf{x}_k^{(q)}),$ and the vector calibration variables $\mathbf{v},$ we seek calibrated weights ${w}_k^C,$ $k \in S,$ as close as possible to the initial weights $\pi_k^{-1}$ subject to the following $q+q'+1$ calibration constraints:
\begin{equation}
\sum_{k \in S}{w}_k^C=N,\label{c1}
\end{equation}
\begin{equation}
\sum_{k \in S}{w}_k^C\widehat{m}^{(j)}(\mathbf{x}_k^{(j)})=\sum_{k \in U} \widehat{m}^{(j)}(\mathbf{x}_k^{(j)}),\quad j=1, \dots, q \label{c3},
\end{equation}
\begin{equation}\label{c4}
\sum_{k \in S} {w}_k^C\mathbf{v}_k=\sum_{k \in U} \mathbf{v}_k.
\end{equation}
More specifically, we seek calibrated weights ${w}_k^C$ such that
$$\sum_{k \in S}G({w}_{k}^C/\pi_k^{-1})$$
is minimized subject to (\ref{c1})--(\ref{c4}), where $G(\cdot)$ is a pseudo-distance function measuring the closeness between two sets of weights, such that $G({w}_{k}^C/\pi_k^{-1})\ge 0,$ differentiable with respect to ${w}_k^C$, strictly convex, with continuous derivatives $g({w}_{k}^C/\pi_k^{-1})=\partial G({w}_{k}^C/\pi_k^{-1})/\partial {w}_{k}^C$  such that $g(1)=0$; see \cite{deville_sarndal_1992}.

The weights ${w}_k^C$ are given by
 \begin{equation}\label{cal_weight}
  {w}_k^C= {\pi}_k^{-1} F(\mbox{\boldmath$ \widehat {\boldsymbol{\lambda}}$}^{\top} \mathbf{h}_k),\nonumber
 \end{equation}
 where $F(.)$ is the calibration function defined as the inverse of $g(.),$ $\mbox{\boldmath$ \widehat {\boldsymbol{\lambda}}$}$ is a $q+q'+1$-vector of estimated coefficients and
 \begin{equation}\label{h}
\mathbf{h}_k = \left(1,\widehat{m}^{(1)}_{k}-\widehat{\overline{m}}^{(1)},\ldots, \widehat{m}^{(q)}_{k}-\widehat{\overline{m}}^{(q)}, v_{1k}, \ldots, v_{q'k}\right)^{\top}
 \end{equation}
with  $\widehat{m}^{(j)}_{k}\equiv m^{(j)}(\mathbf{x}_k^{(j)})$  and $\widehat{\overline{m}}^{(j)}\equiv\sum_{k \in S}\pi_k^{-1}\widehat{m}_{k}^{(j)}/\sum_{k \in S}\pi_k^{-1},$ $j=1, \cdots, q.$ 

The calibrated weights $w_k^C$ may be viewed as a compressed score summarizing the information contained in the $q$ working models (\ref{mod3}) and the vector of calibration variables $\mathbf{v}$.
The weighting system $\{{w}_k^C; k \in S\}$ may be then applied to any survey variable $Y$, which leads to the model calibration type estimator
$$\widehat{t}_{y, mc}=\sum_{k \in S}{w}_k^Cy_k.$$ 
If the number of calibration constraints $q+q'+1$ is large, the resulting weights ${w}_k^C$ may be highly dispersed leading to potentially unstable estimates $\widehat{t}_{y, mc}$. A number of pseudo-distance functions such as the truncated linear and the logit methods may be used to limit the variability of the weights ${w}_k^C;$ see \cite{deville_sarndal_1992} for a description of these methods. A simple alternative is to use additional constraints on the weights as part of the calibration constraints. For instance, we may impose that ${w}_k^C< w_0,$ where $w_0$ is a threshold set by the survey statistician; see also \cite{Santacatterina2018} for alternative constraints on the weights.  Finally, we can relax the calibration constraints (\ref{c1})-(\ref{c4}) by considering a $L^2$-penalized criterion,  leading to a ridge-type model calibration estimator; see \cite{montanari_ranalli_2009}. \cite{montanari_ranalli_2009}  reports the results of a simulation study, assessing the performance of point estimators obtained through multiple and ridge model calibration methods.

\section{Simulation study}\label{simulation}
\subsection{Performance of point estimators} \label{simu1}
We conducted a simulation study to assess the performance of several model-assisted estimators, in terms of bias and efficiency. We generated a finite population of size $N=10, 000,$ consisting of a set of auxiliary variables and  8 survey variables. We first generated 7 auxiliary variables $X_0, \cdots, X_6,$ according to the following distributions:
$X_0 \sim \mathcal{U} (0,1)$;
$X_1 \sim \mathcal{N} \left(0,1\right)$,
$X_2 \sim \text{Beta} \left(3,1\right)$,
$X_3 \sim 2 \times \text{Gamma} \left(3, 2\right)$,
$X_4 \sim \text{Bernoulli}  (0.7),$ 
$X_5 \sim \text{Multinomial} (0.4, 0.3, 0.3)$ and $X_6 \sim \mathcal{E} (1)$. The variables $X_1, X_2, X_3, $ and $X_6$ have been standardized so as to have a mean and a variance equal to 0 and 1, respectively. 
To assess the performance of the proposed method in a high-dimensional setting, we also generated $100$ additional auxiliary variables $V_1, V_2, \cdots, V_{100},$ from a uniform distribution $\mathcal{U} (-1, 1)$. Given the $X$-variables and the $V$-variables, we generated the survey variables according to the following models:
\begin{enumerate}[Model 1:]
\item  $Y_1 = 1 + 2 \left(X_0 - 0.5 \right) + \mathcal{N} \left(0, 0.1\right) $  ; 
\item  $ Y_2 = 1 + 2 \left(X_0 - 0.5 \right)^2 +\mathcal{N} \left(0, 0.1\right)$; 
\item $Y_3 = 2 + X_6 + X_2 + X_3 + X_4 + X_5 + \mathcal{N} \left(0, 1\right)$; 
\item $Y_4 = 2 + \left(X_6 + X_2 + X_3 \right)^2 + \mathcal{N} \left(0, 1\right)$; 
\item $Y_{5} = 0.5 X_5 + \exp(-X_1) + 3X_4 + \exp(-X_6) + \mathcal{E} \left(1\right)$; 
\item $Y_6 = V_1^2 + \exp(-V_2^2) + \mathcal{N} \left(0, 0.3\right)$;
\item $Y_7 = V_1^2 + \exp(-V_2^2) + \mathcal{N} \left(0, 0.3\right)$;
\item $Y_8 = 3 + V_1V_2 + V_3^2 - V_4V_7+V_8V_{10} - V_6^2 + \mathcal{N} \left(0, 0.5\right)$.
\end{enumerate}
The errors in Model 5 have been scaled and centered so as to have a mean and a variance equal to 0 and 1, respectively. 
Models 1 and 2 were used in \cite{breidt_opsomer_2000}, while Models 7 and 8 were introduced in \cite{scornet2017tuning}. Models 1-8 were generated so as to include  a relatively wide range of relationships between the $Y$-variable and the set of explanatory variables:  linear/non-linear relationships, presence/absence of quadratic terms and presence/absence of interactions.  Our scenarios also included low, medium and high-dimensional settings. 
From the population, we selected $R=5\ 000$ samples, of size $n,$ according to simple random sampling without replacement. We used $n=250$ and $n=1000$. In each sample, we computed the following estimators: (i) The Horvitz-Thompson (HT) estimator given by (\ref{ht}); (ii) The generalized regression (GREG) estimator given by (\ref{ma}) with $\widehat{m}(\mathbf{x}_k) =\mathbf{x}_k^{\top} \widehat{\boldsymbol{\beta}}$; (iii) The model-assisted estimator (\ref{ma}) with $\widehat{m}(\mathbf{x}_k)$ obtained through regression trees (CART); and (iv) The model-assisted estimator (\ref{ma}) based on RF, where $\widehat{m}(\mathbf{x}_k)$ is given by (\ref{samplePred}). We considered three RF algorithms, each based on 1, 000 trees. The first (RF1) was based on bootstrap. The second algorithm (RF2) was based on subsampling  with a sampling fraction equal to 0.63 \citep{scornet2017tuning}. For both RF1 and RF2, the minimum number of observations per terminal node was set to $n_0=5$. Finally, the third algorithm (RF3) was based on bootstrap with $n_0=\sqrt{n}$ observations in each terminal node.  In RF1-RF3, we used $m_{try} = \sqrt{p}$ as it is the default number of variables considered for the splitting process in most software packages dealing with RF for regression. 

For the estimators GREG, CART, RF1, RF2 and RF3, the predictions $\widehat{m}(\mathbf{x}_k)$ were obtained using the working models described in Table \ref{Table_wm}.  For the survey variables $Y_7$ and $Y_8$, the working models were based on a large number of superfluous explanatory variables (50 and 100, respectively), which allowed us to assess the behavior of the resulting estimators in a medium/high dimensional setting.

\begin{table}[h!]
\caption{The working models}
	\centering
\begin{tabular}{|c|c|}
  \hline
  Survey variable & Vector of explanatory variable $\mathbf{X}$ \\
                  &  used in the working model  \\
  \hline
  $Y_1$ & $X_0$ \\
  $Y_2$ & $X_0$ \\
  $Y_3$ & $X_1-X_6$ \\
  $Y_4$ & $X_1-X_6$  \\
  $Y_5$ & $X_1-X_6$  \\
  $Y_6$ & $V_1-V_{10}$  \\
  $Y_7$ & $V_1-V_{50}$  \\
  $Y_8$ & $V_1-V_{100}$  \\
  \hline
\end{tabular}\label{Table_wm}
\end{table}

We were interested in estimating the population total $t_{y_{j}}=\sum_{k \in U}y_{kj},$ $j=1, \ldots, 8.$ As a measure of bias of an estimator $\hat{t}_{y_{j}}$, we used the Monte Carlo percent relative bias defined as
\begin{equation*}
RB(\widehat{t}_{y_{j}})=100 \times \dfrac{1}{R} \sum_{r=1}^R \dfrac{ (\widehat{t}^{(r)}_{y_{j}} - t_{y_{j}}) }{ t_{y_{j}}},
\end{equation*}
where $ \widehat{t}^{(r)}_{y_{j}}$ denotes the estimator $\widehat{t}_{y_{j}}$  in the $r$th iteration, $r=1, ..., R$. As a measure of efficiency of an estimator $\hat{t}_{y_{j}}$, we used the relative efficiency, using the Horvitz-Thompson estimator, $\widehat{t}_{y_{j}, \pi},$ as the reference:
\begin{equation*}
RE(\widehat{t}_{y_{j}}) =100 \times \dfrac{ MSE(\widehat{t}_{y_{j}}) }{ MSE(\widehat{t}_{y_{j}, \pi}) },
\end{equation*}
where $$MSE( \widehat{t}_{y_{j}} ) = \dfrac{1}{R} \sum_{r=1}^R (  \widehat{t}^{(r)}_{y_{j}} - t_{y_{j}} )^2$$ 
and $MSE(\widehat{t}_{y_{j}, \pi})$ is defined similarly. 
The results are displayed in Tables \ref{Table_R1} and \ref{Table_R2}. The simulations were performed using the R software with the  package \textit{ranger} \citep{wright2015ranger}.

\begin{table}[h!]
		\centering	\footnotesize{	\caption{Monte Carlo percent relative bias (RB) and Monte Carlo efficiency (RE) of several model-assisted estimators for $n=250$}
	\vspace{2mm}
		\begin{tabular}{llrlrlrlrlr}
			\hline
			Population &    & \multicolumn{1}{c}{GREG} &  & \multicolumn{1}{c}{CART} &  & \multicolumn{1}{c}{RF1} &  & \multicolumn{1}{c}{RF2} &  & \multicolumn{1}{c}{RF3} \\ \cline{1-3} \hline
			&    & \multicolumn{1}{l}{}     &  & \multicolumn{1}{l}{}     &  & \multicolumn{1}{l}{}    &  & \multicolumn{1}{l}{}    &  & \multicolumn{1}{l}{}    \\
			\multirow{2}{*}{$Y_1$} & RB & -0.0                  &  & -0.0                    &  & -0.0                   &  & -0.0                   &  & 0.0                       \\
			& RE & 3.0                     &  & 3.5                     &  & 3.7                    &  & 3.6                    &  & 3.4                    \\
			&    & \multicolumn{1}{l}{}     &  & \multicolumn{1}{l}{}     &  & \multicolumn{1}{l}{}    &  & \multicolumn{1}{l}{}    &  & \multicolumn{1}{l}{}    \\
			\multirow{2}{*}{$Y_2$}                  & RB & -0.0                    &  & 0.0                        &  & 0.0                       &  & 0.0                       &  & 0.0                       \\
			& RE & 101.0                  &  & 37.6                    &  & 39.4                   &  & 38.3                   &  & 35.0                   \\
			&    & \multicolumn{1}{l}{}     &  & \multicolumn{1}{l}{}     &  & \multicolumn{1}{l}{}    &  & \multicolumn{1}{l}{}    &  & \multicolumn{1}{l}{}    \\
			\multirow{2}{*}{$Y_3$}                   & RB & 0.0                     &  & -0.0                    &  & -0.1                   &  & -0.1                    &  & -0.0                   \\
			& RE & 19.6                    &  & 55.2                    &  & 33.8                   &  & 34.0                   &  & 35.4                   \\
			&    & \multicolumn{1}{l}{}     &  & \multicolumn{1}{l}{}     &  & \multicolumn{1}{l}{}    &  & \multicolumn{1}{l}{}    &  & \multicolumn{1}{l}{}    \\
			\multirow{2}{*}{$Y_4$}                   & RB & -0.7                    &  & -1.2                    &  & -1.2                   &  & -1.5                   &  & -0.7                   \\
			& RE & 81.1                    &  & 61.1                    &  & 49.7                   &  & 49.0                   &  & 53.1                  \\
			&    & \multicolumn{1}{l}{}     &  & \multicolumn{1}{l}{}     &  & \multicolumn{1}{l}{}    &  & \multicolumn{1}{l}{}    &  & \multicolumn{1}{l}{}    \\
			\multirow{2}{*}{$Y_5$}                 & RB & -0.1                    &  & 0.1                      &  & -0.0                   &  & -0.0                   &  & -0.0
			                   \\
			& RE & 37.9                    &  & 32.7                    &  & 25.8                   &  & 26.5                   &  & 30.7                   \\
			&    & \multicolumn{1}{l}{}     &  & \multicolumn{1}{l}{}     &  & \multicolumn{1}{l}{}    &  & \multicolumn{1}{l}{}    &  & \multicolumn{1}{l}{}    \\
			\multirow{2}{*}{$Y_6$}                  & RB & -0.0                    &  & 0.3                     &  & -0.0                   &  & -0.0                   &  & -0.0                   \\
			& RE & 105.2                  &  & 72.2                    &  & 57.5                    &  & 57.5                   &  & 58.3                   \\
			&    & \multicolumn{1}{l}{}     &  & \multicolumn{1}{l}{}     &  & \multicolumn{1}{l}{}    &  & \multicolumn{1}{l}{}    &  & \multicolumn{1}{l}{}    \\
			\multirow{2}{*}{$Y_7$}                  & RB & -0.0                    &  & 0.2                     &  & 0.1                    &  & 0.0                    &  & 0.0
			                    \\
			& RE & 127.6                   &  & 84.3                    &  & 75.8                   &  & 75.5                   &  & 76.8                   \\
			&    & \multicolumn{1}{l}{}     &  & \multicolumn{1}{l}{}     &  & \multicolumn{1}{l}{}    &  & \multicolumn{1}{l}{}    &  & \multicolumn{1}{l}{}    \\
			\multirow{2}{*}{$Y_8$} & RB & 0.0                        &  & 0.0                        &  & 0.0                       &  & 0.0                    &  & 0.0                    \\
			& RE & 127.0                   &  & 135.6                   &  & 92.7                   &  & 92.5                   &  & 95.6                   \\
			&    & \multicolumn{1}{l}{}     &  & \multicolumn{1}{l}{}     &  & \multicolumn{1}{l}{}    &  & \multicolumn{1}{l}{}    &  & \multicolumn{1}{l}{}    \\ \cline{1-3} \hline 
		\end{tabular}\label{Table_R1}
	}

\end{table}

\begin{table}[h!]
\footnotesize{	\centering
	
	\caption{Monte Carlo percent relative bias (RB) and Monte Carlo efficiency (RE) of several model-assisted estimators for $n=1000$.}
	
	\vspace{2mm}
				\renewcommand\arraystretch{1.1}
		\begin{tabular}{lllllllllll}
			\hline
			Population &   & GREG   & CART   &  & RF1   &  & RF2   &  & RF3   &  \\ \hline
			&    &        &        &  &       &  &       &  &       &  \\
			\multirow{2}{*}{$Y_1$} & RB & 0.0      & 0.0      &  & 0.0     &  & 0.0     &  & 0.0     &  \\
			& RE & 2.8   & 3.5   &  & 3.6  &  & 3.5  &  & 3.0  &  \\
			&    &        &        &  &       &  &       &  &       &  \\
			\multirow{2}{*}{$Y_2$} & RB & 0.0      & 0.0      &  & 0.0     &  & 0.0     &  & 0.0     &  \\
			& RE & 100.1 & 38.7  &  & 40.5 &  & 39.6 &  & 33.3 &  \\
			&    &        &        &  &       &  &       &  &       &  \\
			\multirow{2}{*}{$Y_3$} & RB & 0.0      & 0.0   &  & -0.1 &  & -0.1  &  & 0.0     &  \\
			& RE & 20.4  & 41.1  &  & 28.1 &  & 27.8 &  & 31.6 &  \\
			&    &        &        &  &       &  &       &  &       &  \\
			\multirow{2}{*}{$Y_4$} & RB & -0.1  & -1.1  &  & -0.9  &  & -0.7 &  & -0.2 &  \\
			& RE & 78.9  & 52.3  &  & 36.7 &  & 36.1 &  & 44.5 &  \\
			&    &        &        &  &       &  &       &  &       &  \\
			\multirow{2}{*}{$Y_5$} & RB & -0.0  & 0.0   &  & 0.0  &  & 0.0  &  & -0.0 &  \\
			& RE & 37.3  & 24.5  &  & 20.9 &  & 21.2 &  & 24.8 &  \\
			&    &        &        &  &       &  &       &  &       &  \\
			\multirow{2}{*}{$Y_6$}& RB & 0.0      & 0.0   &  & -0.0 &  & -0.0 &  & -0.0 &  \\
			& RE & 101.1 & 65.5   &  & 49.1 &  & 49.2  &  & 50.3  &  \\
			&    &        &        &  &       &  &       &  &       &  \\
			\multirow{2}{*}{$Y_7$} & RB & 0.0      & 0.0   &  & 0.0  &  & 0.0  &  & 0.0  &  \\
			& RE & 105.5 & 73.2  &  & 63.3 &  & 63.2 &  & 65.0 &  \\
			&    &        &        &  &       &  &       &  &       &  \\
			\multirow{2}{*}{$Y_8$} & RB & -0.0  & -0.0  &  & -0.0 &  & -0.0 &  & 0.0     &  \\
			& RE & 166.6 & 137.6 &  & 96.0    &  & 95.7 &  & 89.5 &  \\
			&    &        &        &  &       &  &       &  &       &  \\ \hline
		\end{tabular}\label{Table_R2}
	
}	
\end{table}

We start by noting that all the estimators displayed a negligible bias in all the scenarios, as expected. Also, both RF1 and RF2 showed very similar performances in terms of bias and efficiency in all the scenarios. This is consistent with the empirical results of \cite{scornet2017tuning}; i.e.,  the strategy based on bootstrap and the strategy based on subsampling with a sampling fraction of 0.63 led to similar performances. The results for RF3 were similar to those obtained for RF1 and RF2, which suggests that the number of observations in each terminal node did not seem to affect the behavior of the point estimator, at least in our experiments. This may not be the case in other scenarios as we illustrate in Section \ref{hyper_param}.


In the case of a linear relationship (which corresponds to the survey variables $Y_1$ and $Y_3$), the GREG estimator was the most efficient, as expected. For instance, for the survey variables $Y_3$, the value of RE for the GREG estimator was about $19.6\%$, whereas the RF1, RF2 and RF3 estimators showed a value of RE of about $34\%$.  In the case of a nonlinear relationship (which corresponds to the survey variables $Y_2$ and $Y_4, \ldots, Y_8$), the GREG estimator was less efficient than RF1, RF2 and RF3. For instance, in the case of the variable $Y_4$, the GREG showed a value of RE of about 81.1\%, whereas the RE of RF estimators lied between $49.0\%$ and $53.1\%.$ For the variables $Y_6, \ldots,Y_8$, the GREG estimator was even less efficient than the Horvitz-Thompson estimator with values of RE ranging from $105\%$ to $127\%$.

In the case of a single explanatory variable (which corresponds to the survey variables $Y_1$ and $Y_2$), RF and regression trees displayed very similar performances. In contrast, the estimators RF1, RF2 and RF3 were more efficient than the CART estimator when the vector of explanatory variables was multi-dimensional (i.e., variables $Y_3, \ldots,Y_9$). In a high-dimensional setting (which corresponds to the survey variables $Y_7$ and $Y_8$), the estimators based on RF estimators were more efficient than the Horvitz-Thompson estimator, even for $n=250$. 
 
\subsection{Performance of the proposed variance estimator} \label{simu2}

We have also investigated the performance of the variance estimator $\widehat{\V}_{rf}$ given by (\ref{varestimator}) in the case of RF with subsampling, in terms of relative bias and coverage of normal-based confidence intervals. We generated a population of size $N=100,000$ according to Model 5. The sample size was set to $n=500; 1,000; 5,000; 10,000; 20,000$ and $50,000$. Here, we present the results for $B=1$ but other values of $B$ led to similar results and are not shown here. As we suspected that the  number of observations in each terminal node, $n_0,$ may have an impact on the behavior of $\widehat{\V}_{rf}$, we used different values for $n_0$ : $n_0 =\floor*{n^{a/20}}$ for $a=1;3;5;7;9;11;13;15;17$. The choice $n_0 = \floor*{n^{11/20}}$ was advocated by \cite{mcconville2019automated}. Figure \ref{fig2} shows the Monte Carlo percent relative bias of $\widehat{\V}_{rf}$ for different values of $n$ and $n_0$. Figure \ref{fig3} shows the Monte Carlo coverage rate of the confidence interval, $\widehat{t}_{rf} \pm 1.96\sqrt{\widehat{\V}_{rf}},$ for different values of $n$ and $n_0$.

From Figure \ref{fig2}, we note that $\widehat{\V}_{rf}$ was severely biased for small values of $n_0$. For a given value of $n_0$, we note that the bias decreased as $n$ increased and for a given value of $n$, the bias decreased as $n_0$ increased.
 The  bias was likely due to the fact that the predicted values in a terminal node are defined as the weighted mean of the sample observations within the same terminal node. As a result, the resulting estimator may suffer from small sample bias for small values of $n_0.$
From Figure \ref{fig3}, we note that the confidence intervals performed poorly for small values of $n_0$, which can be explained by the fact that the variance estimator $\widehat{\V}_{rf}$  lead to substantial underestimation of the true variance in these scenarios. For $n_0 = \floor{13/20}$, the confidence intervals performed relatively well with coverage rates close to the nominal rate.

\begin{figure}[h!]
	\centering
	\includegraphics[width=0.9\textwidth]{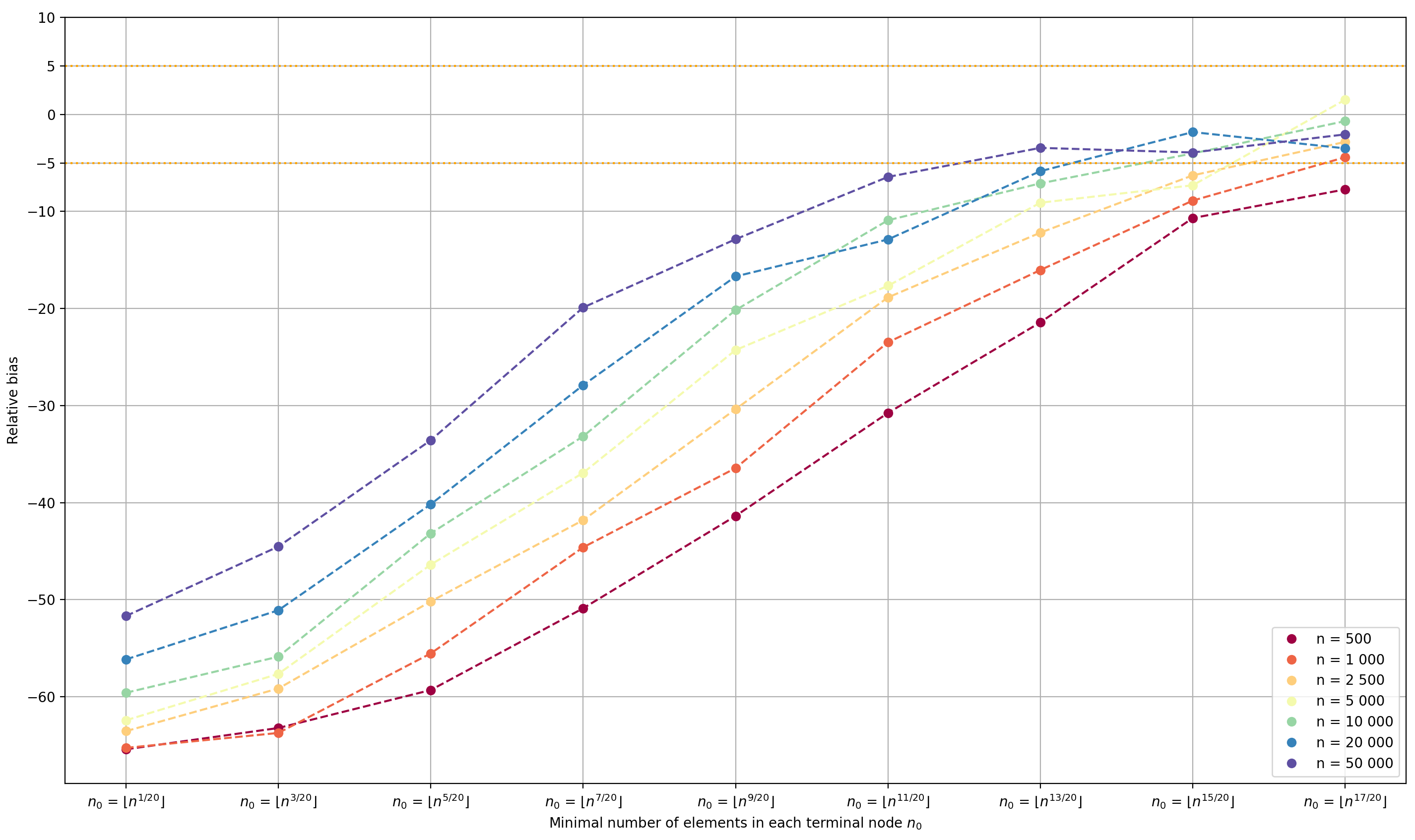}
	\caption{Evolution of the relative bias with respect to $n_0$. } \label{fig2}
\end{figure}



\begin{figure}[h!]
	\centering
	\includegraphics[width=0.9\textwidth]{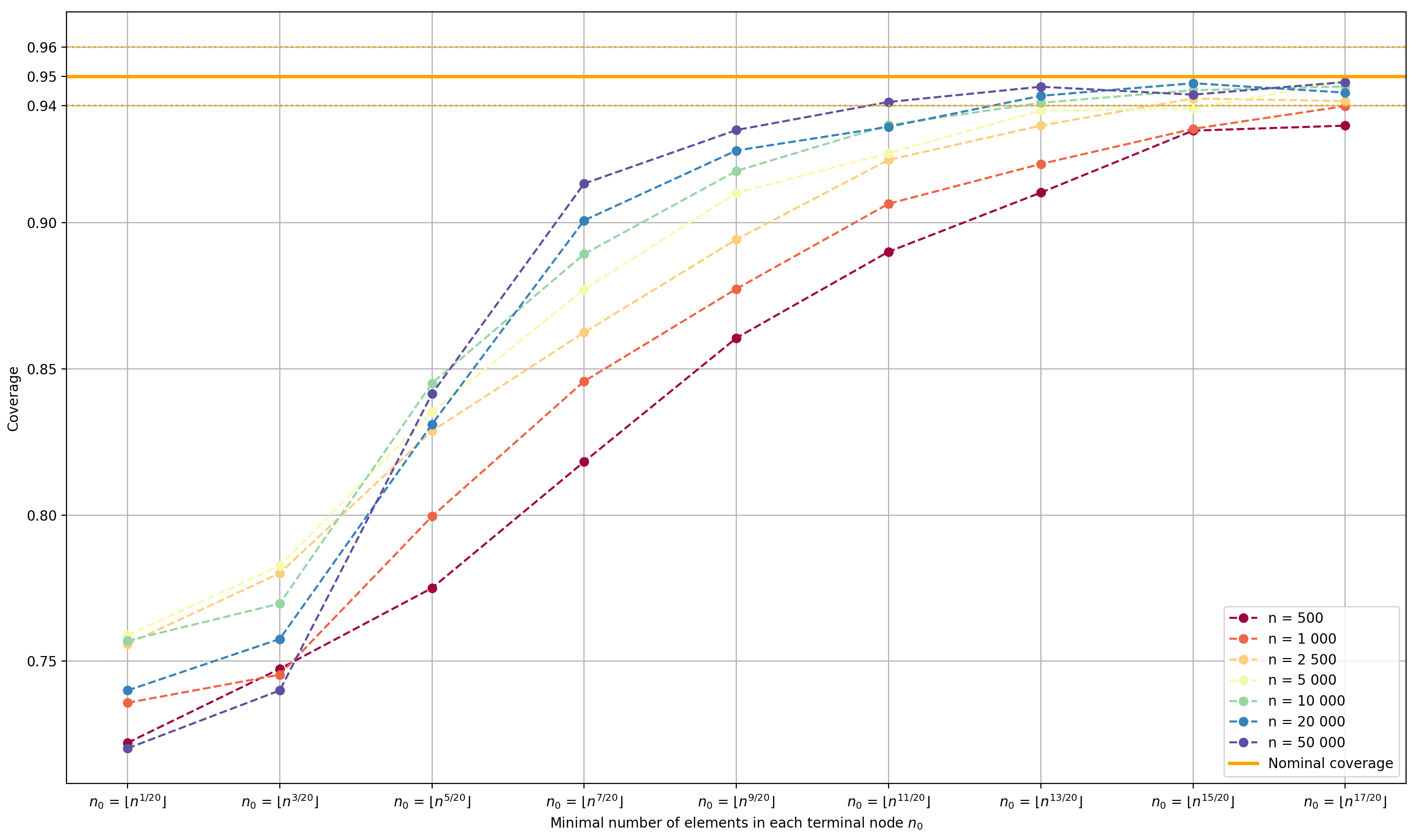}
	\caption{Evolution of  the effective coverage with respect to $n_0$. } \label{fig3}
\end{figure}

 \subsection{Choice of hyper-parameters}\label{hyper_param}

To get a better understanding of how the choice of hyper-parameters impacts the behavior of model-assisted estimators based on RF, we conducted additional scenarios.
We first identified the following important hyper-parameters involved in the RF algorithm of \cite{breiman2001random}: 
\begin{enumerate}[i)]
	\item The minimal number of observations, $n_0,$ in each terminal node; 
	\item The number of trees in the forest $B$; 
	\item The number of variables considered for the search of the best split in the optimization criterion (\ref{opt3});
	\item The resampling mechanism.
\end{enumerate}
The additional scenarios were conducted using a finite population of size $N=10,000$ consisting of the survey variables $Y_5$ and $Y_8$ described in Section \ref{simu1}.  Recall that the working model for the survey variables $Y_5$ included the predictors $X_1$-$X_6$, whereas it included the predictors $V_1-V_{100}$ for the variable $Y_8$ (see Table 1).

From the population, we generated $R = 10,000$ samples, of size $n=1,000,$ according to simple random sampling without replacement. Figure \ref{fig3a} and Figure \ref{fig_n0y8} show, respectively, the relative efficiency of the model-assisted estimators based on RF,  $\widehat{t}_{rf},$ corresponding to $Y_5$ and $Y_8$, respectively, for several values of $n_0$.  Figure \ref{fig3a}  suggests that $\widehat{t}_{rf}$ was much more efficient than the Horvitz-Thompson estimator for small values of $n_0$ and that the value of RE approached 100 as $n_0$ increased. This result can be explained by the fact that small values of $n_0$ led to homogeneous terminal nodes, which in turn led to small residuals $y_k - \widehat{m}_{rf}(\bx_k)$. For the survey variable $Y_8$, we note from Figure \ref{fig_n0y8} that the value of $n_0$ did not seem to affect the efficiency of the corresponding model-assisted estimator.

%
%
\begin{figure}[h!]
	\centering
	\includegraphics[width=0.9\textwidth]{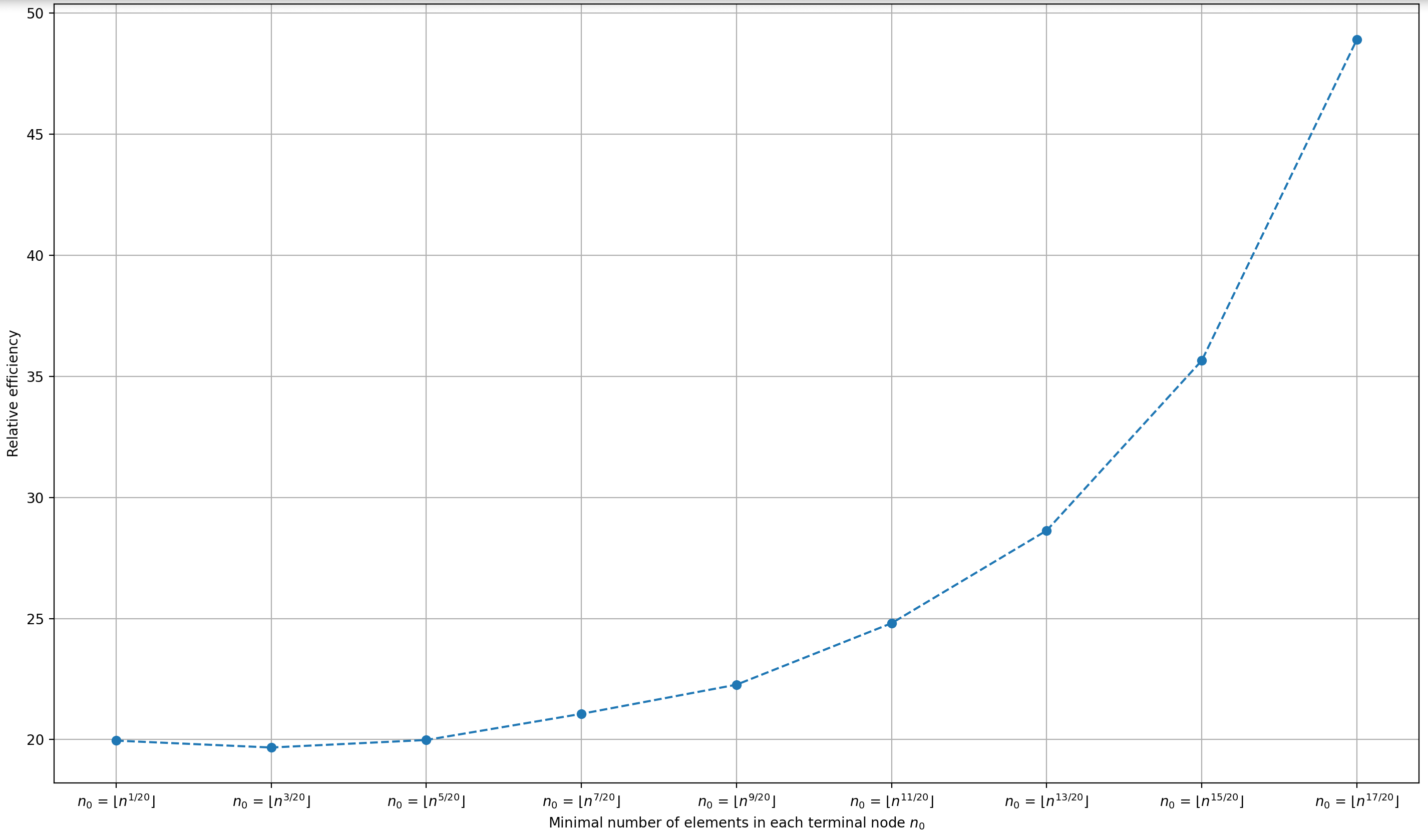}
	\caption{Relative efficiency of $\widehat{t}_{rf}$ for the survey variable $Y_5$ and for several values of $n_0$.}  \label{fig3a}
\end{figure}

\begin{figure}[h!]
	\centering
	\includegraphics[width=0.9\textwidth]{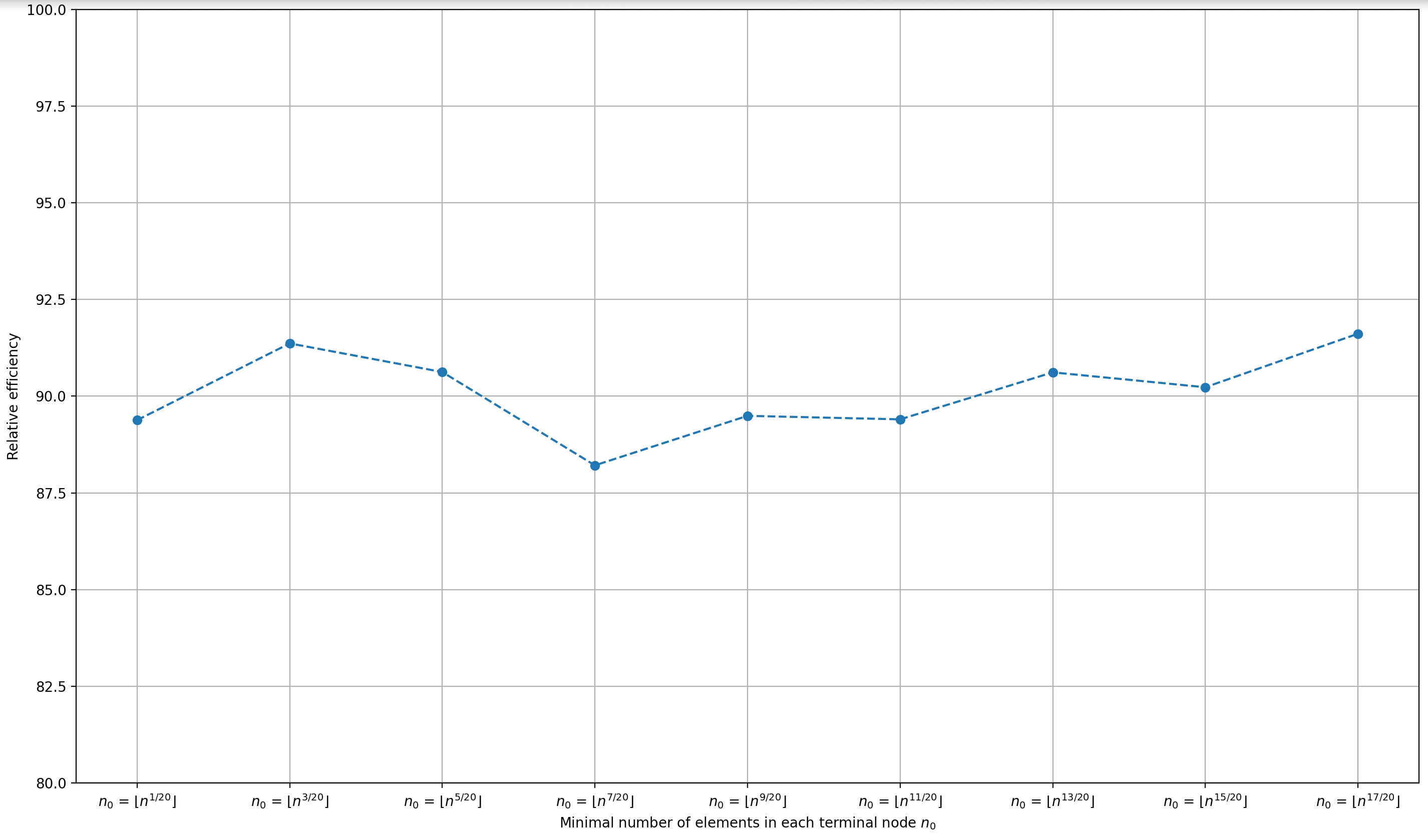}
	\caption{Relative efficiency of $\widehat{t}_{rf}$ for the survey variable $Y_8$ and for several values of $n_0$ } \label{fig_n0y8}
\end{figure}

Figure \ref{fig4} display the relative efficiency for several values of $B,$ the number of trees in the forest for the survey variable $Y_8$. As expected, a small value of $B$ causes the estimator $\widehat{t}_{rf}$ to loose some efficiency.  Figure \ref{fig4} suggests that $B=50$ led to good results and that the efficiency of $\widehat{t}_{rf}$ was not much affected by the number of trees $B$ for $B \ge 50$.  Nevertheless, it is advisable to choose a large value of $B$ if the computational capacity permits. The results for the survey variable $Y_5$ were very similar and so we omit them. 
\begin{figure}[h!]
	\centering
	\includegraphics[width=1\textwidth]{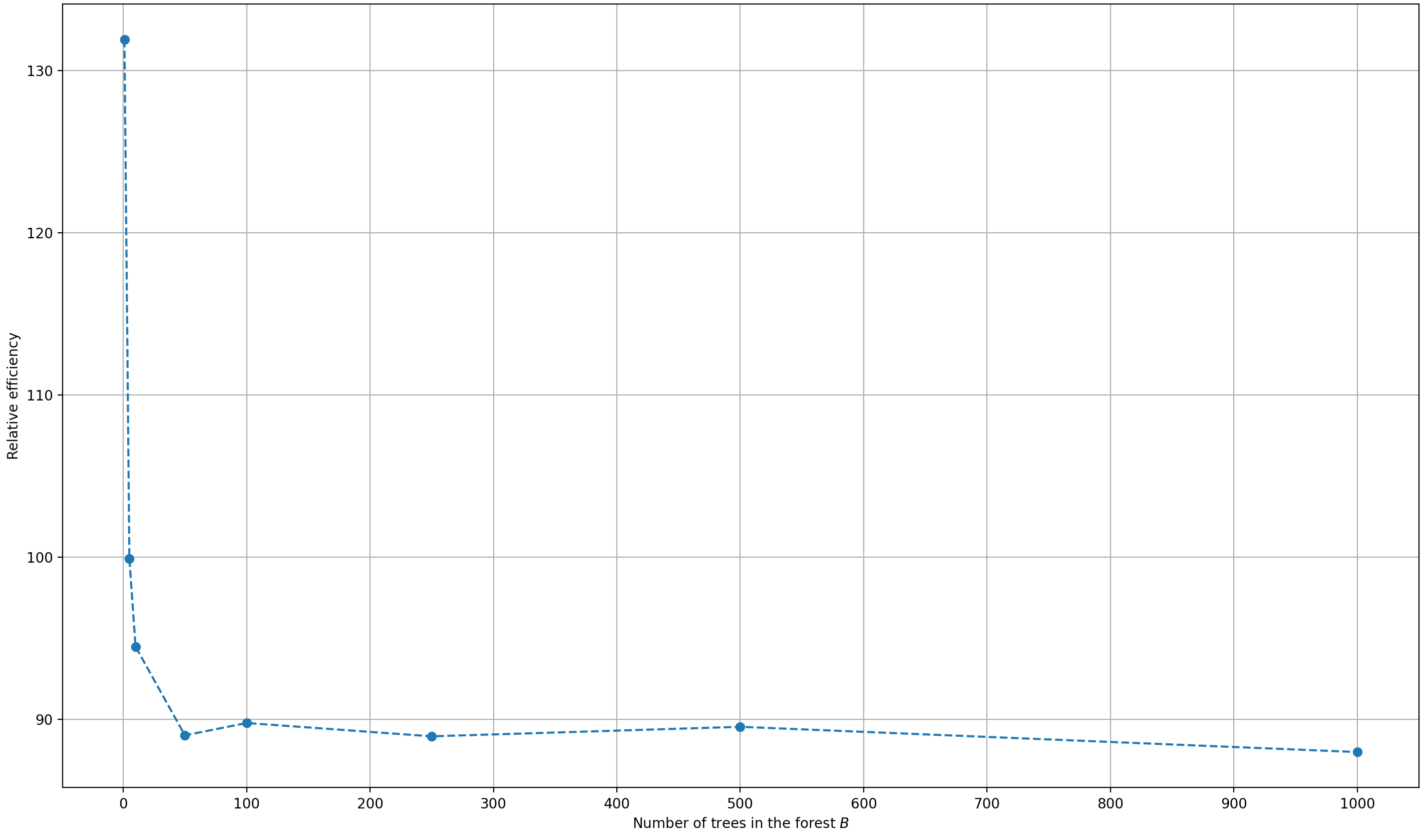}
	\caption{Relative efficiency of $\widehat{t}_{rf}$ for the survey variable $Y_8$ and  for several values of $B$. } \label{fig4}
\end{figure}

In most software packages, the default number of variables considered for the splitting process is $m_{try} = \sqrt{p}$ in case of regression. In our simulations, this choice led to satisfactory results in most scenarios. Figure \ref{fig5} shows the relative efficiency of $\widehat{t}_{rf}$ for the survey variable $Y_8$ and for several values of $m_{try}$. Since the working model for $Y_8$ contained $p=100$ explanatory variables, the default value $\sqrt{p}$ was equal to 10. Although the value $\sqrt{p}=10$ was not the best choice for optimal performances, it led  efficient model-assisted estimators. Furthermore, the relative efficiency did not vary much for values $B$ larger than 30. \\
\begin{figure}[h!]
	\centering
	\includegraphics[width=0.95\textwidth]{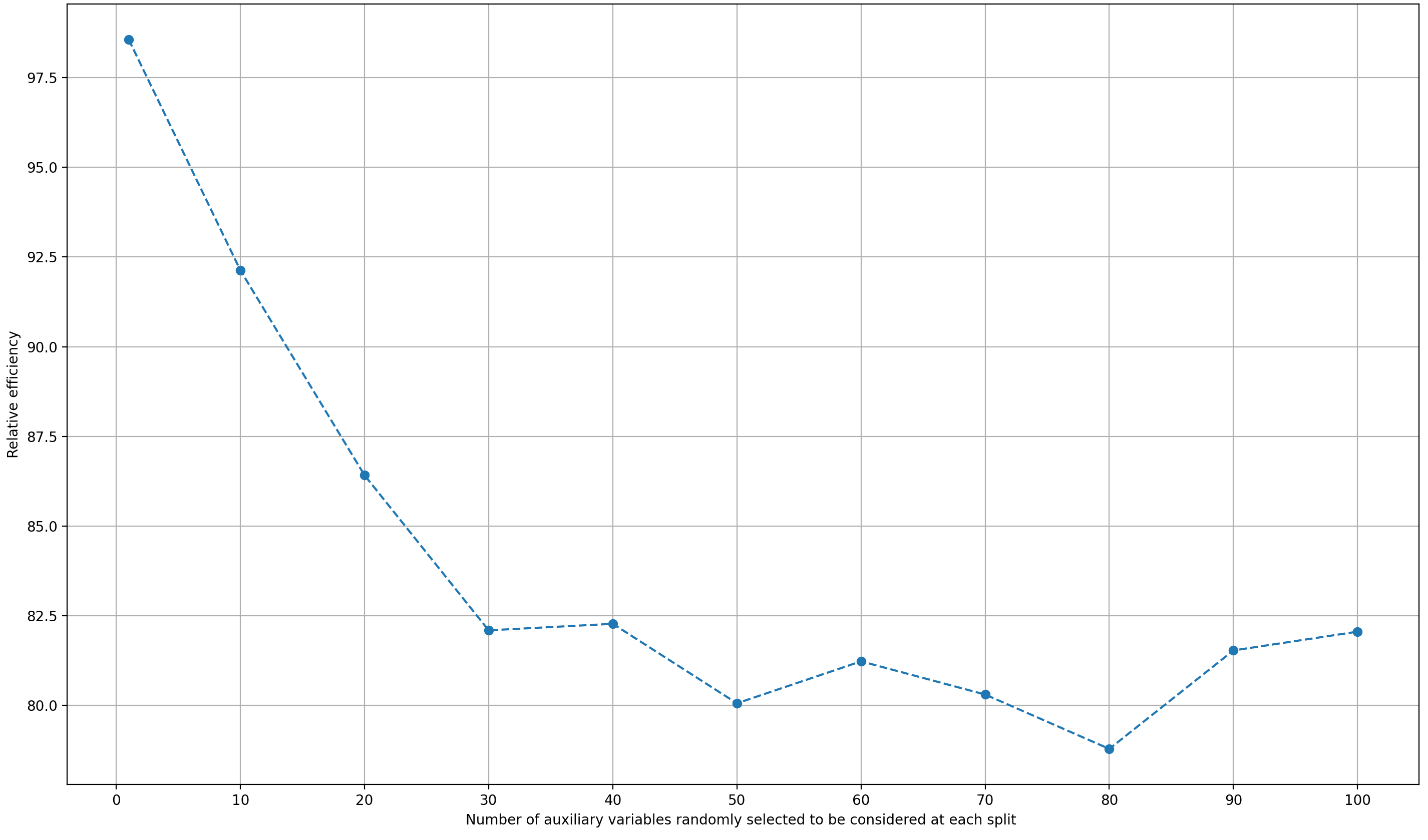}
	\caption{Relative efficiency of $\widehat{t}_{rf}$ for the survey variable $Y_8$ and for several values of $m_{try}$. } \label{fig5}
\end{figure}
Turning to the resampling mechanism, a common choice is to use bootstrap (with replacement), for which some of the results presented in the paper do not apply. However, as noted by several authors (see e.g. \cite{scornet2015consistency}, \cite{wager2014asymptotic} and the references therein) and as shown in our simulations, selecting the points without replacement rather than with replacement does not seem to affect the performance of the resulting model-assisted estimators in most cases. 

\subsection{Real data application}\label{real_data}

In this section, we apply the proposed methods using data collected by M\'ediam\'etrie, the company that measures the media audience in France. In this application, we focus on radio audiences. Each year, M\'ediam\'etrie conducts a survey aiming at gathering detailed information about French individuals 13 years of age and over, including socio-demographic variables and radio listening habits. We used the $2019$ radio audience data that consisted of $N = 26,293$ individuals. 
As a survey variable, we considered the binary variable $Y,$ such that $y_k= 1$ if an individual in the $k$th individual listens to the radio of interest on a daily basis, and $y_k = 0$ otherwise. For confidentiality reasons, we omit the name of the radio broadcaster. We aimed at estimating the proportion of French individuals who listen to the radio of interest on a daily basis, both at the overall population level and for several domains of interest. For each individual on the data set, we had access to $43$ socio-demographic variables (e.g., number of individuals in the household, age of each member of the household, gender, internet habits, occupation, etc.) and their listening habits of $21$ other radios. For each individual, we also knew whether or not the individual listens to any of these 21 radios, for each interval of $7.5$ minutes on a typical day. This led to a data set with $p = 3,882$ variables, among which $3,839$ were binary. 

From the data set, we selected a single sample of size $n=4,000$  according to a stratified sampling design with 5 strata, each stratum corresponding to a French region:
North-East, North-West, Île-de-France, South-east and South-West. The strata sample sizes were determined according to proportional allocation.  We considered the following domains of interest: the sub-population of individuals who connects to the internet everyday, almost every day, once or twice per week, once to three times per month, very rarely, never, the sub-population of individuals with/without children, the sub-populations of individuals living in cities of size (less than $20,000,$ between $20,000$ and $50,000,$  between $50,000$ and $100,000,$ between $100,000$ and $200,000$ and larger than $200,000$) and the sub-population of individuals of size 1, 2, 3, 4, 5 and larger than 5.

We computed the following estimates both at the overall level and at the domain level: (i) The Horvitz-Thompson estimator; (ii) the GREG estimator and (iii) the model-assisted estimator based on RF with hyper-parameters $B=1,000$, $n_0 = \floor {n^{11/20}}$ and $m_{try} = \sqrt{p}$. The working models used for the GREG estimator and the model-assisted RF estimator included $3882$ explanatory variables. In each scenario, we also computed a 95\% confidence interval for the proportion in the population of individuals who listen to the radio of interest. Finally, we computed the ratio of the estimated variances, using the estimated variance of the Horvitz-Thompson estimator, as the reference. Note that the "true value" was known for each domain of interest. The results are given in Table \ref{tab1}; for readability purpose, the results are expressed in percentage.


From Table \ref{tab1}, we note that the Horvitz-Thompson estimator performed relatively well in most scenarios. Because of the large number of predictors, the GREG estimator suffered from significant small sample bias. For instance, the estimate based on the GREG estimate at the overall level was equal to 27.7\%, far from the true value of about 13.5\%. In terms of point estimation,  the model-assisted estimator based on RF led to very similar results than those obtained with the Horvitz-Thompson estimator. However, RF led to substantial improvement in terms of estimated variance. Indeed, out of the 22 domains, the value of RV(RF) was smaller than 0.65 for 20 domains. The results suggest that, unlike the GREG estimator, the model-assisted estimator based on RF was not affected by the large number of explanatory variables in the working model.
The median length of the confidence intervals was equal to 5.4\% for the Horvitz-Thompson estimator, 4.2\% for the RF estimator and 3.8\% for the GREG estimator.


\begin{landscape}

\begin{table}[h!]
	\begin{tabular}{l|c|ccc|cc|ccc}
		\hline
		Domain                          & \multicolumn{1}{l|}{True value} & \multicolumn{1}{l}{HT} & \multicolumn{1}{l}{RF} & \multicolumn{1}{l|}{GREG} & \multicolumn{1}{l}{RV (RF)} & \multicolumn{1}{l|}{RV (GREG)} & \multicolumn{1}{l}{CI HT} & \multicolumn{1}{l}{CI RF} & \multicolumn{1}{l}{CI GREG} \\ \hline
		\multicolumn{1}{c|}{}           &                                 &                        &                        &                           &                           &                              &                           &                           &                             \\
		Overall                         & 13.5                           & 13.7                  & 13.6                  & 27.7                     & 64.6                     & 48.5                        & {[}12.7; 14.7{]}        & {[}12.8; 14.4{]}        & {[}27.0; 28.3{]}          \\
		Freq: Every day                 & 14.6                           & 14.6                  & 14.4                  & 28.1                     & 61.3                     & 44.0                        & {[}13.5; 15.8{]}        & {[}13.5; 15.3{]}        & {[}27.3; 28.9{]}          \\
		Freq: Almost every day          & 14.0                           & 15.5                  & 14.9                  & 31.8                     & 62.0                     & 52.7                        & {[}11.1; 19.8{]}        & {[}11.4; 18.3{]}        & {[}28.6; 35.0{]}          \\
		Freq: 1-2 / week                & 8.7                           & 10.5                  & 11.2                  & 55.1                     & 67.0                     & 62.9                        & {[}6.1; 14.9{]}        & {[}7.6; 14.7{]}        & {[}51.6; 58.5{]}          \\
		Freq: 1-3 / month               & 11.3                           & 16.1                  & 12.8                  & 3.8                     & 47.6                     & 28.1                        & {[}5.8; 26.4{]}        & {[}5.6; 19.9{]}        & {[}-1.5; 9.3{]}         \\
		Freq: Very rarely               & 6.7                           & 7.6                  & 10.0                  & 29.6                     & 60.5                     & 57.7                        & {[}1.4; 13.8{]}        & {[}5.2; 14.8{]}        & {[}24.9; 34.3{]}          \\
		Freq: Never                     & 8.4                           & 2.5                  & 8.1                  & -66.1                    & 103                     & 125                        & {[}-2.0; 7.0{]}       & {[}3.4; 12.7{]}        & {[}-71.2; -61.1{]}        \\
		Children: Yes                   & 11.1                           & 11.5                  & 11.4                  & 33.0                     & 57.2                     & 40.9                        & {[}9.7; 13.3{]}        & {[}10.0; 12.7{]}        & {[}31.9; 34.2{]}          \\
		Children: No                    & 14.4                           & 14.5                  & 14.5                  & 25.7                     & 62.3                     & 47.4                        & {[}13.3; 15.7{]}        & {[}13.5; 15.4{]}        & {[}24.8; 26.5{]}          \\
		Inhabitants: Country            & 12.7                           & 12.8                  & 12.6                  & 45.4                     & 58.9                     & 42.6                        & {[}10.5; 15.1{]}        & {[}10.8; 14.3{]}        & {[}44.0; 46.9{]}          \\
		Inhabitants: \textless{}20K     & 12.9                           & 12.7                  & 12.1                  & 35.2                     & 52.6                     & 37.2                        & {[}10.1; 15.3{]}        & {[}10.2; 14.0{]}        & {[}33.6; 36.8{]}          \\
		Inhabitants: 20-50K             & 12.5                           & 10.8                  & 11.0                  & 36.9                     & 62.2                     & 52.4                        & {[}6.7; 14.9{]}        & {[}7.8; 14.3{]}        & {[}33.9; 39.8{]}          \\
		Inhabitants: 50-100K            & 12.3                           & 11.5                  & 12.2                  & 25.9                     & 58.6                     & 47.7                        & {[}8.8; 14.2{]}        & {[}10.1; 14.3{]}        & {[}24.0; 27.8{]}          \\
		Inhabitants: 100-200K           & 15.5                           & 18.0                  & 17.9                  & 5.8                     & 53.4                     & 38.5                        & {[}14.0; 22.0{]}        & {[}15.0; 20.8{]}        & {[}3.3; 8.2{]}          \\
		Inhabitants: \textgreater{}200K & 14.2                           & 14.0                  & 14.0                  & 27.3                     & 60.4                     & 46.1                        & {[}12.1; 16.0{]}        & {[}12.5; 15.6{]}        & {[}26.0; 28.7{]}          \\
		Inhabitants: Paris              & 14.3                           & 16.3                  & 16.0                  & 0.6                     & 60.7                     & 45.3                        & {[}13.0; 19.6{]}        & {[}13.4; 18.6{]}        & {[}-1.5; 2.8{]}         \\
		N Household: 1                  & 13.9                           & 13.1                  & 13.5                  & 8.2                     & 58.5                     & 44.8                        & {[}11.1; 15.1{]}        & {[}12.0; 15.1{]}        & {[}6.8; 9.5{]}          \\
		N Household: 2                  & 16.3                           & 15.9                  & 16.1                  & 27.3                     & 58.1                     & 45.7                        & {[}14.0; 17.8{]}        & {[}14.6; 17.5{]}        & {[}26.1; 28.6{]}          \\
		N Household: 3                  & 11.7                           & 14.2                  & 13.4                  & 33.9                     & 55.3                     & 37.3                        & {[}11.4; 15.4{]}        & {[}11.5; 16.9{]}        & {[}32.3; 35.6{]}          \\
		N Household: 4                  & 11.3                           & 12.7                  & 11.8                  & 50.3                     & 56.5                     & 41.0                        & {[}10.2; 15.2{]}        & {[}9.9; 13.7{]}        & {[}48.7; 51.9{]}          \\
		N Household: 5                  & 9.7                           & 8.2                  & 9,1                  & 28.4                     & 63.1                     & 46.2                        & {[}4.9; 11.6{]}        & {[}6.5; 11.8{]}        & {[}26.1; 30.6{]}          \\
		N Household: \textgreater{}5    & 6.3                           & 5.0                  & 3.9                  & 50.9                     & 53.8                     & 26.8                        & {[}0.0; 9.5{]}         & {[}0.0; 7.2{]}         & {[}48.5; 53.2{]}          \\
		\multicolumn{1}{c|}{}           &                                 &                        &                        &                           &                           &                              &                           &                           &                             \\ \hline
	\end{tabular}
	\caption{Point and relative variance estimates for the percentage of household who listen to the radio of interest for twenty-two domain of interest and associated 95\% confidence intervals}
	\label{tab1}
\end{table}
\end{landscape}

\section{Final remarks}\label{conclusion}
In this paper, we have introduced a new class of model-assisted estimators based on random forests and derived corresponding variance estimators. We have established the theoretical properties of point and variance estimators obtained through a RF algorithm based on subsampling. The results of an empirical study suggest  that the proposed estimators perform well in a wide variety of settings, unlike the GREG and CART estimators. In practice, this robustness property is especially attractive when the data and the underlying relationships are complex. The application on radio audience data recorded by the French company M\'ediam\'etrie showed that the RF proposed estimator performed well in this high-dimension setting. We have also described a model calibration procedure for handling multiple survey variables, yet producing a single set of weights, which is attractive from a data user's perspective.

In practice, virtually all survey face the problem of missing values. Survey statisticians distinguish unit nonresponse (when no information is collected on a sampled unit) from item nonresponse (when the absence of information is limited to some variables only). The treatment of unit nonresponse starts with postulating a nonresponse model describing the relationship between the response indicators (equal to 1 for respondents and 0 for nonrespondents) and a vector of explanatory variables. The treatment of item nonresponse starts with postulating an imputation model describing the relationship between the variable requiring imputation and a set of explanatory variables. In both unit and item nonresponse, determining a suitable model is crucial. Therefore, regression trees and RF may prove useful for obtaining accurate estimated response propensities and predicted values. To the best of our knowledge, a theoretical treatment of regression trees and RF in the context of either unit nonresponse or item nonresponse in a finite population setting is lacking. These topics are currently under investigation.

Traditionally, survey samples are collected by probability sampling procedures and inferences are conducted with respect to the customary design-based framework. In recent years, there has been a shift of paradigm that can be explained by three main factors: (i) the dramatic decrease of response rates; (ii) the rapid increase in data collection costs; and (iii) the proliferation of nonprobabilistic data sources (e.g., administrative files, web survey panels, social media data, satellite information, etc.). To meet these new challenges, survey statisticians face increasing pressure to utilize these convenient but often uncontrolled data sources.  While such sources provide timely data for a large number of variables and population elements, they often fail to represent the target population of interest because of inherent selection biases. The integration of data from a nonprobability source to data from a probability survey is a topic that is currently being scrutinized by National Statistical Offices. An approach to data integration is statistical matching or mass imputation. Again, regression trees and RF algorithms may be useful techniques in the context of integration of survey data. This topic is currently under investigation. 

In a high-dimensional setting, RF may be used to selecting the most predictive predictors, which in turn may be used in the construction of model-assisted estimators of population totals/means. In this context, issues such as variable selection bias \citep{strobl_2007} in a finite population setting need to be investigated. This will be treated elsewhere.


%

\section*{Appendix}
\small{
\textbf{Proof of Proposition \ref{propr_weight_popRF}}
Since 
\begin{equation}
\widetilde{W}_{\ell} (\bx_k ) = \dfrac{1}{B} \sum_{b=1}^B \dfrac{\psi_\ell^{(b, U)} \mathds{1}_{\bx_\ell\in A^{(U)} \left(\boldsymbol{\mathbf{x}_k} , \theta_b^{(U)}\right)}}{{\widetilde N}(\bx_k, \theta_b^{(U)})  },
\end{equation}
involves positive quantities only, the weights $\widetilde{W}_{\ell} (\bx_k )$ are nonnegative. Since $\psi_\ell^{(b, U)}\in \{0,1\}$ for all $\ell \in U$ and for all $b \in 1, 2, \ldots, B,$  the weight can be bounded as follows:
\begin{align*}
\widetilde{W}_{\ell} (\bx_k ) = \dfrac{1}{B} \sum_{b=1}^B \dfrac{\psi_\ell^{(b, U)} \mathds{1}_{\bx_\ell\in A^{(U)} \left(\boldsymbol{\mathbf{x}_k} , \theta_b^{(U)}\right)}}{{\widetilde N}(\bx_k, \theta_b^{(U)})  } 	&\leqslant \dfrac{1}{B} \sum_{b=1}^B \left(\widetilde N\left(\bx_k , \theta_b^{(U)} \right) \right)^{-1} \\
	&\leqslant  c  N_0^{-1}.
	\end{align*}
	where $ c$ does not depend on $b$ nor on $ k$ or $\ell$. 
To show  ii), fix $b \in 1, 2, ..., B$. The result follows by noting that $\widetilde{W}_{\ell} (\bx_k )=  \left(\widetilde N \left(\bx_k , \theta_b^{(U)} \right) \right)^{-1}$ exactly $\widetilde N \left(\bx_k , \theta_b^{(U)} \right) $ times.

\vspace{0.5cm}

\noindent\textbf{Proof of Proposition \ref{bagging}} Let $\{\widehat{t}^{(b)}\}$ be a sequence of estimators of $t_y$. Then,
\begin{align*}
\V_p \left(\dfrac{1}{B} \sum_{b=1}^B \widehat{t}^{(b)}\right) 
&= \dfrac{1}{B^2}  \sum_{b=1}^B \left(\V_p \left( \widehat{t}^{(b)} \right) + \sum_{b=1}^B \sum_{b\neq b'=1}^B \mathbb{C}or_p \left(\widehat{t}^{(b)},  \widehat{t}^{(b')}\ \right) \V^{1/2} _p \left( \widehat{t}^{(b)} \right) \V^{1/2} _p \left( \widehat{t}^{(b')} \right)\right).  \\
& \leqslant  \dfrac{\V_p(\widehat{t}^{(1)})}{B} +  \V_p(\widehat{t}^{(1)})\max_{b \neq b'} \left|\mathbb{C}or_p \left(\widehat{t}^{(b)},  \widehat{t}^{(b')}\right)\right|
\end{align*}
and $\mbox{Biais}^2_p(B^{-1}\sum_{b=1}^B\widehat{t}^{(b)})=\mbox{Biais}^2_p(\widehat t^{(1)})=MSE_p(\widehat{t}^{(1)})-\V_p(\widehat{t}^{(1)}).$ So, for $B$ large enough:
\begin{eqnarray*}
MSE_p\left(\dfrac{1}{B} \sum_{b=1}^B \widehat{t}^{(b)}\right)\leqslant \V_p(\widehat{t}^{(1)})\max_{b \neq b'} \left|\mathbb{C}or_p \left(\widehat{t}^{(b)},  \widehat{t}^{(b')}\right)\right|-\V_p(\widehat{t}^{(1)})+MSE_p(\widehat{t}^{(1)}).
\end{eqnarray*}

\noindent\textbf{Proof of Proposition \ref{prop_projection}}
Consider the $B$ partitions build at the sample level $\widehat{\mathcal{P}}_S=\{\widehat{\mathcal{P}}_S^{(b)}\}_{b=1}^B.$ For a given $b=1, \ldots, B,$  the partition $\widehat{\mathcal{P}}_S^{(b)}$ is composed by disjointed regions as follows $\widehat{\mathcal{P}}_S^{(b)}=\{A^{(bS)}_{j}\}_{j=1}^{J_{bS}}$ and for each $b,$ consider the $J_{bS}$ dimensional vector $\hat{\mathbf{z}}^{(b)}_k=\left(\mathds{1}_{\bx_k\in A^{(bS)}_{1}}, \ldots,\mathds{1}_{\bx_k\in A^{(bS)}_{J_{bS}}}\right)^{\top}$ where $\mathds{1}_{\bx_k\in A^{(bS)}_{j}}=1$ if  $\bx_k$ belongs to the region ${A}^{(bS)}_{j}$ and zero otherwise for all $j=1, \ldots, J_{bS}$. Since  $\{A^{(bS)}_{j}\}_{j=1}^{J_{bS}}$ is a partition, then $\bx_k$ will belong to only one region and so, the vector $\hat{\mathbf{z}}^{(b)}_k$ will contain only one non zero component. We have $\widehat{m}_{rf}(\mathbf{x}_k)=B^{-1} \sum_{b = 1}^B \widehat{m}^{(b)}_{tree}(\mathbf{x}_k, \theta_b^{(S)})$ and $\widehat{m}^{(b)}_{tree}(\mathbf{x}_k, \theta_b^{(S)})$ can be written as $\widehat{m}^{(b)}_{tree}(\mathbf{x}_k, \theta_b^{(S)})=(\hat{\mathbf{z}}^{(b)}_k)^{\top}\widehat{\boldsymbol{\beta}}^{(b)}$
where $\widehat{\boldsymbol{\beta}}^{(b)}=\left(\sum_{\ell\in S}\pi^{-1}_{\ell}\psi_{\ell}^{(b,S)}\hat{\mathbf{z}}^{(b)}_{\ell}(\hat{\mathbf{z}}^{(b)}_{\ell})^{\top}\right)^{-1}\sum_{\ell\in S}\pi^{-1}_{\ell}\psi_{\ell}^{(b,S)}\hat{\mathbf{z}}^{(b)}_{\ell}y_{\ell}$  (see also the supplementary materiel for more details). Now,
\begin{eqnarray*}
\sum_{k\in S}\frac{y_k-\widehat m_{rf}(\bx_k)}{\pi_k}=\frac{1}{B}\sum_{b=1}^B\sum_{k\in S}\frac{y_k-\widehat{m}^{(b)}_{tree}(\mathbf{x}_k, \theta_b^{(S)})}{\pi_k}&=&\frac{1}{B}\sum_{b=1}^B\sum_{k\in S}\frac{(1-\psi^{(b,S)}_{k})(y_k-\widehat{m}^{(b)}_{tree}(\mathbf{x}_k, \theta_b^{(S)}))}{\pi_k}\\
&&+\frac{1}{B}\sum_{b=1}^B\sum_{k\in S}\frac{\psi^{(b,S)}_{k}(y_k-\widehat{m}^{(b)}_{tree}(\mathbf{x}_k, \theta_b^{(S)}))}{\pi_k}.
\end{eqnarray*}
For each $b,$ consider the $J_{bS}$ dimensional vector $\mathbf{1}_{J_{bS}}$ whose elements are all equal to one, we have then $\mathbf{1}^{\top}_{J_{bS}}\hat{\mathbf{z}}^{(b)}_{k}=1$ for all $k,$ so
$$
\sum_{k\in S}\frac{\psi^{(b,S)}_{k}}{\pi_k}\widehat{m}^{(b)}_{tree}(\mathbf{x}_k, \theta_b^{(S)})=\sum_{k\in S}\frac{\psi^{(b,S)}_{k}}{\pi_k}(\hat{\mathbf{z}}^{(b)}_k)^{\top}\widehat{\boldsymbol{\beta}}^{(b)}=\mathbf{1}^{\top}_{J_{bS}}\sum_{k\in S}\frac{\psi^{(b,S)}_{k}}{\pi_k}\hat{\mathbf{z}}^{(b)}_{k}(\hat{\mathbf{z}}^{(b)}_k)^{\top}\widehat{\boldsymbol{\beta}}^{(b)}=\sum_{\ell\in S}\frac{\psi^{(b,S)}_{\ell}}{\pi_{\ell}}y_{\ell}.
$$
\bibliographystyle{apalike}
\bibliography{biblio}
\addcontentsline{toc}{section}{References}

\end{document}


\baselineskip .3in

\title{\bf {\Large Model-assisted estimation through random forests in finite population sampling\\
Supplementary material}}

\author{Mehdi {\sc Dagdoug}$^{(a)}$,  Camelia {\sc Goga}$^{(a)}$ and  David  {\sc Haziza}$^{(b)}$ \\
	(a) Universit\'e de Bourgogne Franche-Comt\'e,\\ Laboratoire de Math\'ematiques de Besan\c con,  Besan\c con, FRANCE \\
	(b) University of Ottawa, Department of mathematics and statistics,\\ Ottawa, CANADA
}

\date{\today}
\maketitle


\section{Asymptotic assumptions}
\subsection*{Assumptions: population-based RF model-assisted estimator $\hat t^*_{rf}$} \label{sec1}

To establish the properties of the proposed estimators, we will consider three categories of assumptions: assumptions on the sampling design, assumptions on the survey variable and, finally, assumptions on the random forests. 


\hypH{We assume that there exists a positive constant $C$ such that  $ \sup_{k\in U_v}|y_k| \leqslant C<\infty.  $\label{H1}}

\hypH{We assume  that $ \lim\limits_{v \to \infty}\dfrac{n_v}{N_v} = \pi \in (0;1)$.\label{H2}}

\hypH{ There exist positive constants $\lambda$ and $\lambda^*$ such that $\min\limits_{k \in U_v} \pi_k \geqslant \lambda > 0,$ $\min\limits_{k,\ell \in U_v} \pi_{k\ell} \geqslant \lambda^*> 0$ and $\limsup\limits_{v \to \infty} n_v \max\limits_{k \neq \ell \in U_v} | \pi_{k\ell} - \pi_k \pi_\ell| < \infty$.
\label{H3}}




\hypC{The number of subsampled elements $N'_v$ is such that $\lim_{v \to \infty} N'_v/N_v \in (0; 1]$.\label{C1}}

\subsection*{Assumptions: sample-based  RF model-assisted estimator $\hat t_{rf}$} 
 \hypH{We assume  that there exists a positive constant $C_1>0$ such that $n_v\max_{k\neq \ell\in U_v}\left|\E_p\left\{(I_k-\pi_k)(I_{\ell}-\pi_{\ell})|\mathcal{\widehat{P}}_S\right\}\right|\leq \displaystyle C_1$
 \label{H5}}

\hypH{The random forests based on population partitions and those based on sample partitions are such that, for all $\bx \in \mathbb R^p:$ 
$$
\E_p\left(\widehat{\widetilde m}_{rf}(\bx)-\widetilde m_{rf}(\bx)\right)^2=o(1).
$$
where $\widehat{\widetilde m}_{rf}(\bx)$ is given by 
$$
\widehat{\widetilde m}_{rf}(\bx)=\sum_{\ell\in U_v}\frac{1}{B}\sum_{b=1}^B\dfrac{\psi_\ell^{(b, S)} \mathds{1}_{\bx_\ell\in  A^{(S)} \left(\boldsymbol{\mathbf{x}} , \theta_b^{(S)}\right)}}{{\widehat{\widetilde N}}(\bx, \theta_b^{(S)})}y_{\ell}
$$
with ${\widehat{\widetilde N}}(\bx, \theta_b^{(S)})=\sum_{k\in U_v}\psi_k^{(b, S)} \mathds{1}_{\bx_k\in  A^{(S)} \left(\boldsymbol{\mathbf{x}} , \theta_b^{(S)}\right)}$ and 
$$
\widetilde m_{rf}(\bx)=\sum_{\ell\in U_v}\frac{1}{B}\sum_{b=1}^B\dfrac{\psi_\ell^{(b, U)} \mathds{1}_{\bx_\ell\in  A^{(U)} \left(\boldsymbol{\mathbf{x}} , \theta_b^{(U)}\right)}}{{\widetilde N}(\bx, \theta_b^{(U)})}y_{\ell}
$$
with ${\widetilde N}(\bx, \theta_b^{(U)})=\sum_{k\in U_v}\psi_k^{(b, U)} \mathds{1}_{\bx_k\in  A^{(U)} \left(\boldsymbol{\mathbf{x}} , \theta_b^{(U)}\right)}.$
\label{H6}}

Below, we include a graph illustrating the convergence of the difference $\widehat{\widetilde m}_{rf} - \widetilde m_{rf}$ towards $0$ in $L^2$ where the regression function was defined as $m(X)  = 2+ 2 X_1 + X_2 + X_3$, with $X_1, X_2$ and $X_3$ defined as in Section 6 from the main paper. The population sizes were such that the sampling fraction was of $10\%$. Similar results can be obtained using other simulation parameters. \\

\begin{figure}[h!]
	\centering
	
	\includegraphics[width=1\linewidth]{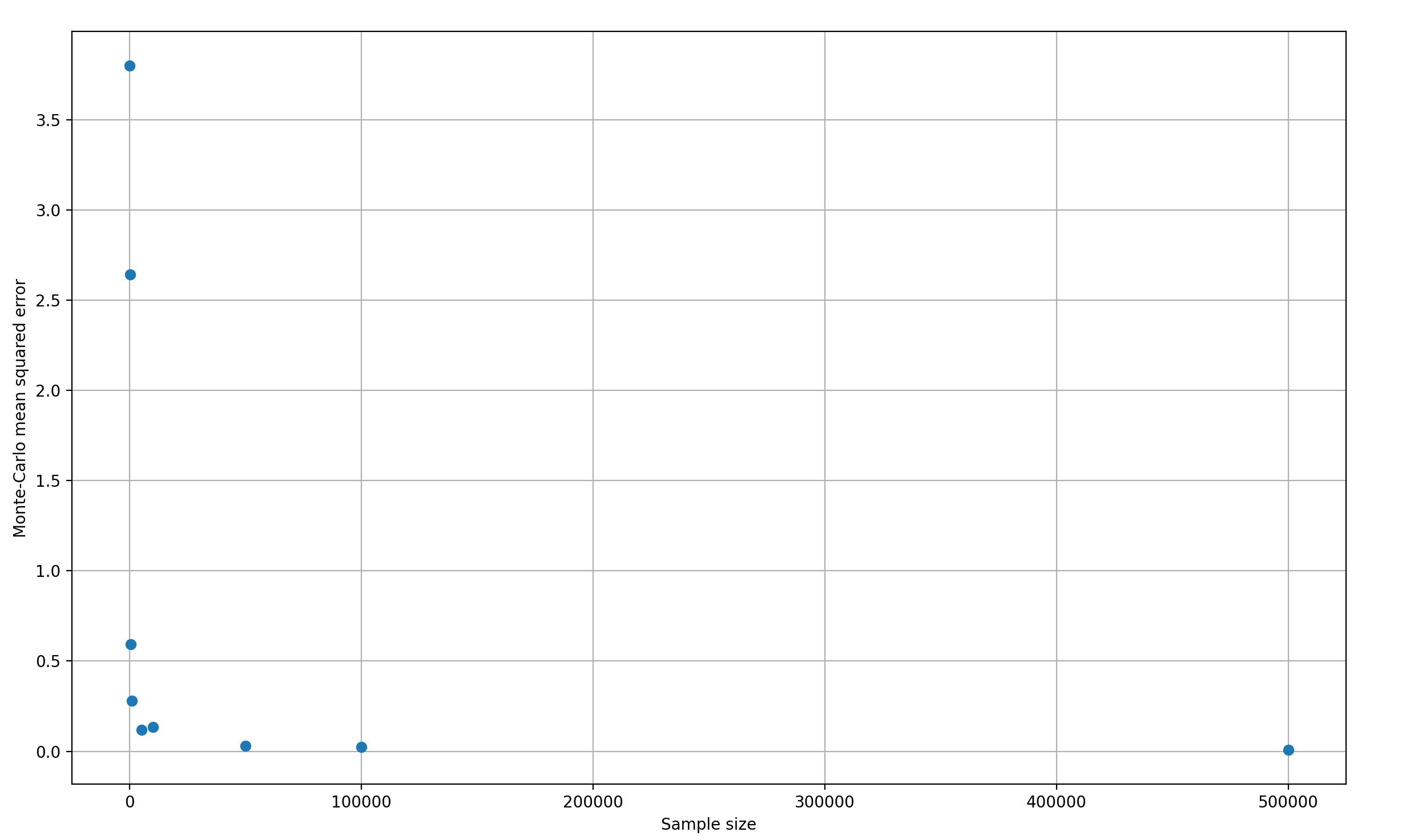}

\end{figure}	

\hypC{The number of subsampled elements $n'_v$ is such that $\lim_{v \to \infty} n'_v/n_v \in (0; 1]$.\label{C2}}

\subsection*{Consistency of the Horvitz-Thompson variance estimator}
\hypH{Assume that 
	$ \lim\limits_{v \to \infty} \max\limits_{i,j,k,\ell \in D_{4, N_v}} \rvert \mathbb{E}_p \left\{ \left(I_iI_j - \pi_i\pi_j\right) \left(I_kI_\ell - \pi_k \pi_\ell\right) \right\}\rvert =0,$
	where $D_{4,N_v}$ denotes the set of distinct $4$-tuples from $U_v.$ \label{H4}}


\section{Asymptotic Results}
\subsection{Asymptotic results of the population RF model-assisted estimator $\widehat{t}^*_{rf}$} \label{sec2}

The population RF model-assisted estimator is given by
$$
\widehat{t}^*_{rf}=\sum_{k\in U_v}\widehat{m}^*_{rf}(\bx_k)+\sum_{k\in S_v}\frac{y_k-\widehat{m}^*_{rf}(\bx_k)}{\pi_k},
$$
where $\widehat{m}^*_{rf}$ is the sample-based estimator of $m$ by using RF built at the population level (for more details, see relation (13) from the main paper):  
\begin{eqnarray}
\widehat{m}^*_{rf}(\bx_k)=\sum_{\ell\in S_v}\frac{1}{\pi_{\ell}}\widetilde{W}^*_{\ell}(\bx_k)y_{\ell},\label{m^*_poids}
\end{eqnarray}
where 
$$
\widetilde{W}^*_{\ell}(\bx_k)=\frac{1}{B}\sum_{b=1}^B\dfrac{\psi_\ell^{(b, U)} \mathds{1}_{\bx_{\ell}\in A^{*(U)} \left(\boldsymbol{\mathbf{x}_k} , \theta_b^{(U)}\right)} }{{\widetilde N}^*(\bx_k, \theta_b^{(U)})}
$$
and $\widetilde{N}^*(\bx_k, \theta_b^{(U)})=\sum_{\ell\in U_v}\psi_\ell^{(b, U)} \mathds{1}_{\bx_{\ell}\in A^{*(U)} \left(\boldsymbol{\mathbf{x}_k} , \theta_b^{(U)}\right)}$ is the number of units falling in the terminal node  $A^{*(U)} \left(\boldsymbol{\mathbf{x}_k} , \theta_b^{(U)}\right)$ containing $\bx_k. $ The estimator $\widehat{m}^*_{rf}(\bx_k)$ can be written as a bagged estimator as follows:
$$
\widehat{m}^*_{rf}(\bx_k)=\frac{1}{B}\sum_{b=1}^B\widehat{m}_{tree}^{*(b)}(\bx_k,\theta_b^{(U)}),
$$
where $\widehat{m}^{*(b)}_{tree}(\bx_k,\theta_b^{(U)})$ is the sample-based estimation of $m$ based on the $b$-th stochastic tree:
\begin{eqnarray}
\widehat{m}_{tree}^{*(b)}(\bx_k,\theta_b^{(U)})=\sum_{\ell\in S_v}\frac{1}{\pi_{\ell}}\dfrac{\psi_\ell^{(b, U)} \mathds{1}_{\bx_{\ell}\in  A^{*(U)} \left(\boldsymbol{\mathbf{x}_k} , \theta_b^{(U)}\right)}}{{\widetilde N}^*(\bx_k, \theta_b^{(U)})}y_{\ell}.\label{hat_m*_1}
\end{eqnarray}
 For more readability, we will use in the sequel $\widehat{m}_{tree}^{*(b)}(\bx_k)$ instead of $\widehat{m}_{tree}^{*(b)}(\bx_k,\theta_b^{(U)}).$
Consider the pseudo-generalized difference estimator:
$$
\widehat{t}_{pgd}= \sum_{k\in U_v} \widetilde{m}^*_{rf}(\bx_k)+\sum_{k\in S_v}\frac{y_k- \widetilde{m}^*_{rf}(\bx_k)}{\pi_k},
$$
where $\widetilde{m}^*_{rf}(\bx_k)$ is the population-based estimator of $m$ by using RF built at the population level (for more details, see relation (12) from the main paper):
$$
\widetilde{m}^*_{rf}(\bx_k)=\sum_{\ell\in U_v}\widetilde{W}^*_{\ell}(\bx_k)y_{\ell}.
$$
The estimator $\widetilde{m}^*_{rf}$ can be written as a bagged estimator as follows:
$$
\widetilde{m}^*_{rf}(\bx_k)=\frac{1}{B}\sum_{b=1}^B\widetilde{m}_{tree}^{*(b)}(\bx_k)
$$
and $\widetilde{m}_{tree}^{*(b)}(\bx_k)$ is the predictor associated with unit $k$ and based on the $b$-th stochastic tree: 
\begin{eqnarray}
\widetilde{m}_{tree}^{*(b)}(\bx_k)=\sum_{\ell\in U_v}\dfrac{\psi_\ell^{(b, U)} \mathds{1}_{\bx_{\ell}\in A^{*(U)} \left(\boldsymbol{\mathbf{x}_k} , \theta_b^{(U)}\right)} }{{\widetilde N}^*(\bx_k, \theta_b^{(U)})}y_{\ell}.\label{tilde_m*_1}
\end{eqnarray}
We remark that $\widehat{m}_{tree}^{*(b)}(\bx_k)$ is the Horvitz-Thompson estimator of $\widetilde{m}_{tree}^{*(b)}(\bx_k). $ As before, $\widetilde{m}_{tree}^{*(b)}$ depends on $\theta_b^{(U)}$ but, for more readability, we drop $\theta_b^{(U)}$ from the expression of  $\widetilde{m}_{tree}^{*(b)}(\bx_k). $

We give in the next equivalent expressions of $\widetilde{m}_{tree}^{*(b)}$ and $\widehat{m}_{tree}^{*(b)}.$ Consider for that the $B$ partitions built at the population level: $\widetilde{\mathcal{P}}^*_U=\{\widetilde{\mathcal{P}}_U^{*(b)}\}_{b=1}^B.$ For a given $b=1, \ldots, B,$ the partition $\widetilde{\mathcal{P}}_U^{*(b)}$ build in the $b$-th stochastic tree is composed by the $J^*_{bU}$ disjointed regions: $\widetilde{\mathcal{P}}_U^{*(b)}=\{A^{*(bU)}_j\}_{j=1}^{J^*_{bU}}.$ Consider $\mathbf{z}^{*(b)}_k=( \mathds{1}_{\bx_k \in A^{*(bU)}_1}, \ldots,\mathds{1}_{\bx_k \in A^{*(bU)}_{J^*_{bU}}})^{\top}$ where $\mathds{1}_{\bx_k \in {A}^{*(bU)}_j}=1$ if  $\bx_k$ belongs to the region ${A}^{*(bU)}_j$ and zero otherwise for all $j=1, \ldots, J^*_{bU}$. We drop the exponent $U$ from the expression of   $\mathbf{z}^{*(b)}_k$ for more readability. 
Since $\widetilde{\mathcal{P}}_U^{*(b)}$ is a partition, then $\bx_k$ belongs to only one region of the $b$-th tree, so the vector $\mathbf{z}^{*(b)}_k$ will contain only one non-null component. Consider for example that $\bx_k\in {A}^{*(bU)}_j,$ then $\widetilde{m}_{tree}^{*(b)}(\bx_k)$ is the mean of $y$-values of individuals $\ell$ for which $\bx_{\ell}\in {A}^{*(bU)}_j:$ 
$$
\widetilde{m}_{tree}^{*(b)}(\bx_k)=\sum_{\ell\in U_v}\frac{\psi_{\ell}^{(b,U)}\mathds{1}_{\bx_{\ell} \in {A}^{*(bU)}_j}}{\widetilde N^{*(b)}_{j}}y_{\ell}, \quad\mbox{for}\quad \bx_k\in {A}^{*(bU)}_j,
$$
where $\widetilde N^{*(b)}_{j}$ is the number of units belonging to the region ${A}^{*(bU)}_j:$
\begin{eqnarray}
\widetilde N^{*(b)}_{j}=\sum_{\ell\in U_v}\psi_{\ell}^{(b,U)}\mathds{1}_{\bx_{\ell} \in {A}^{*(bU)}_j}, \quad j=1, \ldots, J^*_{bU}.\label{effectif_partition_pop}
\end{eqnarray}
Then, $\widetilde{m}_{tree}^{*(b)}(\bx_k)$ can be written as follows:
\begin{eqnarray}
\widetilde{m}_{N, rf}^{*(b)}(\bx_k)=(\mathbf{z}^{*(b)}_k)^{\top}\widetilde{\boldsymbol{\beta}}^{*(b)}, \quad k\in U_v\label{expr_tilde_m_pop}
\end{eqnarray}
where 
$$
\widetilde{\boldsymbol{\beta}}^{*(b)}=\left(\sum_{\ell\in U_v}\psi_{\ell}^{(b,U)}\mathbf{z}^{*(b)}_{\ell}(\mathbf{z}^{*(b)}_{\ell})^{\top}\right)^{-1}\sum_{\ell\in U_v}\psi_{\ell}^{(b,U)}\mathbf{z}^{*(b)}_{\ell}y_{\ell}.
$$
Remark that $\widetilde{\boldsymbol{\beta}}^{*(b)}$ may be obtained as solution of the following weighted estimating equation:
$$
\sum_{\ell\in U_v}\psi_{\ell}^{(b,U)}\mathbf{z}^{*(b)}_{\ell}(y_{\ell}-(\mathbf{z}^{*(b)}_{\ell})^{\top}\boldsymbol{\beta}^{*(b)})=0.
$$
Since the regions ${A}^{*(bU)}_j, j=1, \ldots, J^*_{bU},$ form a partition, then the matrix $\sum_{\ell\in U_v}\psi_{\ell}^{(b,U)}\mathbf{z}^{*(b)}_{\ell}(\mathbf{z}^{*(b)}_{\ell})^{\top}$ is diagonal with diagonal elements equal to $\widetilde N^{*(b)}_{j},$ the number of units falling in  the region ${A}^{*(bU)}_j$ for all $j=1, \ldots, J^*_{bU}.$ By the stopping criterion, we have that all $\widetilde N^{*(b)}_{j}\geq N_{0v}>0$ for all $j,$ so the matrix $\sum_{\ell\in U_v}\psi_{\ell}^{(b,U)}\mathbf{z}^{*(b)}_{\ell}(\mathbf{z}^{*(b)}_{\ell})^{\top}$ is always invertible and $\widetilde{\boldsymbol{\beta}}^{*(b)}$ is well-defined.

Consider now $\widehat{m}_{tree}^{*(b)}(\bx_k),$ the estimator of the unknown $\widetilde{m}_{tree}^{*(b)}(\bx_k). $ Then, $\widehat{m}_{tree}^{*(b)}(\bx_k)$  is  the weighted mean of  $y$-values for sampled individuals $\ell$ belonging to the same region ${A}^{*(bU)}_j$ as unit $k:$ 
$$
\widehat{m}_{tree}^{*(b)}(\bx_k)=\sum_{\ell\in S_v}\frac{1}{\pi_{\ell}}\frac{\psi_{\ell}^{(b,U)}\mathds{1}_{\bx_{\ell} \in {A}^{*(bU)}_j}}{N^{*(b)}_{j,N}}y_{\ell}\quad \mbox{for}\quad \bx_k\in {A}^{*(bU)}_j
$$ 
and we can write:
\begin{eqnarray}
\widehat{m}_{tree}^{*(b)}(\bx_k)=(\mathbf{z}^{*(b)}_k)^{\top}\widehat{\boldsymbol{\beta}}^{*(b)}, k\in U_v\label{expr_hat_m_pop}
\end{eqnarray}
where 
$$
\widehat{\boldsymbol{\beta}}^{*(b)}=\left(\sum_{\ell\in U_v}\psi_{\ell}^{(b,U)}\mathbf{z}^{*(b)}_{\ell}(\mathbf{z}^{*(b)}_{\ell})^{\top}\right)^{-1}\sum_{\ell\in S_v}\frac{1}{\pi_{\ell}}\psi_{\ell}^{(b,U)}\mathbf{z}^{*(b)}_{\ell}y_{\ell}.
$$
In the expression of $\widehat{\boldsymbol{\beta}}^{*(b)}, $ we do not estimate the matrix $\sum_{\ell\in U_v}\psi_{\ell}^{(b,U)}\mathbf{z}^{*(b)}_{\ell}(\mathbf{z}^{*(b)}_{\ell})^{\top}$ since it is known and besides, we guarantee in this way that we will always have non-empty terminal nodes at the population level.  So, $\widehat{\boldsymbol{\beta}}^{*(b)}$ will be always well-defined whatever the sample $S$ is.

Let denote by $ \alpha_k = \pi_k^{-1}I_k-1$ for all $k \in U_v,$ where $I_k$ is the sample membership, $I_k=1$ if $k\in S$ and zero otherwise.   In order to prove the consistency of $\widehat{t}^*_{rf}$ as well as its asymptotic equivalence to the pseudo-generalized difference estimator $\widehat{t}_{pgd},$ we use the following decomposition:
\begin{align}
	\frac{1}{N_v}\left(\widehat{t}^*_{rf} - t_y\right)&=  \frac{1}{N_v}\left(\widehat{t}_{pgd} - t_y\right)  - \frac{1}{N_v}\sum_{k\in U_v} \alpha_k(\widehat{m}^*_{rf}(\bx_k) - \widetilde{m}^*_{rf} (\bx_k))\nonumber\\
	&=  \dfrac{1}{N_v} \left( \widehat{t}_{pgd} - t_y \right)-\dfrac{1}{B} \sum_{b=1}^B \bigg[  \dfrac{1}{N_v}  \sum_{k \in U_v} \alpha_k \bigg(\widehat{m}_{tree}^{*(b)}(\bx_k) - \widetilde{m}_{tree}^{*(b)}(\bx_k)\bigg) \bigg].\label{erreur_1}
	\end{align} 
We will prove that each term form the decomposition (\ref{erreur_1}) is convergent to zero. We give first several useful lemmas.

\begin{lemma}\label{ordre_t_diff_pop}
There exists a positive constant $\tilde c_1$ such that:
$$
\frac{n_v}{N^2_v}\mathbb{E}_p(\widehat{t}_{pgd} - t_y)^2\leqslant \tilde c_1.
$$
\end{lemma}
\begin{proof}
First of all, from relation (\ref{m^*_poids}), $\widetilde{m}^*_{rf}(\bx_k)$ is a weighted sum at the population level of $y$-values  with positive weights summing to one (see proposition 2.1. from the main paper). Then, we get that $\sup_{k\in U_v}|\widetilde{m}^*_{rf}(\bx_k)|\leqslant C$ by using also  assumption (H\ref{H1}). 
We have:
$$
\frac{1}{N_v}(\widehat{t}_{pgd} - t_y) = \frac{1}{N_v}\displaystyle\sum_{k \in U_v} \alpha_k(y_k - \widetilde{m}^*_{rf}(\bx_k))
$$ 
and
\begin{align*}
n_v\mathbb{E}_p \left(\dfrac{\widehat{t}_{pgd} - t_y}{N_v}\right)^2 &= \dfrac{n_v}{N_v^2} \V_p \left( \sum_{k \in S_v} \dfrac{(y_k - \widetilde{m}^*_{rf}(\bx_k))}{\pi_k} \right)  \\
&\leqslant \left( \dfrac{n_v}{N_v} \cdot \dfrac{1 }{\lambda} + \frac{n_v \max_{ k \neq \ell \in U_v} \rvert \pi_{k\ell}-\pi_k\pi_{\ell} \rvert}{\lambda^2}\right) \cdot \dfrac{2}{N_v} \sum_{k \in U_v} \left(y^2_k + (\widetilde{m}^*_{rf}(\bx_k))^2 \right)\\
&\leqslant \tilde c_1
\end{align*}
by assumptions (H\ref{H1})-(H\ref{H3}).
\end{proof}

\begin{lemma}\label{ordre_hat_beta*}
There exists a positive constant $\tilde c_2$ not depending on $b=1, \ldots, B,$ such that 
$$
\mathbb{E}_p ||\widehat{\boldsymbol{\beta}}^{*(b)}-\widetilde{\boldsymbol{\beta}}^{*(b)}||_2^2\leqslant  \frac{\tilde c_2N_v}{N^2_{0v}}\quad \mbox{for all}\quad b=1, \ldots, B.
$$
\end{lemma}

\begin{proof}
We can write
\begin{eqnarray}
\widehat{\boldsymbol{\beta}}^{*(b)}-\widetilde{\boldsymbol{\beta}}^{*(b)}&= & \left(\sum_{\ell\in U_v}\psi_{\ell}^{(b,U)}\mathbf{z}^{*(b)}_{\ell}(\mathbf{z}^{*(b)}_{\ell})^{\top}\right)^{-1}\left(\sum_{\ell\in S_v}\frac{1}{\pi_{\ell}}\psi_{\ell}^{(b,U)}\mathbf{z}^{*(b)}_{\ell}y_{\ell}-\sum_{\ell\in U_v}\psi_{\ell}^{(b,U)}\mathbf{z}^{*(b)}_{\ell}y_{\ell}\right)\nonumber\\
	&= & \left(\sum_{\ell\in U_v}\psi_{\ell}^{(b,U)}\mathbf{z}^{*(b)}_{\ell}(\mathbf{z}^{*(b)}_{\ell})^{\top}\right)^{-1}\left(\sum_{\ell\in U_v}\alpha_{\ell}\psi_{\ell}^{(b,U)}\mathbf{z}^{*(b)}_{\ell}y_{\ell}\right)
\end{eqnarray}
Let denote by $\widetilde{\mathbf{T}}^{*(b)}=\sum_{\ell\in U_v}\psi_{\ell}^{(b,U)}\mathbf{z}^{*(b)}_{\ell}(\mathbf{z}^{*(b)}_{\ell})^{\top}. $ As already mentioned before, the matrix $\widetilde{\mathbf{T}}^{*(b)}$ is diagonal with positive diagonal elements given by $\widetilde N^{*(b)}_{j}$ the number of units falling in  the region ${A}^{*(bU)}_j$ (see relation \ref{effectif_partition_pop}) for $j=1, \ldots, J^*_{bU}$	and by the stopping criterion, we have that $\widetilde N^{*(b)}_{j}\geq N_{0v}>0. $ We obtain then 
\begin{eqnarray}
||(\widetilde{\mathbf{T}}^{*(b)})^{-1}||_2=\max_{j=1, \ldots, J_{bU}}\left(\frac{1}{\widetilde N^{*(b)}_{j}}\right)\leq N^{-1}_{0v}, \quad \mbox{for all}\quad b=1, \ldots, B\label{ordre_T*}
\end{eqnarray}
where $||\cdot||_2$ is the spectral norm matrix  defined for a squared $p\times p$ matrix $\mathbf{A}$ by $||\mathbf A||_2=\sup_{\mathbf{x}\in \mathbf{R}^p, ||\mathbf x||_2\neq 0}||\mathbf A\mathbf x||_2/||\mathbf x||_2$. For a symmetric and positive definite matrix $\mathbf{A},$ we have that $||\mathbf A||_2=\lambda_{max}(\mathbf A)$ where $\lambda_{max}(\mathbf A)$ is the largest eigenvalue of $\mathbf A. $ We get, for $b=1, \ldots, B:$ 
\begin{eqnarray}
\mathbb{E}_p ||\widehat{\boldsymbol{\beta}}^{*(b)}-\widetilde{\boldsymbol{\beta}}^{*(b)}||_2^2&\leqslant &\mathbb{E}_p \bigg[ \rvert\rvert N_v(\widetilde{\mathbf{T}}^{*(b)})^{-1} \rvert \rvert_2^2 \cdot  \bigg\rvert \bigg\rvert  \dfrac{1}{N_v}\sum_{\ell\in U_v}\alpha_{\ell}\psi_{\ell}^{(b,U)}\mathbf{z}^{*(b)}_{\ell}y_{\ell}\bigg\rvert \bigg\rvert_2^2 \bigg]\nonumber\\
&\leqslant& \frac{N^2_v}{N^2_{0v}}\mathbb{E}_p\bigg\rvert \bigg\rvert  \dfrac{1}{N_v}\sum_{\ell\in U_v}\alpha_{\ell}\psi_{\ell}^{(b,U)}\mathbf{z}^{*(b)}_{\ell}y_{\ell}\bigg\rvert \bigg\rvert_2^2\label{expr1}
\end{eqnarray}
and
\begin{eqnarray}
&&	\mathbb{E}_p \bigg\rvert \bigg\rvert  \dfrac{1}{N_v}\sum_{k\in U_v}\alpha_{k}\psi_k^{(b,U)}\mathbf{z}^{*(b)}_{k}y_{k}\bigg\rvert \bigg\rvert_2^2\nonumber\\
	&=&\dfrac{1}{N^2_v}\bigg(\sum_{ k \in U_v} (\psi_{k}^{(b,U)})^2y^2_{k}\rvert \rvert \bz_k^{*(b)} \rvert\rvert_2^2\mathbb{E}_p (\alpha_k^2) +  \sum_{ k \in U_v} \sum_{ \substack{ \ell \in U_v \\ \ell \neq k}} \psi_k^{(b,U)}\psi_{\ell}^{(b,U)}y_ky_{\ell} (\bz_k^{*(b)})^{\top} \bz^{*(b)}_\ell \mathbb{E}_p (\alpha_k \alpha_{\ell})\bigg)\nonumber\\
	&\leqslant &\frac{1}{n_v}\bigg(\dfrac{n_v}{\lambda N_v}+  \dfrac{n_v \max_{k,\ell\in U_v,k\neq \ell} \rvert \pi_{k\ell } - \pi_k \pi_\ell\rvert }{\lambda^2} \bigg) \bigg(\dfrac{1}{N_v} \sum_{k \in U_v} (\psi^{(b,U)}_k)^2y_k^2||\bz^{*(b)}_k||_2^2\bigg)\nonumber\\
	&\leqslant& \frac{C_0}{n_v}\label{HT_z_y}
	\end{eqnarray}
by assumptions (H\ref{H1})-(H\ref{H3}) and the fact that $||\bz_k^{*(b)}||_2^2=1$ for all $k\in U_v$ and $b=1, \ldots, B.$ From (\ref{expr1}), (\ref{HT_z_y}) and assumption (H\ref{H1}), we obtain that it exists a positive constant $\tilde c_2$ such that 
$$
\mathbb{E}_p ||\widehat{\boldsymbol{\beta}}^{*(b)}-\widetilde{\boldsymbol{\beta}}^{*(b)}||_2^2\leqslant \frac{\tilde c_2N_v}{N^2_{0v}}.
$$

\end{proof}

\begin{result} \label{res1}
	Consider a sequence of population RF estimators $\{\widehat{t}^*_{rf}\}$. Then, there exist positive constants $\tilde C_1, \tilde C_2$  such that 
	\begin{equation*}
	\mathbb{E}_p  \bigg\rvert \dfrac{1}{N_v} \left(\widehat{t}^*_{rf} - t_y\right)\bigg\rvert  \leqslant \dfrac{\tilde C_1}{\sqrt{n_{v}}} + \dfrac{\tilde C_2}{N_{0v}}, \quad\mbox{with $\xi$-probability one.}
	\end{equation*}
If $\displaystyle\frac{N^{a}_v}{N_{0v}}=O(1)$ with $1/2\leqslant a \leqslant 1,$ then 
	\begin{equation*}
	\mathbb{E}_p  \bigg\rvert \dfrac{1}{N_v} \left(\widehat{t}^*_{rf} - t_y\right)\bigg\rvert  \leqslant \dfrac{\tilde C}{\sqrt{n_{v}}}, \quad\mbox{with $\xi$-probability one.}
	\end{equation*}
\end{result}
\begin{proof}
We get from relation (\ref{erreur_1}) :
	\begin{align*}
\dfrac{1}{N_v}\mathbb{E}_p  \bigg\rvert  \widehat{t}^*_{rf} - t_y\bigg\rvert  &\leqslant    \dfrac{1}{N_v}\mathbb{E}_p \bigg\rvert \widehat{t}_{pgd} - t_y  \bigg\rvert +  \dfrac{1}{B} \sum_{b=1}^B  \dfrac{1}{N_v} \mathbb{E}_p\bigg| \sum_{k\in U_v}\alpha_k(\widehat{m}^{*(b)}_{tree}(\bx_k) - \widetilde{m}_{tree}^{*(b)}(\bx_k))\bigg\rvert. 
	\end{align*}
Lemma \ref{ordre_t_diff_pop} gives us that there exists positive constant $\tilde C_1$ such that 
\begin{eqnarray}
\dfrac{1}{N_v}\mathbb{E}_p \bigg\rvert \widehat{t}_{pgd} - t_y  \bigg\rvert \leqslant \dfrac{\tilde C_1}{\sqrt{n_{v}}}. \label{ordre_t*_diff}
\end{eqnarray}
Now, by using relations (\ref{expr_tilde_m_pop}) and (\ref{expr_hat_m_pop}), we can then write  for any $b=1, \ldots, B$:
$$
 \sum_{k\in U_v}\alpha_k(\widehat{m}^{*(b)}_{tree}(\bx_k) - \widetilde{m}_{tree}^{*(b)}(\bx_k))= \sum_{k\in U_v}\alpha_k (\mathbf{z}^{*(b)}_k)^{\top}(\widehat{\boldsymbol{\beta}}^{*(b)}-\widetilde{\boldsymbol{\beta}}^{*(b)}) 
$$
and
\begin{eqnarray}
\dfrac{1}{N_v} \mathbb{E}_p\bigg| \sum_{k\in U_v}\alpha_k(\widehat{m}^{*(b)}_{tree}(\bx_k) - \widetilde{m}_{tree}^{*(b)}(\bx_k))\bigg\rvert\leqslant \bigg(\mathbb{E}_p  \bigg\rvert \bigg\rvert \dfrac{1}{N_v}\sum_{k\in U_v}\alpha_k \bz^{*(b)}_k\bigg\rvert \bigg\rvert_2^2 \bigg)^{1/2} \bigg(\mathbb{E}_p ||\widehat{\boldsymbol{\beta}}^{*(b)}-\widetilde{\boldsymbol{\beta}}^{*(b)}||_2^2\bigg)^{1/2}\nonumber\\\label{ordre_terme_2_1}
\end{eqnarray}
and 
\begin{eqnarray}
\mathbb{E}_p  \bigg\rvert \bigg\rvert \dfrac{1}{N_v}\sum_{k\in U_v}\alpha_k \bz^{*(b)}_k\bigg\rvert \bigg\rvert_2^2 &=& 
\dfrac{1}{N^2_v}\bigg(\sum_{ k \in U_v} \mathbb{E}_p (\alpha_k^2)\rvert \rvert \bz_k^{*(b)} \rvert\rvert_2^2 +  \sum_{ k \in U_v} \sum_{ \substack{ \ell \in U_v \\ \ell \neq k}} \mathbb{E}_p (\alpha_k \alpha_{\ell})(\bz_k^{*(b)})^{\top} \bz^{*(b)}_\ell \bigg)\nonumber\\
& \leqslant & \frac{1}{n_v}\left( \dfrac{n_v}{\lambda N_v}  + \frac{n_v \max_{ k \neq \ell \in U_v} \rvert \pi_{k\ell}-\pi_k\pi_{\ell}\rvert}{\lambda^2} \right) \cdot \dfrac{1}{N_v} \sum_{k \in U_v}||\bz_k^{*(b)}||_2^2\nonumber\\
&\leqslant & \frac{C_2}{n_v}\label{ordre_HT_z*}
\end{eqnarray}
by assumptions (H\ref{H1})-(H\ref{H3}) and the fact that $||\bz_k^{*(b)}||_2^2=1$ for all $k\in U_v$ and $b=1, \ldots, B.$ Then, from relations (\ref{ordre_terme_2_1}), (\ref{ordre_HT_z*}) and lemma \ref{ordre_hat_beta*}, we get that there exists a positive constant $\tilde C_2$ such that, for any $b=1, \ldots, B,$ we have:
\begin{eqnarray}
\dfrac{1}{N_v} \mathbb{E}_p\bigg| \sum_{k\in U_v}\alpha_k(\widehat{m}^{*(b)}_{rf}(\bx_k) - \widetilde{m}_{N, rf}^{*(b)}(\bx_k))\bigg\rvert\leqslant\sqrt{\frac{C_2}{n_v}\frac{\tilde c_2N_v}{N^2_{0v}}}\leqslant \frac{\tilde C_2}{N_{0v}}\label{ordre_terme_2}
\end{eqnarray}
by using also the assumption (H\ref{H1}). The result follows then from relations (\ref{ordre_t*_diff}) and (\ref{ordre_terme_2}). 
\end{proof}


\begin{result}\label{res2}
	Consider a sequence of RF estimators $\{\widehat{t}^*_{rf}\}.$  
	If $\displaystyle\frac{N^{a}_v}{N_{0v}}=O(1)$ with $1/2< a \leqslant 1,$ then 
	\begin{equation*}
	\dfrac{\sqrt{n_v}}{N_v} \left(\widehat{t}^*_{rf}- t_y\right)=\dfrac{\sqrt{n_v}}{N_v}\left(\widehat{t}_{pgd} -t_y\right) + o_\mathbb{P}(1).
	\end{equation*}
\end{result}
\begin{proof}
We get from relation (\ref{erreur_1}) and lemmas (\ref{ordre_t_diff_pop}) and the proof of result \ref{res1} (relation \ref{ordre_terme_2}) that
	\begin{align*}
\dfrac{\sqrt{n_v}}{N_v} \left(\widehat{t}^*_{rf} - t_y\right) &= \dfrac{\sqrt{n_v}}{N_v} \left( \widehat{t}_{pgd} - t_y \right)+ \dfrac{1}{B} \sum_{b=1}^B \bigg[\dfrac{\sqrt{n_v}}{N_v}  \sum_{k \in U_v} \alpha_k (\widehat{m}^{*(b)}_{tree}(\bx_k) - \widetilde{m}_{tree}^{*(b)}(\bx_k)) \bigg]. \\
&= \dfrac{\sqrt{n_v}}{N_v} \left( \widehat{t}_{pgd} - t_y \right) + \mathcal{O}_\mathbb{P} \left(\dfrac{\sqrt{n_v}}{N_{0v}}\right)\\
&= \dfrac{\sqrt{n_v}}{N_v} \left( \widehat{t}_{pgd} - t_y \right) + {o}_\mathbb{P}(1)
\end{align*}
provided that$\displaystyle\frac{N^{a}_v}{N_{0v}}=O(1)$ with $1/2< a \leqslant 1.$
\end{proof}

\begin{result}\label{var_est_F1}
	Consider a sequence of population RF estimators $\{\widehat{t}^*_{rf}\}.$ 
	Assume that $\displaystyle\frac{N^{a}_v}{N_{0v}}=O(1)$ with $1/2< a \leqslant 1,$ then the variance estimator $\widehat{\V}_{rf}(\widehat t_{rf}^*)$ is design-consistent for the asymptotic variance $\mathbb{AV}_p \left( \widehat{t}^*_{rf}\right).$ That is,
	\begin{equation*}
	\lim_{v \to \infty}\mathbb{E}_p \left(\dfrac{n_v}{N_v^2} \biggr \rvert\widehat{\V}_{rf}(\widehat t_{rf}^*)- \mathbb{AV}_p (\widehat{t}^*_{rf})  \biggr \rvert \right) = 0. 
	\end{equation*}
\end{result}

\noindent\textbf{Proof}  Consider the following decomposition 
\begin{eqnarray}
&&n_v\left(\widehat{\mathbb{V}}_p \left(N_v^{-1}\widehat{t}^*_{rf}\right)-\mathbb{AV}_p \left(N_v^{-1}\widehat{t}^*_{rf}\right)\right)\nonumber\\
&=&n_v\left(\widehat{\mathbb{V}}_p \left(N_v^{-1}\widehat{t}^*_{rf}\right)-\widehat{\mathbb{V}}_p \left(N_v^{-1}\widehat{t}_{pgd}\right)\right)+n_v\left(\widehat{\mathbb{V}}_p \left(N_v^{-1}\widehat{t}_{pgd}\right)-\mathbb{AV}_p \left(N_v^{-1}\widehat{t}^*_{rf}\right)\right)\nonumber
\label{var_decomp}
\end{eqnarray}
where $\widehat{\mathbb{V}}_p \left(N^{-1}_v\widehat{t}_{pgd}\right)$ is the pseudo-type variance estimator of $\mathbb{V}_p \left(N^{-1}_v\widehat{t}_{pgd}\right)=\mathbb{AV}_p \left(N_v^{-1}\widehat{t}^*_{rf}\right)$ given by
\begin{eqnarray*}
	\widehat{\mathbb{V}}_p \left(N^{-1}_v\widehat{t}_{pgd}\right)= \frac{1}{N^2_v}\sum_{k\in U_v}\sum_{\ell \in U_v}\frac{\pi_{k\ell }-\pi_k\pi_{\ell}}{\pi_{k\ell }}\frac{y_k-\widetilde{m}^*_{rf}(\mathbf{x}_k)}{\pi_k}\frac{y_{\ell }-\widetilde{m}^*_{rf}(\mathbf{x}_{\ell })}{\pi_{\ell}}I_kI_{\ell }.
\end{eqnarray*}
Now, to prove that the consistency of the first  term from right of (\ref{var_decomp}), we use the same decomposition as in \cite{goga_ruiz-gazen}.
Denote $\tilde e_k=y_k-\widetilde{m}^*_{rf}(\mathbf{x}_k),$ $\hat e_k=y_k-\widehat{m}^*_{rf}(\mathbf{x}_k)$ and $c_{k\ell}=\displaystyle\frac{\pi_{k\ell }-\pi_k\pi_{\ell}}{\pi_{k\ell }\pi_k\pi_{\ell}}I_kI_{\ell}.$ Then, 
\begin{align*}
	n_v(\widehat{\mathbb{V}}_p \left(N_v^{-1}\widehat{t}^*_{rf}\right)-\widehat{\mathbb{V}}_p \left(N_v^{-1}\widehat{t}_{pgd}\right) )&= \frac{n_v}{N^2_v}\sum_{k\in U_v}\sum_{\ell\in U_v}c_{k\ell}\left(\hat e_k\hat e_{\ell}-\tilde e_k\tilde e_{\ell}\right)\\
	&=  \frac{n_v}{N^2_v}\sum_{k\in U_v}\sum_{\ell\in U_v}c_{k\ell}\left[(\hat e_k-\tilde e_k)(\hat e_{\ell}-\tilde e_{\ell})+\tilde e_k(\hat e_{\ell}-\tilde e_{\ell})+\tilde e_{\ell}(\hat e_k-\tilde e_k)\right]\\
	&= A_1+A_2+A_3.
\end{align*} 
For all $k \in U_v$, $\hat e_k-\tilde e_k=\widetilde{m}^*_{rf}(\mathbf{x}_k)-\widehat{m}^*_{rf}(\mathbf{x}_k)$ and thus, 
\begin{align*}
	\E_p|A_1|\leq \left(\frac{n_v}{\lambda^2N_v}+\frac{n_v\max_{k\neq \ell\in U_v}|\pi_{k\ell}-\pi_k\pi_{\ell}|}{\lambda^*\lambda^2}\right)\frac{1}{N_v}\sum_{k\in U_v}\E_p(\hat e_k-\tilde e_k)^2,
\end{align*}
by assumptions (H\ref{H2})-(H\ref{H3}).  Therefore, it suffices to show that, for all $k \in U_v$, one has $\E_p(\hat e_k-\tilde e_k)^2 = o(1)$ uniformly in $k$, which we show next. 
We have 
$$
\mathbb{E}_p(\widetilde{m}^*_{rf}(\mathbf{x}_k)-\widehat{m}^*_{rf}(\mathbf{x}_k))^2\leqslant \frac{1}{B}\sum_{b=1}^B\mathbb{E}_p(\widetilde{m}_{tree}^{*(b)}(\bx_k)-\widehat{m}^{*(b)}_{tree}(\bx_k))^2.
$$
We can write by using relations (\ref{expr_tilde_m_pop}) and (\ref{expr_hat_m_pop}):
\begin{eqnarray*}
\widehat{m}^{*(b)}_{tree}(\bx_k) - \widetilde{m}_{tree}^{*(b)}(\bx_k) &=&(\mathbf{z}^{*(b)}_k)^{\top} (\widehat{\boldsymbol{\beta}}^{*(b)}-\widetilde{\boldsymbol{\beta}}^{*(b)})
\end{eqnarray*}
and then, by using lemma (\ref{ordre_hat_beta*}),
\begin{eqnarray*}
\mathbb{E}_p(\widetilde{m}^*_{rf}(\mathbf{x}_k)-\widehat{m}^*_{rf}(\mathbf{x}_k))^2&\leqslant& \frac{1}{B}\sum_{b=1}^B\mathbb{E}_p\left( ||\mathbf{z}^{*(b)}_k||^2_2||\widehat{\boldsymbol{\beta}}^{*(b)}-\widetilde{\boldsymbol{\beta}}^{*(b)}||^2_2\right)\\
&\leqslant & \frac{\tilde c_2N_v}{N^2_{0v}}
\end{eqnarray*}
quantity going to zero provided that $\displaystyle\frac{N^{a}_v}{N_{0v}}=O(1)$ with $1/2< a \leqslant 1.$


Using the same arguments, we obtain that $\E_p|A_2|=o(1)$ and $\E_p|A_3|=o(1). $ We get then 
$$
n_v\E_p|\widehat{\mathbb{V}}_p \left(N_v^{-1}\widehat{t}^*_{rf}\right)-\widehat{\mathbb{V}}_p \left(N_v^{-1}\widehat{t}_{pgd}\right)|=o(1).
$$

The second term from right of (\ref{var_decomp}) concerns the consistency of the estimator of the Horvitz-Thompson variance computed for the population residuals  $y_k-\widetilde{m}^*_{rf}(\mathbf{x}_k), k\in U_v. $ The proof of this consistency \citep{breidt_opsomer_2000} requires  assumptions only on the higher order inclusion probabilities (H\ref{H4}) as well as finite  forth moment of $y_k-\widetilde{m}^*_{rf}(\mathbf{x}_k):$ 
$$
\frac{1}{N_v}\sum_{k\in U_v}(y_k-\widetilde{m}^*_{rf}(\mathbf{x}_k))^4\leq \frac{4}{N_v}\sum_{k\in U_v}(y^4_k+(\widetilde{m}^*_{rf}(\mathbf{x}_k))^4)<\infty.
$$
So,
$$
n_v\E_p|\widehat{\mathbb{V}}_p \left(N_v^{-1}\widehat{t}_{pgd}\right)-\mathbb{AV}_p \left(N_v^{-1}\widehat{t}^*_{rf}\right)|=o(1)
$$ 
and the result follows.

\subsection{Asymptotic results: the sample RF model-assisted estimator $\widehat{t}_{rf}$}


The sample RF model-assisted estimator is given by 
$$
\widehat{t}_{rf}=\sum_{k\in U_v}\widehat{m}_{rf}(\bx_k)+\sum_{k\in S_v}\frac{y_k-\widehat{m}_{rf}(\bx_k)}{\pi_k},
$$
where $\widehat{m}_{rf}$ is the estimator of $m$ built at the sample level and by using RF based on partition built at the sample level (for more details, see relation (17) from the main paper):  
$$
\widehat{m}_{rf}(\bx_k)=\sum_{\ell\in S_v}\frac{1}{\pi_{\ell}}\widehat{W}_{\ell}(\bx_k)y_{\ell},
$$
where 
$$
\widehat{W}_{\ell}(\bx_k)=\frac{1}{B}\sum_{b=1}^B\frac{\psi_\ell^{(b, S)} \mathds{1}_{\bx_{\ell}\in  A^{(S)} \left(\boldsymbol{\mathbf{x}_k} , \theta_b^{(S)}\right)}}{{\widehat N}(\bx_k, \theta_b^{(S)})}
$$
and ${\widehat N}(\bx_k, \theta_b^{(S)})=\sum_{\ell\in S_v}\pi_{\ell}^{-1}\psi_\ell^{(b, S)} \mathds{1}_{\bx_{\ell}\in A^{(S)} \left(\boldsymbol{\mathbf{x}_k} , \theta_b^{(S)}\right)}$ is the estimated number of units falling in the terminal node $A^{(S)} \left(\boldsymbol{\mathbf{x}_k} , \theta_b^{(S)}\right)$ containing $\bx_k.$ As in Section \ref{sec2}, the estimator  $\widehat{m}_{rf}(\bx_k)$ can be written as a bagged estimator of $m$ as follows:
$$
\widehat{m}_{rf}(\bx_k)=\frac{1}{B}\sum_{b=1}^B\widehat{m}_{tree}^{(b)}(\bx_k)
$$
and $\widehat{m}^{(b)}_{tree}(\bx_k)$ is the  estimation of $m$ based on the $b$-th stochastic tree:
\begin{eqnarray}
\widehat{m}_{tree}^{(b)}(\bx_k)=\sum_{\ell\in S_v}\frac{1}{\pi_{\ell}}\dfrac{\psi_\ell^{(b, S)} \mathds{1}_{\bx_{\ell}\in  A^{(S)} \left(\boldsymbol{\mathbf{x}_k} , \theta_b^{(S)}\right)}}{{\widehat N}(\bx_k, \theta_b^{(S)})}y_{\ell}\label{hat_m_1}
\end{eqnarray}
As in Section \ref{sec2}, for more readability,  we  note in the sequel $\widehat{m}_{tree}^{(b)}(\bx_k)$ instead of  $\widehat{m}_{tree}^{(b)}(\bx_k,\theta_b^{(S)}). $
Consider the pseudo-generalized difference estimator:
$$
\hat{t}_{pgd}= \sum_{k\in U_v} \widetilde{m}_{rf}(\bx_k)+\sum_{k\in S_v}\frac{y_k- \widetilde{m}_{rf}(\bx_k)}{\pi_k}
$$
where $\widetilde{m}_{rf}$ is the estimation of $m$ built at the population level by using RF based on partition built also at the population level (relation (9) from the main paper):
$$
\widetilde{m}_{rf}(\bx_k)=\sum_{\ell\in U_v}\widetilde{W}_{\ell}(\bx_k)y_{\ell},
$$
where 
$$
\widetilde{W}_{\ell}(\bx_k)=\frac{1}{B}\sum_{b=1}^B\frac{\psi_\ell^{(b, U)} \mathds{1}_{\bx_{\ell}\in  A^{(U)} \left(\boldsymbol{\mathbf{x}_k} , \theta_b^{(U)}\right)}}{{\widetilde N}(\bx_k, \theta_b^{(U)})}
$$ 
with ${\widetilde N}(\bx_k, \theta_b^{(U)})=\sum_{\ell\in U_v}\psi_\ell^{(b, U)} \mathds{1}_{\bx_{\ell}\in A^{(U)} \left(\boldsymbol{\mathbf{x}_k} , \theta_b^{(U)}\right)}$ is the number of units falling in the terminal node $A^{(U)} \left(\boldsymbol{\mathbf{x}_k} , \theta_b^{(U)}\right)$ containing $\bx_k.$ The estimator $\widehat{m}_{rf}$ can be also written as a bagged estimator as follows: 
$$
\widetilde{m}_{rf}(\bx_k)=\frac{1}{B}\sum_{b=1}^B\widetilde{m}_{tree}^{(b)}(\bx_k)
$$
and 
$$
\widetilde{m}_{tree}^{(b)}(\bx_k)=\sum_{\ell\in U_v}\dfrac{\psi_\ell^{(b, U)} \mathds{1}_{\bx_\ell\in A^{(U)} \left(\boldsymbol{\mathbf{x}_k} , \theta_b^{(U)}\right)}}{{\widetilde N}(\bx_k, \theta_b^{(U)})}y_{\ell}.
$$
As in the previous section, we will write $\widetilde{m}_{rf}$ and $\widehat{m}_{rf}$ in equivalent forms. Consider for that the $B$ partitions build at the population level $\widetilde{\mathcal{P}}_U=\{\widetilde{\mathcal{P}}_U^{(b)}\}_{b=1}^B.$ For a given $b=1, \ldots, B,$ the partition $\widetilde{\mathcal{P}}_U^{(b)}$ is composed by the disjointed regions $\widetilde{\mathcal{P}}_U^{(b)}=\{A^{(bU)}_j\}_{j=1}^{J_{bU}}.$ Consider $\mathbf{z}^{(b)}_k=( \mathds{1}_{\bx_k \in A^{(bU)}_1}, \ldots,\mathds{1}_{\bx_k \in A^{(bU)}_{J_{bU}}})^{\top}$ where $\mathds{1}_{\bx_k \in {A}^{(bU)}_j}=1$ if  $\bx_k$ belongs to the region ${A}^{(bU)}_j$ and zero otherwise for all $j=1, \ldots, J_{bU}$. Since $\widetilde{\mathcal{P}}_U^{(b)}$ is a partition, then $\bx_k$ belongs to only one region at the $b$-th step. Suppose for example that $\bx_k\in {A}^{(bU)}_j,$ then $\widetilde{m}_{tree}^{(b)}(\bx_k)$ is the mean of $y$-values for individuals $\ell$ for which $\bx_{\ell}\in {A}^{(bU)}_j:$ 
$$
\widetilde{m}_{tree}^{(b)}(\bx_k)=\sum_{\ell\in U_v}\frac{\psi_{\ell}^{(b,U)}\mathds{1}_{\bx_{\ell} \in {A}^{(bU)}_j}}{\widetilde{N}^{(b)}_{j}}y_{\ell}, \quad\mbox{for}\quad \bx_k\in {A}^{(bU)}_j,
$$
where $\widetilde N^{(b)}_{j}$ is the number of units belonging to the region ${A}^{(bU)}_j:$
\begin{eqnarray}
\widetilde N^{(b)}_{j}=\sum_{\ell\in U_v}\psi_{\ell}^{(b,U)}\mathds{1}_{\bx_{\ell} \in {A}^{(bU)}_j}, \quad j=1, \ldots, J_{bU}.
\end{eqnarray}
Then, $\widetilde{m}_{tree}^{(b)}(\bx_k)$ can be written as a regression-type estimator with $\mathbf{z}^{(b)}_k$ as explanatory variables:
\begin{eqnarray}
\widetilde{m}_{tree}^{(b)}(\bx_k)=(\mathbf{z}^{(b)}_k)^{\top}\widetilde{\boldsymbol{\beta}}^{(b)}, \quad k\in U_v\
\end{eqnarray}
where 
$$
\widetilde{\boldsymbol{\beta}}^{(b)}=\left(\sum_{\ell\in U_v}\psi_{\ell}^{(b,U)}\mathbf{z}^{(b)}_{\ell}(\mathbf{z}^{(b)}_{\ell})^{\top}\right)^{-1}\sum_{\ell\in U_v}\psi_{\ell}^{(b,U)}\mathbf{z}^{(b)}_{\ell}y_{\ell}.
$$
Based on the same arguments as  in Section \ref{sec2}, the matrix $\sum_{\ell\in U_v}\psi_{\ell}^{(b,U)}\mathbf{z}^{(b)}_{\ell}(\mathbf{z}^{(b)}_{\ell})^{\top}$ is diagonal with diagonal elements equal to $\widetilde N^{(b)}_{j}, j=1, \ldots, J_{bU}.$ By the stopping criterion, we have that all $\widetilde N^{(b)}_{j}\geq N_0>0,$ so the matrix $\sum_{\ell\in U_v}\psi_{\ell}^{(b,U)}\mathbf{z}^{(b)}_{\ell}(\mathbf{z}^{(b)}_{\ell})^{\top}$ is  invertible and $\widetilde{\boldsymbol{\beta}}^{(b)}$ is well-defined. \\
Consider now the $B$ partitions build at the sample level $\widehat{\mathcal{P}}_S=\{\widehat{\mathcal{P}}_S^{(b)}\}_{b=1}^B.$ For a given $b=1, \ldots, B,$ the partition $\widehat{\mathcal{P}}_S^{(b)}$ is composed by the disjointed regions $\widehat{\mathcal{P}}_S^{(b)}=\{A^{(bS)}_{j}\}_{j=1}^{J_{bS}}.$ Consider $\hat{\mathbf{z}}^{(b)}_k=( \mathds{1}_{\bx_k \in A^{(b)}_{1S}}, \ldots,\mathds{1}_{\{\bx_k \in A^{(b)}_{J_{bS}}\}})^{\top}$ where $\mathds{1}_{\bx_k \in {A}^{(bS)}_{j}}=1$ if  $\bx_k$ belongs to the region ${A}^{(bS)}_{j}$ and zero otherwise for all $j=1, \ldots, J_{bS}$. Here, the hat notation is to design the fact that the vector $\hat{\mathbf{z}}^{(b)}_k$ depends on random dummy variables $\mathds{1}_{\bx_k \in {A}^{(bS)}_{j}}. $ Since $\{A^{(bS)}_{j}\}_{j=1}^{J_{bS}}$ form a partition, then $\bx_k$ belongs to only one terminal node.
Suppose for example that $\bx_k\in {A}^{(bS)}_{j},$ then $\widehat{m}_{tree}^{(b)}(\bx_k)$ is a Hajek-type estimator:
$$
\widehat{m}_{tree}^{(b)}(\bx_k)=\sum_{\ell\in S_v}\frac{1}{\pi_{\ell}}\frac{\psi^{(b,S)}_{\ell}\mathds{1}_{\bx_{\ell} \in {A}^{(bS)}_j}y_{\ell}}{\widehat N^{(b)}_j}, \quad \mbox{for}\quad \bx_k\in {A}^{(bS)}_j,
$$
where $\widehat N^{(b)}_j$ is the estimated number of units falling in the terminal node ${A}^{(bS)}_j:$
$$
\widehat N^{(b)}_j=\sum_{\ell\in S_v}\frac{1}{\pi_{\ell}}\psi^{(b,S)}_{\ell}\mathds{1}_{\bx_{\ell} \in {A}^{(bS)}_j}, \quad j=1, \ldots, J_{bS}.
$$
Then,
$\widehat{m}_{tree}^{(b)}(\bx_k)$ can be written also as a regression-type estimator with $\hat{\mathbf{z}}^{(b)}_k$ as explanatory variables:
\begin{eqnarray}
\widehat{m}_{tree}^{(b)}(\bx_k)=(\hat{\mathbf{z}}^{(b)}_k)^{\top}\widehat{\boldsymbol{\beta}}^{(b)}, k\in U_v,\label{hat_m_tree_sample}
\end{eqnarray}
where 
$$
\widehat{\boldsymbol{\beta}}^{(b)}=\left(\sum_{\ell\in S_v}\frac{1}{\pi_{\ell}}\psi_{\ell}^{(b,S)}\hat{\mathbf{z}}^{(b)}_{\ell}(\hat{\mathbf{z}}^{(b)}_{\ell})^{\top}\right)^{-1}\sum_{\ell\in S_v}\frac{1}{\pi_{\ell}}\psi_{\ell}^{(b,S)}\hat{\mathbf{z}}^{(b)}_{\ell}y_{\ell}.
$$
As in Section \ref{sec2}, remark that $\widehat{\boldsymbol{\beta}}^{(b)}$  may be obtained as solution of the following weighted estimating equation:
$$
\sum_{\ell\in S_v}\frac{1}{\pi_{\ell}}\psi_{\ell}^{(b,S)}\hat{\mathbf{z}}^{(b)}_{\ell}(y_{\ell}-(\hat{\mathbf{z}}^{(b)}_{\ell})^{\top}\boldsymbol{\beta}^{(b)})=0.
$$
Since $\{A^{(bS)}_{j}\}_{j=1}^{J_{bS}}$ is a partition, then the matrix $\sum_{\ell\in S_v}\frac{1}{\pi_{\ell}}\psi_{\ell}^{(b,S)}\hat{\mathbf{z}}^{(b)}_{\ell}(\hat{\mathbf{z}}^{(b)}_{\ell})^{\top}$ is diagonal with diagonal elements equal to $\widehat N^{(b)}_j, j=1, \ldots, J_{bS}. $ By the stopping criterion and assumption (H\ref{H3}), we have that $\sum_{\ell\in S_v}\frac{1}{\pi_{\ell}}\psi_{\ell}^{(b,S)}\mathds{1}_{\bx_\ell \in A^{(b)}_{jS}}\geq  n_{0v}>0,$ so $\sum_{\ell\in S_v}\frac{1}{\pi_{\ell}}\psi_{\ell}^{(b,S)}\hat{\mathbf{z}}^{(b)}_{\ell}(\hat{\mathbf{z}}^{(b)}_{\ell})^{\top}$ is always invertible  is and $\widehat{\boldsymbol{\beta}}^{(b)}$ is well-defined whatever the sample $S$ is.

We need to consider also a second pseudo-generalized difference estimator:
$$
\widehat{\widetilde{t}}_{pgd}= \sum_{k\in U_v} \widehat{\widetilde m}_{rf}(\bx_k)+\sum_{k\in S_v}\frac{y_k- \widehat{\widetilde m}_{rf}(\bx_k)}{\pi_k}
$$
where 
\begin{eqnarray*}
\widehat{\widetilde m}_{rf}(\bx_k)&=& \sum_{\ell\in U_v}\left(\frac{1}{B}\sum_{b=1}^B\frac{\psi_{\ell}^{(b,S)}\mathds{1}_{\bx_{\ell}\in A^{(S)}(\bx_k, \theta_b^{(S)})}}{\widehat{\widetilde N}(\bx_k, \theta_b^{(S)})}\right)y_{\ell}\label{expr1_hat_tilde_m}\\
&=&\frac{1}{B}\sum_{b=1}^B\widehat{\widetilde m}_{tree}^{(b)}(\bx_k)
\end{eqnarray*}
with $\widehat{\widetilde N}(\bx_k, \theta_b^{(S)})=\sum_{\ell\in U_v}\psi_{\ell}^{(b,S)}\mathds{1}_{\bx_{\ell}\in A(\bx_k, \theta_b^{(S)})}$ and
\begin{eqnarray}
\widehat{\widetilde m}_{tree}^{(b)}(\bx_k)=
\sum_{\ell\in U_v}\frac{\psi_{\ell}^{(b,S)}\mathds{1}_{\bx_{\ell}\in A^{(S)}(\bx_k, \theta_b^{(S)})}y_{\ell}}{\widehat{\widetilde N}(\bx_k, \theta_b^{(S)})}
=(\hat{\mathbf{z}}^{(b)}_k)^{\top}\widehat{\widetilde{\boldsymbol{\beta}}}^{(b)}, \quad k\in U_v\label{tilde_hat_m_tree_sample}
\end{eqnarray}
for 
\begin{eqnarray*}
\widehat{\widetilde{\boldsymbol{\beta}}}^{(b)}=\left(\sum_{\ell\in U_v}\psi_{\ell}^{(b,S)}\hat{\mathbf{z}}^{(b)}_{\ell}(\hat{\mathbf{z}}^{(b)}_{\ell})^{\top}\right)^{-1}\sum_{\ell\in U_v}\psi_{\ell}^{(b,S)}\hat{\mathbf{z}}^{(b)}_{\ell}y_{\ell}.
\end{eqnarray*}
The matrix $\sum_{\ell\in U_v}\psi_{\ell}^{(b,S)}\hat{\mathbf{z}}^{(b)}_{\ell}(\hat{\mathbf{z}}^{(b)}_{\ell})^{\top}$ is also diagonal with diagonal elements equal to $\sum_{\ell\in U_v}\psi_{\ell}^{(b,S)}\mathds{1}_{\bx_{\ell}\in {A}^{(bS)}_j}\geq n_{0v}>0, j=1, \ldots, J_{bS}$ so $\widehat{\widetilde{\boldsymbol{\beta}}}^{(b)}$ is also well-defined whatever the sample $S$ is. 
In order to prove the consistency of the sample-based RF estimator $\widehat{t}_{rf},$ we use the following decomposition:
\begin{eqnarray}
\frac{1}{N_v}(\hat t_{rf}-t_y)=\frac{1}{N_v}(\widehat{\widetilde{t}}_{pgd} - t_y)-\frac{1}{N_v}\sum_{k\in U_v}\alpha_k(\widehat{m}_{rf}(\bx_k)-\widehat{\widetilde m}_{rf}(\bx_k)). \label{desc_t_rf}
\end{eqnarray}
We will give first several useful lemmas. The constants used in the following results may not be the same as the ones from Section \ref{sec2} even if they are denoted in the same way for simplicity. 
\begin{lemma}\label{ordre_t_diff_pop2}
There exists a positive constant $\tilde c_1$ such that:
$$
\frac{n_v}{N^2_v}\mathbb{E}_p(\widehat{t}_{pgd} - t_y)^2\leqslant \tilde c_1.
$$
\end{lemma}
\begin{proof}
 The proof is similar to that of lemma \ref{ordre_t_diff_pop}. We also have that $\sup_{k\in U_v}|\widetilde{m}_{rf}(\bx_k)|\leqslant C$ by using assumption (H\ref{H1}). Further, 
\begin{align*}
n_v\mathbb{E}_p \left(\dfrac{\widehat{t}_{pgd} - t_y}{N_v}\right)^2 &\leqslant \left( \dfrac{n_v}{N_v} \cdot \dfrac{1 }{\lambda} + \frac{n_v \max_{ k \neq \ell \in U_v} \rvert \pi_{k\ell}-\pi_k\pi_{\ell} \rvert}{\lambda^2}\right) \cdot \dfrac{2}{N_v} \sum_{k \in U_v} \left(y^2_k + (\widetilde{m}_{rf}(\bx_k))^2 \right)\\
&\leqslant \tilde c_1
\end{align*}
by assumptions (H\ref{H1})-(H\ref{H3}).
\end{proof}
\begin{lemma}\label{t_pgd_tilde_hat}
There exists a positive constant $\tilde c_2$ such that:
$$
\frac{n_v}{N^2_v}\mathbb{E}_p(\widehat{\widetilde{t}}_{pgd} - t_y)^2\leqslant \tilde c_2.
$$
\end{lemma}
\begin{proof}
Using (\ref{expr1_hat_tilde_m}), we get that $\widehat{\widetilde m}_{rf}(\bx_k)$ can be written as a weighted sum of $y$-values with positive weights summing to unity, so  $\sup_{k\in U_v}|\widehat{\widetilde m}_{rf}(\bx_k)|\leqslant C$ by using also assumption (H\ref{H1}). Now,
\begin{eqnarray*}
\widehat{\widetilde{t}}_{pgd} - t_y=\sum_{k\in U_v}\alpha_k(y_k-\widehat{\widetilde m}_{rf}(\bx_k))
\end{eqnarray*}
and 
\begin{eqnarray*}
\frac{n_v}{N^2_v}\mathbb{E}_p(\widehat{\widetilde{t}}_{pgd} - t_y)&=&\frac{n_v}{N^2_v}\sum_{k\in U_v}\E_p\left[\alpha^2_k(y_k-\widehat{\widetilde m}_{rf}(\bx_k))^2\right]\\
&&+\frac{n_v}{N^2_v}\sum_{k\in U_v}\sum_{\ell\neq k, \ell\in U_v}\E_p\left[(y_k-\widehat{\widetilde m}_{rf}(\bx_k))(y_\ell-\widehat{\widetilde m}_{rf}(\bx_\ell))\E_p(\alpha_k\alpha_{\ell}|\mathcal {\hat{P}}_S)\right]\\
&\leqslant & \frac{2n_vC^2}{\lambda N_v}+\frac{n_v}{N^2_v}\sum_{k\in U_v}\sum_{\ell\neq k, \ell\in U_v}\E_p\left[|y_k-\widehat{\widetilde m}_{rf}(\bx_k)||y_\ell-\widehat{\widetilde m}_{rf}(\bx_\ell)| \max_{\ell\neq k\in U_v}|\E_p(\alpha_k\alpha_{\ell}|\mathcal {\hat{P}}_S)|\right]\\
&\leqslant & \tilde c_2,
\end{eqnarray*}
by assumptions (H\ref{H2}) and (H\ref{H5}). 
\end{proof}
\begin{lemma}\label{borne_beta_split_sample}
There exists a positive constant $\tilde c_3$ not depending on $b$ such that:
$$
\E_p\left|\left|\widehat{\boldsymbol{\beta}}^{(b)}-\widehat{\widetilde{\boldsymbol{\beta}}}^{(b)}\right|\right|^2_2\leqslant \frac{\tilde c_3n_v}{n^2_{0v}}, 
$$
for all $b=1, \ldots, B.$
\begin{proof}
Let denote by $\widehat{\mathbf{T}}^{(b)}=\sum_{\ell\in S_v}\frac{1}{\pi_{\ell}}\psi_{\ell}^{(b,S)}\hat{\mathbf{z}}^{(b)}_{\ell}(\hat{\mathbf{z}}^{(b)}_{\ell})^{\top}.$ As already mentioned, the $J_{bS}\times J_{bS}$ dimensional matrix $\widehat{\mathbf{T}}^{(b)}$ is diagonal with diagonal elements given by $\widehat N^{(b)}_j=\sum_{\ell\in S_v}\frac{1}{\pi_{\ell}}\psi_{\ell}^{(b,S)}\mathds{1}_{\bx_\ell \in A^{(b)}_{jS}}$ the weighted somme of units falling in  the region ${A}^{(b)}_{jS}$ for $j=1, \ldots, J_{bS}$ and by the stopping criterion, we have that $\widehat N^{(b)}_j\geq  n_{0v}>0. $ The matrix $\widehat{\mathbf{T}}^{(b)}$ is then always invertible with 
\begin{eqnarray}
||(\widehat{\mathbf{T}}^{(b)})^{-1}||_2\leq  n^{-1}_{0v}\quad \mbox{for all} \quad b=1, \ldots B.\label{borne_hat_T}
\end{eqnarray}
Now, write
\begin{eqnarray}
\widehat{\boldsymbol{\beta}}^{(b)}-\widehat{\widetilde{\boldsymbol{\beta}}}^{(b)}&=&(\widehat{\mathbf{T}}^{(b)})^{-1}\left(\sum_{\ell\in S_v}\frac{1}{\pi_{\ell}}\psi_{\ell}^{(b,S)}\hat{\mathbf{z}}^{(b)}_{\ell}y_{\ell}-\widehat{\mathbf{T}}^{(b)}\widehat{\widetilde{\boldsymbol{\beta}}}^{(b)}\right)\nonumber\\
&=& (\widehat{\mathbf{T}}^{(b)})^{-1}\sum_{\ell\in S_v}\frac{1}{\pi_{\ell}}\psi_{\ell}^{(b,S)}\hat{\mathbf{z}}^{(b)}_{\ell}\left(y_{\ell}-  \widehat{\widetilde m}_{tree}^{(b)}(\bx_\ell)\right)\nonumber\\
&=& (\widehat{\mathbf{T}}^{(b)})^{-1}\sum_{\ell\in U_v}\alpha_{\ell}\widehat E^{(b)}_{\ell}\label{decomp_hat_beta}
\end{eqnarray}
where $\widehat E^{(b)}_{\ell}=\psi_{\ell}^{(b,S)}\hat{\mathbf{z}}^{(b)}_{\ell}(y_{\ell}-\widehat{\widetilde m}_{tree}^{(b)}(\bx_\ell))$ with $\sum_{\ell\in U_v}\widehat E^{(b)}_{\ell}=0. $ We have that $||\hat{\mathbf{z}}^{(b)}_{\ell}||_2=1$ and $\sup_{\ell\in U_v}|\widehat{\widetilde m}_{tree}^{(b)}(\bx_\ell))|\leq C$  for all $\ell\in U_v$ and $b=1, \ldots, B,$ then:
\begin{eqnarray*}
||\widehat E^{(b)}_{\ell}||^2_2\leqslant 2C^2.
\end{eqnarray*}
Following the same lines as in lemma \ref{t_pgd_tilde_hat}, we get that it exists a positive constant $\tilde C_0$ not depending on $b$ such that 
\begin{eqnarray}
\frac{1}{N^2_v}\E_p\left|\left|\sum_{\ell\in U_v}\alpha_{\ell}\widehat E^{(b)}_{\ell}\right|\right|^2_2\leqslant \frac{\tilde C_0}{n_v},\quad \mbox{for all} \quad b=1, \ldots B.
\end{eqnarray}
We obtain then from relations (\ref{borne_hat_T}) and (\ref{decomp_hat_beta}) that:
\begin{eqnarray}
\E_p\left|\left|\widehat{\boldsymbol{\beta}}^{(b)}-\widehat{\widetilde{\boldsymbol{\beta}}}^{(b)}\right|\right|^2_2&\leqslant& \E_p\left(N^2_v||(\widehat{\mathbf{T}}^{(b)})^{-1}||^2_2\frac{1}{N^2_v}\left|\left|\sum_{\ell\in U_v}\alpha_{\ell}\widehat E^{(b)}_{\ell}\right|\right|^2_2\right)\nonumber\\
&\leqslant & \frac{ N^2_v}{n^2_{0v}}\frac{1}{N^2_v}\E_p\left|\left|\sum_{\ell\in U_v}\alpha_{\ell}\widehat E^{(b)}_{\ell}\right|\right|^2_2\nonumber\\
&\leqslant & \frac{ N^2_v}{n^2_{0v}}\frac{\tilde C_0}{n_v}\nonumber\\
&\leqslant & \frac{\tilde c_3n_v}{n^2_{0v}}
\end{eqnarray}
by assumption (H\ref{H2}).
\end{proof}
\end{lemma}

\begin{result}
Consider a sequence of sample RF estimators $\{\widehat{t}_{rf}\}$. Then, there exist positive constants $\tilde C_1, \tilde C_2$  such that 
$$
\frac{1}{N_v}\E_p|\hat t_{rf}-t_y|\leqslant \frac{\tilde C_1}{\sqrt{n_v}}+\frac{\tilde C_2}{n_{0v}}. 
$$
If $\displaystyle\frac{n^{u}_v}{n_{0v}}=O(1)$ with $1/2\leqslant u \leqslant 1,$ then 
	\begin{equation*}
	\mathbb{E}_p  \bigg\rvert \dfrac{1}{N_v} \left(\widehat{t}_{rf} - t_y\right)\bigg\rvert  \leqslant \dfrac{\tilde C}{\sqrt{n_{v}}}, \quad\mbox{with $\xi$-probability one.}
	\end{equation*}
\end{result}
\begin{proof} We use the decomposition given in relation (\ref{desc_t_rf}):
$$
\frac{1}{N_v}(\hat t_{rf}-t_y)=\frac{1}{N_v}(\widehat{\widetilde{t}}_{pgd} - t_y)-\frac{1}{N_v}\sum_{k\in U_v}\alpha_k(\widehat{m}_{rf}(\bx_k)-\widehat{\widetilde m}_{rf}(\bx_k)).
$$
Now,
\begin{eqnarray*}
\E_p\left|\frac{1}{N_v}\sum_{k\in U_v}\alpha_k(\widehat{m}_{rf}(\bx_k)-\widehat{\widetilde m}_{rf}(\bx_k))\right|\leqslant \frac{1}{B}\sum_{b=1}^B\frac{1}{N_v}\E_p\left|\sum_{k\in U_v}\alpha_k(\widehat{m}^{(b)}_{tree}(\bx_k)-\widehat{\widetilde m}^{(b)}_{tree}(\bx_k))\right|
\end{eqnarray*}
and using relations (\ref{hat_m_tree_sample}) and (\ref{tilde_hat_m_tree_sample}), we get:
\begin{eqnarray*}
\frac{1}{N_v}\E_p\left|\sum_{k\in U_v}\alpha_k(\widehat{m}^{(b)}_{tree}(\bx_k)-\widehat{\widetilde m}^{(b)}_{tree}(\bx_k))\right|&\leqslant& \E_p\left(\left|\left|\frac{1}{N_v}\sum_{k\in U_v}\alpha_k \hat{\mathbf{z}}^{(b)}_k\right|\right|_2\left|\left|\widehat{\boldsymbol{\beta}}^{(b)}-\widehat{\widetilde{\boldsymbol{\beta}}}^{(b)}\right|\right|_2\right)\nonumber\\
&\leqslant& \sqrt{\E_p\left|\left|\frac{1}{N_v}\sum_{k\in U_v}\alpha_k \hat{\mathbf{z}}^{(b)}_k\right|\right|^2_2\E_p\left|\left|\widehat{\boldsymbol{\beta}}^{(b)}-\widehat{\widetilde{\boldsymbol{\beta}}}^{(b)}\right|\right|^2_2}.
\end{eqnarray*}
We have that $||\hat{\mathbf{z}}^{(b)}_k||_2=1$ for all $k\in U_v$ and $b=1, \ldots, B.$ We can show then by using the same arguments as in the proof of lemma \ref{t_pgd_tilde_hat}, that there exists a positive constant $\tilde C'_0$ such that 
\begin{eqnarray*}
E_p\left|\left|\frac{1}{N_v}\sum_{k\in U_v}\alpha_k \hat{\mathbf{z}}^{(b)}_k\right|\right|^2_2\leqslant \frac{\tilde C'_0}{n_v}\label{borne_est_HT_hatz}
\end{eqnarray*}
which together with lemma \ref{borne_beta_split_sample} gives us that there exists a positive constant $\tilde C_2$ such that 
\begin{eqnarray}
\frac{1}{N_v}\E_p\left|\sum_{k\in U_v}\alpha_k(\widehat{m}_{rf}(\bx_k)-\widehat{\widetilde m}_{rf}(\bx_k))\right|&\leqslant&\frac{\tilde C_2}{n_{0v}}.
\label{terme_2_erreur_t_rf}
\end{eqnarray}
Now, 
\begin{eqnarray*}
\frac{1}{N_v}\E_p\bigg|\hat t_{rf}-t_y\bigg|&\leqslant &\frac{1}{N_v}\E_p\bigg|\widehat{\widetilde{t}}_{pgd} - t_y\bigg|+\frac{1}{B}\sum_{b=1}^B\frac{1}{N_v}\E_p\left|\sum_{k\in U_v}\alpha_k(\widehat{m}^{(b)}_{tree}(\bx_k)-\widehat{\widetilde m}^{(b)}_{tree}(\bx_k))\right|\\
&\leqslant& \frac{\tilde C_1}{\sqrt{n_v}}+\frac{\tilde C_2}{n_{0v}} 
\end{eqnarray*}
by using lemma \ref{t_pgd_tilde_hat} and relation (\ref{terme_2_erreur_t_rf}).
\end{proof}
\begin{result}
Consider a sequence of RF estimators $\{\widehat{t}_{rf}\}.$  Assume that $\displaystyle\frac{n^{u}_v}{n_{0v}}=O(1)$ with $1/2< u \leqslant 1.$ Then, 
	\begin{equation*}
	\dfrac{\sqrt{n_v}}{N_v} \left(\widehat{t}_{rf}- t_y\right)=\dfrac{\sqrt{n_v}}{N_v}\left(\widehat{t}_{pgd} -t_y\right) + o_\mathbb{P}(1).
	\end{equation*}
\end{result}
\begin{proof}
We have 
\begin{eqnarray}
\dfrac{\sqrt{n_v}}{N_v} \left(\widehat{t}_{rf}- t_y\right)=\dfrac{\sqrt{n_v}}{N_v}\left(\widehat{t}_{pgd} -t_y\right)+ \dfrac{\sqrt{n_v}}{N_v}\sum_{k\in U_v}\alpha_k(\widehat m_{rf}(\bx_k)-\widetilde m_{N,rf}(\bx_k)).\label{decomp_1}
\end{eqnarray}
Now,
\begin{eqnarray}
&&\dfrac{\sqrt{n_v}}{N_v}\sum_{k\in U_v}\alpha_k(\widehat m_{rf}(\bx_k)-\widetilde m_{rf}(\bx_k))\nonumber\\
&=&\dfrac{\sqrt{n_v}}{N_v}\sum_{k\in U_v}\alpha_k(\widehat m_{rf}(\bx_k)-\widehat{\widetilde m}_{rf}(\bx_k))+\dfrac{\sqrt{n_v}}{N_v}\sum_{k\in U_v}\alpha_k(\widehat{\widetilde m}_{rf}(\bx_k)-\widetilde m_{rf}(\bx_k)). \label{decomp_terme2_rf_sample}
\end{eqnarray}
Relation (\ref{terme_2_erreur_t_rf}) gives us that 
\begin{eqnarray}
\dfrac{\sqrt{n_v}}{N_v}\sum_{k\in U_v}\alpha_k(\widehat m_{rf}(\bx_k)-\widehat{\widetilde m}_{rf}(\bx_k))=O_\mathbb{P}\left(\frac{\sqrt{n_v}}{n_{0v}}\right)=o_{\mathbb{P}}(1)\label{decomp_2}
\end{eqnarray}
provided that $\displaystyle\frac{n^{u}_v}{n_{0v}}=O(1)$ with $1/2< u \leqslant 1.$ Consider now the second term  from the right-side of relation (\ref{decomp_terme2_rf_sample}). We have:
\begin{eqnarray*}
&&\frac{n_v}{N^2_v}\E_p\left(\sum_{k\in U_v}\alpha_k(\widehat{\widetilde m}_{rf}(\bx_k)-\widetilde m_{rf}(\bx_k))\right)^2\nonumber\\
&\leqslant&\frac{n_v}{N^2_v}\frac{(1+\lambda)^2}{\lambda^2}\sum_{k\in U_v}\E_p\left(\widehat{\widetilde m}_{rf}(\bx_k)-\widetilde m_{rf}(\bx_k)\right)^2\nonumber\\
&+& \frac{n_v}{N^2_v}\sum_{k\in U_v}\sum_{\ell\neq k, \ell\in U_v}\E_p\left[|\widehat{\widetilde m}_{rf}(\bx_k)-\widetilde m_{rf}(\bx_k)| |\widehat{\widetilde m}_{rf}(\bx_\ell)-\widetilde m_{rf}(\bx_\ell)| \max_{\ell\neq k\in U_v}|\E_p(\alpha_k\alpha_{\ell}|\mathcal {\hat{P}}_S)|\right]\nonumber\\
&\leqslant& \left(\frac{n_v}{N_v}\frac{(1+\lambda)^2}{\lambda^2}+\frac{C_1}{\lambda^2}\right)\frac{1}{N_v}\sum_{k\in U_v}\E_p\left(\widehat{\widetilde m}_{rf}(\bx_k)-\widetilde m_{rf}(\bx_k)\right)^2=o(1),
\end{eqnarray*}
by assumptions (H\ref{H2}), (H\ref{H3}), (H\ref{H5}) and (H\ref{H6}). It follows then that 
\begin{eqnarray}
\frac{\sqrt{n}_v}{N_v}\sum_{k\in U_v}\alpha_k(\widehat{\widetilde m}_{rf}(\bx_k)-\widetilde m_{rf}(\bx_k))=o_{\mathbb{P}}(1).\label{decomp_3}
\end{eqnarray}
Relations (\ref{decomp_1}), (\ref{decomp_terme2_rf_sample}), (\ref{decomp_2}) and (\ref{decomp_3}) give then the result.
\end{proof}
\begin{result}\label{var_est_F2}
	Consider a sequence of population RF estimators $\{\widehat{t}_{rf}\}.$  Assume also  that  $\displaystyle\frac{n^{u}_v}{n_{0v}}=O(1)$ with $1/2< u \leqslant 1.$
	Then, the variance estimator $\widehat{\V}_{rf}(\widehat t_{rf})$ is asymptotically design-consistent for the asymptotic variance $\mathbb{AV}_p \left( \widehat{t}_{rf}\right).$ That is,
	\begin{equation}
	\lim_{v \to \infty}\mathbb{E}_p \left(\dfrac{n_v}{N_v^2} \biggr \rvert\widehat{\V}_{rf}(\widehat t_{rf})- \mathbb{AV}_p (\widehat{t}_{rf})  \biggr \rvert \right) = 0. 
	\end{equation}
\end{result}
\begin{proof}
The proof follows the same steps as those of result (\ref{var_est_F1}). We need to show that 
\begin{eqnarray}
\E_p \bigg[\bigg(\widehat{m}_{rf}(\bx_k) - \widetilde{m}_{rf}(\bx_k)\bigg)^2\bigg]=o(1),\label{conv_tildem_hattildem}
\end{eqnarray}
uniformly in $k\in U_v. $ We have $\widehat{m}_{rf}(\bx_k) - \widetilde{m}_{rf}(\bx_k)=\widehat{m}_{rf}(\bx_k) - \widehat{\widetilde{m}}_{rf}(\bx_k)+\widehat{\widetilde{m}}_{rf}(\bx_k)- \widetilde{m}_{rf}(\bx_k)$ and 
\begin{eqnarray*}
\E_p(\widehat{m}_{rf}(\bx_k) - \widehat{\widetilde{m}}_{rf}(\bx_k))^2&\leqslant& \frac{1}{B}\sum_{b=1}^B\E_p(\widehat{m}^{(b)}_{tree}(\bx_k) - \widehat{\widetilde{m}}^{(b)}_{tree}(\bx_k))^2\\
&\leqslant& \frac{1}{B}\sum_{b=1}^B\E_p\left(||\hat{\mathbf{z}}^{(b)}_k||^2_2\left|\left|\widehat{\boldsymbol{\beta}}^{(b)}-\widehat{\widetilde{\boldsymbol{\beta}}}^{(b)}\right|\right|^2_2\right)\\
&\leqslant& \frac{1}{B}\sum_{b=1}^B\E_p\left(\left|\left|\widehat{\boldsymbol{\beta}}^{(b)}-\widehat{\widetilde{\boldsymbol{\beta}}}^{(b)}\right|\right|^2_2\right)\\
&\leqslant& \frac{\tilde c_3n_v}{n^2_{0v}}=o(1)
\end{eqnarray*}
by lemma \ref{borne_beta_split_sample} and provided that $\displaystyle\frac{n^{u}_v}{n_{0v}}=O(1)$ with $1/2< u<1.$ The result (\ref{conv_tildem_hattildem}) follows then by using also assumption (H\ref{H6}).
\end{proof}
\bibliographystyle{apalike}
\bibliography{biblio}